\documentclass[11pt]{article}
\usepackage{amsmath,amstext,amssymb,amsfonts}
\usepackage{color,graphicx}
\usepackage{tabularx}
\usepackage{verbatim}

\usepackage{amsthm}
\usepackage{ifdraft}

\usepackage{nicefrac}

\usepackage{microtype}

\newtheorem{theorem}{Theorem}[section]
\newtheorem*{theorem*}{Theorem}

\newtheorem{Claim}[theorem]{Claim}
\newtheorem*{claim*}{Claim}

\newtheorem{proposition}[theorem]{Proposition}
\newtheorem*{proposition*}{Proposition}
\newtheorem{lemma}[theorem]{Lemma}
\newtheorem*{lemma*}{Lemma}
\newtheorem{corollary}[theorem]{Corollary}

\newtheorem*{conjecture*}{Conjecture}

\newtheorem{fact}[theorem]{Fact}
\newtheorem*{fact*}{Fact}
\newtheorem{hypothesis}[theorem]{Hypothesis}
\newtheorem*{hypothesis*}{Hypothesis}

\theoremstyle{definition}
\newtheorem{definition}[theorem]{Definition}

\newtheorem{reduction}[theorem]{Reduction}

\newtheorem{algorithm}[theorem]{Algorithm}
\newtheorem{SDP}[theorem]{SDP}
\newtheorem{problem}[theorem]{Problem}

\newtheorem{remark}[theorem]{Remark}

\usepackage{prettyref}
\newcommand{\savehyperref}[2]{\texorpdfstring{\hyperref[#1]{#2}}{#2}}

\newrefformat{eq}{\savehyperref{#1}{\textup{(\ref*{#1})}}}
\newrefformat{lem}{\savehyperref{#1}{Lemma~\ref*{#1}}}
\newrefformat{def}{\savehyperref{#1}{Definition~\ref*{#1}}}
\newrefformat{thm}{\savehyperref{#1}{Theorem~\ref*{#1}}}
\newrefformat{cor}{\savehyperref{#1}{Corollary~\ref*{#1}}}
\newrefformat{cha}{\savehyperref{#1}{Chapter~\ref*{#1}}}
\newrefformat{sec}{\savehyperref{#1}{Section~\ref*{#1}}}
\newrefformat{app}{\savehyperref{#1}{Appendix~\ref*{#1}}}
\newrefformat{tab}{\savehyperref{#1}{Table~\ref*{#1}}}
\newrefformat{fig}{\savehyperref{#1}{Figure~\ref*{#1}}}
\newrefformat{hyp}{\savehyperref{#1}{Hypothesis~\ref*{#1}}}
\newrefformat{alg}{\savehyperref{#1}{Algorithm~\ref*{#1}}}
\newrefformat{sdp}{\savehyperref{#1}{SDP~\ref*{#1}}}
\newrefformat{rem}{\savehyperref{#1}{Remark~\ref*{#1}}}
\newrefformat{item}{\savehyperref{#1}{Item~\ref*{#1}}}
\newrefformat{step}{\savehyperref{#1}{step~\ref*{#1}}}
\newrefformat{conj}{\savehyperref{#1}{Conjecture~\ref*{#1}}}
\newrefformat{fact}{\savehyperref{#1}{Fact~\ref*{#1}}}
\newrefformat{prop}{\savehyperref{#1}{Proposition~\ref*{#1}}}
\newrefformat{claim}{\savehyperref{#1}{Claim~\ref*{#1}}}
\newrefformat{relax}{\savehyperref{#1}{Relaxation~\ref*{#1}}}
\newrefformat{red}{\savehyperref{#1}{Reduction~\ref*{#1}}}
\newrefformat{part}{\savehyperref{#1}{Part~\ref*{#1}}}
\newrefformat{prob}{\savehyperref{#1}{Problem~\ref*{#1}}}
\newrefformat{ass}{\savehyperref{#1}{Assumption~\ref*{#1}}}

\newcommand{\Sref}[1]{\hyperref[#1]{\S\ref*{#1}}}

\usepackage[varg]{txfonts}    

\renewcommand{\mathbb}{\varmathbb} 

\renewcommand{\leq}{\leqslant}
\renewcommand{\le}{\leqslant}
\renewcommand{\geq}{\geqslant}
\renewcommand{\ge}{\geqslant}

\usepackage{bm}

\usepackage{xspace}

\usepackage[pdftex,pagebackref,colorlinks,linkcolor=blue,filecolor = blue, citecolor = blue, urlcolor  = blue]{hyperref}

\usepackage{boxedminipage}

\newenvironment{mybox}
{\center \noindent\begin{boxedminipage}{1.0\linewidth}}
{\end{boxedminipage}
\noindent
}

\newcommand{\mper}{\,.}

\newcommand{\paren}[1]{\left(#1 \right )}

\newcommand{\Brac}[1]{\left[#1\right]}

\newcommand{\set}[1]{\left\{#1\right\}}

\newcommand{\abs}[1]{\left\lvert#1\right\rvert}
\newcommand{\Abs}[1]{\left\lvert#1\right\rvert}

\newcommand{\norm}[1]{\left\lVert#1\right\rVert}

\newcommand{\defeq}{\stackrel{\textup{def}}{=}}

\newcommand{\inprod}[1]{\left\langle #1\right\rangle}

\newcommand{\projsymb}{\Pi}

\newcommand{\normo}[1]{\norm{#1}_{\scriptstyle 1}}
\newcommand{\normi}[1]{\norm{#1}_{\scriptstyle \infty}}
\newcommand{\normb}[1]{\norm{#1}_{\scriptstyle \square}}

\newcommand{\Z}{{\mathbb Z}}
\newcommand{\N}{{\mathbb Z}_{\geq 0}}
\newcommand{\R}{\mathbb R}

\newcommand{\sdp}{{\sf SDP }}

\newcommand{\opt}{{\sf OPT}}
\newcommand{\OPT}{{\sf OPT}}

\newcommand{\sse}{{\sf SSE}}

\newcommand{\subjectto}{\text{subject to}}

\newcommand{\Esymb}{\mathbb{E}}
\newcommand{\Psymb}{\mathbb{P}}

\newcommand{\Isymb}{\mathbb{I}}
\DeclareMathOperator*{\E}{\Esymb}

\DeclareMathOperator*{\ProbOp}{\Psymb}

\newcommand{\Ex}[1]{\E\Brac{#1}}

\renewcommand{\Pr}[1]{\ProbOp\Brac{#1}}

\newcommand{\Ind}[1]{\Isymb\Brac{#1}}

\newcommand{\eset}{\emptyset}
\newcommand{\e}{\epsilon}

\definecolor{DSgray}{cmyk}{0,0,0,0.7}

\let\e\varepsilon

\newcommand{\cA}{\mathcal A}

\newcommand{\cL}{\mathcal L}

\newcommand{\cN}{\mathcal N}

\newcommand{\etal}{et. al.}
\newcommand{\bigO}{\mathcal{O}}
\newcommand{\bigo}[1]{\bigO\left(#1\right)}
\newcommand{\tbigO}{\tilde{\mathcal{O}}}
\newcommand{\tbigo}[1]{\tbigO\left(#1\right)}

\newcommand{\eigvec}{{\sf v}}

\newcommand{\argmax}{{\sf argmax}}
\newcommand{\argmin}{{\sf argmin}}

\newcommand{\supp}{{\sf supp}}

\newcommand{\U}{\bar{u}}
\newcommand{\V}{\bar{v}}
\newcommand{\W}{\bar{w}}

\newcommand{\yes}{\textsc{Yes}\xspace}    
\newcommand{\no}{\textsc{No}\xspace}	

\usepackage[mathscr]{euscript}
\usepackage{fullpage}

\newcommand{\SSEH}{Small-Set Expansion Hypothesis\xspace}

\newcommand{\tmix}[1]{{\sf t}^{\sf mix}_{\delta}\paren{#1}}

\newcommand{\eig}{\gamma}
\newcommand{\lh}{\eig_2}

\newcommand{\vol}{{\sf vol }}
\newcommand{\one}{\textbf{1}} 
\newcommand{\ralsymb}{\mathscr{R}}
\newcommand{\ral}[1]{\ralsymb\left(#1\right)} 
\newcommand{\mustat}{\mu^*} 

\newcommand{\p}{\projsymb}

\newcommand*\diff{\mathop{}\!\mathrm{d}}
\newcommand{\dt}{{\sf dt}}
\newcommand{\dr}{{\sf dr}}
\newcommand{\diam}{{\sf diam }}

\newcommand{\Nin}{N^{\sf in}}
\newcommand{\Nout}{N^{\sf out}}
\newcommand{\phiv}{\phi^{\sf V}}
\newcommand{\linf}{\lambda_{\infty}}
\newcommand{\mvert}{M^{\sf vert}}
\newcommand{\lapvert}{L^{\sf vert}}

\newcommand{\I}{{\bar Y_i}}
\newcommand{\J}{{\bar Y_j}}
\newcommand{\sdpval}{{\sf SDPval}}

\newcommand{\lmvalue}{ k \log k \log \log k \sqrt{\log r}}
\newcommand{\ssevalue}{ \min \set{\sqrt{r \log k}, \ \lmvalue} }
\newcommand{\kspvalue}{ \min \set{\sqrt{r \log k}, \ k^2 \log k \log \log k \sqrt{\log r}} }

\newcommand{\rmin}{r_{\min}}
\newcommand{\rmax}{r_{\max}}
\newcommand{\whp}{w.h.p.}

\title{Hypergraph Markov Operators, Eigenvalues and Approximation Algorithms}
\author{Anand Louis
\thanks{Supported by the Simons Collaboration on Algorithms and Geometry.
Part of work done while the author was a student at Georgia Tech and supported by Santosh Vempala's NSF award CCF-1217793.} 
\\ Princeton University \\ alouis@princeton.edu}

\date{}
\begin{document}
\begin{titlepage}

\maketitle

\begin{abstract}
The celebrated Cheeger's Inequality \cite{am85,a86}  
establishes a bound on the expansion of a graph via its spectrum.
This inequality is central to a 
rich spectral theory of graphs, based on studying the eigenvalues and
eigenvectors of the adjacency matrix (and other related matrices) of graphs. 
It has remained open to define a suitable spectral model for hypergraphs whose 
spectra can be used to estimate various combinatorial properties of the hypergraph.

In this paper we introduce a new hypergraph
Laplacian operator (generalizing the Laplacian matrix of graphs)
and study its spectra. 
We prove a {\em Cheeger}-type inequality for hypergraphs, relating the second smallest eigenvalue of this 
operator 
to the expansion of the hypergraph.
We bound other hypergraph expansion parameters via higher eigenvalues of this operator.
We give bounds on the diameter of the hypergraph as a function of the second smallest eigenvalue of the 
Laplacian operator. 
The Markov process underlying the Laplacian operator can be viewed as a dispersion process on the vertices of the hypergraph
that can be used to model rumour spreading in networks, brownian motion, etc., and
might be of independent interest.
We bound  the {\em Mixing-time} of this process as a function of the second smallest
eigenvalue of the Laplacian operator.
All these results are generalizations of the corresponding results for graphs.

We show that there can be no linear operator for hypergraphs whose spectra captures hypergraph expansion in a 
Cheeger-like manner. 
Our Laplacian operator 
is non-linear and thus
computing its eigenvalues exactly is intractable.
For any $k \in \N$,
we give a polynomial time algorithm to compute an approximation to the $k^{th}$ smallest eigenvalue
of the operator .
We show that this approximation factor is optimal under the \sse~ hypothesis (introduced by \cite{rs10}) 
for constant values of $k$.

We give a $\bigo{\sqrt{\log k \log r} \log \log k }$-approximation algorithm for the general sparsest cut
in hypergraphs, where $k$ is the number of ``demands'' in the instance and $r$ is the size of the largest hyperedge.

Finally, using the factor preserving reduction from vertex expansion in graphs to hypergraph expansion,
we show that all our results for hypergraphs extend to vertex expansion in graphs.

\end{abstract}

\end{titlepage}

\tableofcontents

\newpage

\section{Introduction}
There is a rich spectral theory of graphs, based on studying the eigenvalues and
eigenvectors of the adjacency matrix (and other related matrices) of graphs
\cite{am85,a86,ac88,abs10,lrtv11,lrtv12,lot12}.
We refer the reader to \cite{chung97,mt06} for a comprehensive survey on Spectral Graph Theory.
A fundamental graph parameter is its  expansion or conductance defined for a graph $G = (V,E)$ as:
\[  \phi_G \defeq \min_{S \subset V} \frac{ \Abs{E(S,\bar{S})} }{  \min \set{ \vol(S), \vol(\bar{S}) } }    \]
where by $\vol(S)$ we denote the sum of degrees of the vertices in $S$ and
$E(S,T)$ is the set of edges which have one endpoint in $S$ and one endpoint in $T$.
Cheeger's Inequality \cite{am85,a86}, a central inequality in Spectral Graph Theory, 
establishes a bound on expansion via the spectrum of the graph:
\[  \frac{\lambda_2}{2} \leq \phi_G \leq \sqrt{2 \lambda_2} \]
where $\lambda_2$ is the second smallest eigenvalue of the normalized 
Laplacian\footnote{ The normalized Laplacian matrix is defined as  $\cL_G \defeq D^{-1/2} (D- A) D^{-1/2}$ 
where $A$ is the adjacency matrix of the graph and $D$ is the diagonal matrix whose
$(i,i)^{th}$ entry is equal to the degree of vertex $i$.}
matrix of the graph.
This theorem and its many (minor) variants have played a major role in the design of
algorithms as well as in understanding the limits of computation \cite{js89,ss96,d07,arv09,abs10}.
We refer the reader to \cite{hlw06} for a comprehensive survey.

It has remained open to define a spectral model of hypergraphs, whose spectra 
can be used to estimate hypergraph parameters \`{a} la Spectral Graph Theory. 
Hypergraph expansion\footnote{See \prettyref{def:hyper-expansion} for formal definition}
 and related hypergraph partitioning problems are of 
immense practical importance, having applications in parallel and distributed computing \cite{ca99}, VLSI circuit 
design and computer architecture \cite{kaks99,gglp00}, scientific computing \cite{dbh06} and other areas.
Inspite of this, hypergraph expansion problems haven't been studied as well as their graph counterparts
(see \prettyref{sec:hyper-related} for a brief survey). 
Spectral graph partitioning
algorithms are widely used in practice for their efficiency and the high quality of solutions that they often
provide \cite{bs94,hl95}. 
Besides being of natural theoretical interest, a spectral theory of hypergraphs might also 
be relevant for practical applications.

The various spectral models for hypergraphs considered in the literature haven't been without shortcomings.
An important reason for this is that there is no canonical matrix representation 
of hypergraphs.
For an $r$-uniform hypergraph $H=(V,E)$ on the vertex set $V$ and having edge set $E \subseteq V^{ r} $,
one can define the canonical $r$-tensor form $A$ as follows.
\[ A_{(i_1, \ldots, i_r)} \defeq \begin{cases} 1 & \set{i_1, \ldots, i_r} \in E \\ 0 & \textrm{otherwise}   \end{cases} \mper \]
This tensor form and its minor variants have been explored in the literature
(see \prettyref{sec:hyper-related} for a brief survey),
but have not been understood very well.
Optimizing over tensors is NP-hard \cite{hl09}; even getting good approximations  
might be intractable \cite{bv09}.
Moreover,
the spectral properties of tensors seem to be unrelated to combinatorial properties
of hypergraphs (See \prettyref{app:hyper-tensor}).

Another way to study a hypergraph, say $H = (V,E)$, is to replace each hyperedge $e \in E$ by complete graph
or a low degree expander on the vertices of $e$ to obtain a graph $G = (V,E')$. If we let $r$ denote the size of the largest hyperedge
in $E$, then it is easy to see that the combinatorial properties  of $G$ and $H$, 
like min-cut, sparsest-cut, among others,  could be separated by a factor of $\Omega(r)$. 
Therefore, this approach will not be useful when $r$ is {\em large}.

In general, one can not hope to have a linear operator for hypergraphs whose spectra captures hypergraph expansion in 
a Cheeger-like manner. 
This is because the existence of such an operator will imply the existence of a polynomial
time algorithm obtaining a $\bigo{\sqrt{\opt}}$ bound on hypergraph expansion, but we rule this out
by giving a lower bound of $\Omega(\sqrt{\opt\, \log r})$ for computing hypergraph expansion,
where $r$ is the size of the largest hyperedge 
(\prettyref{thm:hyper-expansion-hardness-informal}).

Our main contribution is the definition of a new Markov operator for hypergraphs, obtained by generalizing the 
random-walk operator on graphs. 
Our operator is simple and does not require the hypergraph to be uniform (i.e. does not require all the hyperedges
to have the same size). 
We describe this operator in \prettyref{sec:hyper-defs}   
(See \prettyref{def:hyper-markov}).
We present our main results about this hypergraph operator in \prettyref{sec:hyper-eigs} and  \prettyref{sec:hyper-results}.
Most of our results are independent of $r$ (the size of the hyperedges), some of our bounds have a logarithmic
dependence on $r$, and none of our bounds have a polynomial dependence on $r$. 
All our bounds are generalizations of the corresponding bounds for graphs.

\subsection{Related Work}
\label{sec:hyper-related}
Freidman and Wigderson \cite{fw95} study the canonical tensors of hypergraphs. They bound the
second eigenvalue of such tensors for hypergraphs drawn randomly from various distributions 
and show their connections to randomness dispersers. 
Rodriguez \cite{r09} studies the eigenvalues of graph obtained by replacing each hyperedge 
by a clique (Note that this step incurs a loss of $\bigO(r^2)$, where $r$ is the size of the hyperedge).
Cooper and Dutle \cite{cd12} study the roots of the characteristic polynomial of hypergraphs
and relate it to its chromatic number. 
\cite{hq13,hq14} also study the canonical tensor form of the hypergraph and relate
its eigenvectors to some configured components of that hypergraph.
\cite{lm12,lm13,lm13b} relate the eigenvector corresponding to the second largest eigenvalue of the canonical
tensor to hypergraph quasi-randomness.
Chung \cite{c93} defines a notion of Laplacians for hypergraphs and studies the relationship
between its eigenvalues and a very different notion of hypergraph cuts and homologies. 
\cite{prt12,ms12,pr12,p13,kkl14,skm14} study the relation of simplicial complexes to rather different notion of Laplacian forms
and prove isoperimetric inequalities, study homologies and mixing times.
Ene and Nguyen \cite{en14} studied the hypergraph multiway partition problem (generalizing the graph multiway partition problem)
and gave a $4/3$-approximation algorithm.
Concurrent to this work, 
\cite{lm14b} gave approximation algorithms for hypergraph expansion, and more generally, hypergraph small set
expansion; they gave a $\tbigo{k\sqrt{\log n}}$-approximation algorithm and a $\tbigo{k\sqrt{\OPT\, \log r}}$
approximation bound for the problem of computing the set of vertices of size at most $\Abs{V}/k$ in a hypergraph 
$H = (V,E)$,  having the least expansion.

Bobkov, Houdr\'e and Tetali \cite{bht00} defined a Poincair\'e-type functional graph parameter called 
$\linf$ and showed that it relates to the vertex expansion of a graph in Cheeger like manner, i.e.
it satisfies $\linf/2 \leq \phiv = \bigo{\sqrt{\linf}}$ where $\phiv$ is the vertex expansion of the 
graph (see \prettyref{sec:vertex-results} for the definition of vertex expansion of a graph).
\cite{lrv13} gave a $\bigo{\sqrt{\OPT \log d}}$ approximation bound for computing the vertex
expansion in graphs having the largest vertex degree $d$.

Peres \etal \cite{pssw09} study a ``tug of war'' Laplacian operator on graphs that is 
similar to our hypergraph Markov operator and use it to prove that 
every bounded real-valued Lipschitz function $F$ on a subset $Y$ of a length space $X$ admits a 
unique absolutely minimal extension to $X$. Subsequently a variant of this operator was used to  
for analyzing the rate of convergence of local dynamics in bargaining 
networks \cite{cdp10}. 
\cite{lrtv11,lrtv12,lot12,lm14} study higher eigenvalues of graph Laplacians
and relate them to graph multi-partitioning parameters (see \prettyref{sec:higher-cheeger}).

\subsection{Notation}
\label{sec:hyper-notation}
We denote a hypergraph $H = (V,E,w)$, where $V$ is the vertex set of the hypergraph,
$E \subset 2^V \setminus \set{ \set{} } $ is the set of hyperedges and 
 and  $w : E \to \R^+$ gives the edge weights.
We define the degree of a vertex $v \in V$ as 
$d_v \defeq \sum_{e \in E: v \in e} w(e)$.
We use $n \defeq \Abs{V}$ to denote the number of vertices in the hypergraph and $m \defeq \Abs{E}$ to denote the 
number of hyperedges.
We use $\rmin \defeq \min_{e \in E} \Abs{e}$ to denote the size of the smallest hyperedge
and use $\rmax \defeq \max_{e \in E} \Abs{e}$ to denote the size of the largest hyperedge.
Since, most of our bounds will only need $\rmax$, we use $r \defeq \rmax$ for brevity. 
We say that a hypergraph is {\em regular} if all its vertices have the same degree. 
We say that a hypergraph is {\em uniform} if all its hyperedges have the same cardinality.

For $S,T \subset V$, we denote by $E(S,T)$ the set of hyperedges which have at least one vertex in $S$ 
and at least one vertex in $T$.
We use $\phi_H(\cdot)$ to denote expansion of sets in the hypergraph $H$ 
(see \prettyref{def:hyper-expansion}). We drop the subscript whenever the
hypergraph is clear from the context.

A list of edges $e_1, \ldots, e_l$ such that $e_i \cap e_{i+1} \neq \eset$ for $i \in [l-1]$
is referred as a {\em path}. The length of a path is the number of edges in it. 
We say that a path $e_1, \ldots, e_l$ connects two vertices $u,v \in V$
if $u \in e_1$ and $v \in e_l$.
We say that the hypergraph is {\em connected} if for each pair of vertices $u,v \in V$, there
exists a path connecting them. 
The {\em diameter} of a hypergraph, denoted by $\diam(H)$, is the smallest value $l \in \N$, such that
each pair of vertices $u,v \in V$ have a path of length at most $l$ connecting them.

For an $x \in R$, we define $x^+ \defeq \max \set{x,0}$ and $x^- \defeq \max \set{-x,0}$.
For a non-zero vector $u$, we define $\tilde{u} \defeq u/\norm{u}$.
We use $\one \in \R^n$ to denote the vector having $1$ in every coordinate. 
For a vector $X \in \R^n$, we define its support as the set of coordinates at which $X$ is non-zero, i.e.
$\supp(X) \defeq \set{i : X(i) \neq 0}$.
We use $\Ind{\cdot}$ to denote the indicator variable, i.e. $\Ind{x}$ is equal to $1$
if event $x$ occurs, and is equal to $0$ otherwise.
We use $\chi_S$ to denote the indicator function of the set $S \subset V$, i.e.
\[ \chi_S(v) = \begin{cases} 1 & v \in S \\ 0 & \textrm{otherwise}  \end{cases} \mper \]

We use $\mu(\cdot)$ to denote probability distributions on vertices.
We use $\mustat$ to denote the stationary distributions on vertices
(we will define the stationary distribution later, see \prettyref{sec:hyper-eigs}).
We denote the $2$-norm of a vector by $\norm{\cdot}$, and its $1$ norm by $\normo{\cdot}$.

We use $\projsymb(\cdot)$ to denote projection operators. For a subspace $S$, we denote by 
$\projsymb_S : \R^n \to \R^n $ the projection operator that maps a vector to its projection on $S$.  
We denote by $\projsymb^{\perp}_S : \R^n \to \R^n$ the projection operator that maps a vector to its projection
orthogonal to $S$.

We use $\frac{\diff^+ }{ \dt } f$ to denote the right-derivative of a function $f$, i.e.
\[  \frac{\diff^+ }{ \dt } f(a) = \lim_{x \to a^+} \frac{f(x) - f(a) }{x - a}  \mper \] 
Similarly, we use $\frac{\diff^- }{ \dt } f$ to denote the left-derivative of $f$.

\section{The Hypergraph Markov Operator}
\label{sec:hyper-defs}
We now formally define the hypergraph Markov operator $M: \R^n \to \R^n$.  
For a hypergraph $H$, we denote its Markov operator by $M_H$. We drop the subscript 
whenever the hypergraph is clear from the context.

\begin{mybox}
\begin{definition}[The Hypergraph Markov Operator]~

 Given a vector $X \in \R^n$, $M(X)$ is computed as follows.

	\begin{enumerate}
	\item 	For each hyperedge $e \in E$, let $(i_e,j_e) := {\sf argmax}_{i,j \in e} \Abs{X_i - X_j}$, breaking 
	ties randomly (See \prettyref{rem:hyper-walk-ties}). 	
	\item We now construct the weighted graph $G_X$ on the vertex set $V$ as follows.
		We add edges $\set{ \set{i_e,j_e} : e \in E }$ having weight $w(\set{i_e,j_e}) := w(e)$ to $G_X$. 
		Next, to each vertex $v$ we add self-loops of sufficient weight such that its degree in $G_X$ is equal 
		to $d_v$; more formally we add self-loops of weight  
		\[ w(\set{v,v}) := d_v - \sum_{ e \in E : v \in \set{i_e,j_e} } w(e) \mper   \]  	
	\item	We define $A_X$ to be the  random walk matrix of $G_X$, 	i.e., $A_X$ is obtained from the adjacency  
		matrix of $G_X$ by dividing the entries of the $i^{\sf th}$ row by the degree of vertex $i$ in $G_X$.
	
	\end{enumerate}
Then, 
	\[ M(X) \defeq A_X X \mper \]

\label{def:hyper-markov}
\end{definition}
\end{mybox}

\begin{remark}
We note that unlike most of spectral models for hypergraphs considered in the literature, our Markov operator $M$ 
does not require the hypergraph to be uniform (i.e. it does not require all hyperedges to have
the same number of vertices in them).
\end{remark}

\begin{remark}
\label{rem:hyper-irregular}
Let $G_X$ denote the adjacency matrix of the graph in \prettyref{def:hyper-markov}.
Then, by construction, $A_X = G_X D^{-1}$, where $D$ is the diagonal matrix whose
$(i,i)^{th}$ entry is $d_i$. 
A folklore result in linear algebra is that the matrices $G_X D^{-1}$
and $D^{-1/2}G_X D^{-1/2}$ have the same set of eigenvalues. This can be seen as follows;
let $v$ be an eigenvector of $G_X D^{-1}$ with eigenvalue $\lambda$, then
\[  D^{-1/2}G_X D^{-1/2} \paren{ D^{-1/2} v } = \paren{D^{-1/2}} \cdot \paren{ G_X D^{-1} } v =   
\paren{D^{-1/2}} \cdot \paren{\lambda\, v} =  \lambda \paren{D^{-1/2} v } \mper \]
Hence, $D^{-1/2} v$ will be an eigenvector of $D^{-1/2}G_X D^{-1/2}$ having the same eigenvalue $\lambda$.
Therefore, we will often study $D^{-1/2}G_X D^{-1/2}$ in the place of studying $G_X D^{-1}$.
\end{remark}

\begin{definition}[Hypergraph Laplacian]~
Given a hypergraph $H$, we define its Laplacian operator $L$ as 
\[ L \defeq I - M \mper \]
Here, $I$ is the identity operator and $M$ is the hypergraph Markov operator. The action of $L$ on a vector $X$ is
$ L(X) \defeq X - M(X) $.   
We define the matrix $L_X \defeq I - A_X$ (See \prettyref{rem:hyper-irregular}).
We define the {\em Rayleigh quotient} $\ral{\cdot}$ of a vector $X$ as 
\[ \ral{X} \defeq \frac{ X^T L(X)   }{X^T X} \mper \]
\label{def:hyper-laplacian}
\end{definition}

Our definition of $M$ is inspired by the $\infty$-Harmonic Functions studied by \cite{pssw09}.
We note that $M$ is a generalization of the random-walk matrix for graphs to hypergraphs; if all 
hyperedges had exactly two vertices, then $\set{i_e,j_e} = e$ for each hyperedge $e$ and 
$M$ would be the random-walk matrix.

Let us consider the special case when the hypergraph $H = (V,E,w)$ is $d$-regular.
We can also view the operator $M$ as a collection of maps $\set{f_r: \R^{r} \to \R^{r}}_{r \in \N}$
as follows. 
We define the action of $f_r$ on a tuple $(x_1,\ldots,x_r)$  as follows. 
It picks the coordinates $i, j \in [r]$ 
which have the highest and the lowest values respectively. Then it decreases the value at the
$i^{th}$ coordinate by $(x_i - x_j)/d$ and increases the value at the
$j^{th}$ coordinate by $(x_i - x_j)/d$, whereas all other coordinates
remain unchanged.
For a vector $X \in \R^n$,  the computation of $M(X)$ in \prettyref{def:hyper-markov}
can be viewed as 
simultaneously applying these maps to each edge $e \in E$, i.e. 
for each hyperedge $e \in E$,
$f_{\Abs{e}}$ is applied to the tuple corresponding to the coordinates of $X$ represented by the
vertices in $e$.

\paragraph{Comparison to other operators.}
One could ask if any other set of maps $\set{g_r: \R^{r} \to \R^{r}}_{r \in \N}$
used in this manner gives a `better'  Markov operator?
A natural set of maps that one would be tempted to try are the {\em averaging} maps
which map an $r$-tuple $(x_1, \ldots, x_r)$ to $\paren{ \sum_i x_i/r, \ldots, \sum_i x_i/r}$.

If we consider the embedding of the vertices of a hypergraph $H = (V,E,w)$ on $\R$,
given by the vector $X \in \R^{V}$,  
then the length $l(\cdot)$ of a hyperedge $e \in E$ is $l(e) \defeq \max_{i,j \in e} \Abs{X_i - X_j}$.
We believe that $l(e)$ is the most essential piece of information about the hyperedge $e$.
As a motivating example, consider the special case when all the entries of $X$ are in $\set{0,1}$.
In this case, the vector $X$ defines a cut $(S,\bar{S})$, where $S = \supp(X)$,
and the $l(e)$ indicates whether $e$ is cut by $S$ or not. 
Building on this idea, we can use the average length of edges to bound expansion of sets.
We will be studying the length of the hyperedges in the proofs of all the results in this paper. 
A well known fact from Statistical Information Theory is that moving in the direction of  
$\nabla l$ will yield the most {\em information} about the function in question.
We refer the reader to \cite{ni83,bn01} for the formal statement and  proof of this fact, and 
for a comprehensive discussion on this topic.
Our set of maps move a tuple precisely in the direction of $\nabla l$,
thereby achieving this goal.

For a hyperedge $e \in E$ the averaging maps will yield information about the function
$\E_{i,j \in e} \Abs{X_i - X_j}$ and not about $ l(e)$. 
In particular, the averaging maps will have a gap of factor $\Omega(r)$ between the hypergraph 
expansion\footnote{See \prettyref{def:hyper-expansion}.}
and the square root spectral gap\footnote{The spectral gap of a Laplacian operator is defined as its second smallest 
eigenvalue. See \prettyref{def:hyper-eigs} for the definition of eigenvalues of the Markov operator $M$. } 
of the operator. In general, if a set of maps changes $r'$ out of $r$ coordinates, it will have
a gap of $\Omega(r')$ between hypergraph expansion and the square root of the spectral gap.

Our set of maps $\set{f_r}_{r \in \N}$ are also the very natural {\em greedy} maps which bring the 
pair of coordinates which are farthest apart slightly closer to each other.
Let us consider the continuous dispersion process where we repeatedly apply the markov operator 
$\paren{ (1 -\dt)I + \dt\, M }$ ( for an infinitesimally small value of $\dt$) 
to an arbitrary starting probability distribution on the vertices (see \prettyref{def:hyper-randomwalk}).
In the case when the maximum value (resp. minimum value) in the $r$-tuple is much higher (resp. much lower) 
than the second maximum value (resp. second minimum value), then these set of greedy maps are 
essentially the best we can hope for, as they will lead to the greatest decrease in variance
of the values in the tuple. 
In the case when the maximum value (resp. minimum value) in the tuple, located at some coordinate $i_1 \in [r]$ 
is close to the second 
maximum value (resp. second minimum value), located at some coordinate $i_2 \in [r]$, the 
dispersion process is likely to decrease the value at coordinate $i_1$ till it equals the value
at coordinate $i_2$ after which these two coordinates will decrease at the same rate 
(see \prettyref{sec:hyper-walk} and \prettyref{rem:hyper-walk-ties}). 
Therefore, our set of greedy maps addresses all cases satisfactorily.

\subsection{Hypergraph Eigenvalues}
\label{sec:hyper-eigs}

\paragraph{Stationary Distribution.}
A probability distribution $\mu$ on $V$ is said to be {\em stationary} if
$ M(\mu) = \mu \mper$
We define the probability distribution $\mustat$ as follows.
\[ \mustat(i) = \frac{d_i}{\sum_{j \in V} d_j } \qquad \textrm{for } i \in V \mper  \]
$\mustat$ is a stationary distribution of $M$, as it is an eigenvector with eigenvalue $1$ of $A_X$ $\forall X \in \R^n$.

\paragraph{Laplacian Eigenvalues.}
An operator $L$ is said to have an eigenvalue $\lambda \in \R$ if for some vector $X \in \R^n$, 
$L(X) = \lambda\, X$.
It follows from the definition of $L$ that $\lambda$ is an eigenvalue of $L$ if and only if $1 - \lambda$
is an eigenvalue of $M$. 
In the case of graphs, the Laplacian Matrix and the adjacency matrix have $n$ orthogonal eigenvectors. However
for hypergraphs, the Laplacian operator $L$ (respectively $M$) is a highly non-linear operator. In general 
non-linear operators can have a lot more more than $n$ eigenvalues or a lot fewer than $n$ eigenvalues.

From the definition of stationary distribution we get that $\mustat$ is an
eigenvector of $M$ with eigenvalue $1$. Therefore, $\mustat$ is an eigenvector of $L$ 
with eigenvalue $0$.
As in the case of graphs, it is easy to see that the hypergraph Laplacian operator has only
non-negative eigenvalues. 
\begin{proposition}
Given a hypergraph $H$ and its Laplacian operator $L$,
all eigenvalues of $L$ are non-negative.
\end{proposition}
\begin{proof}
Let $\eigvec$ be an eigenvector of $L$ and let $\eig$ be the corresponding eigenvalue.
Then, from the definition of $L$ (\prettyref{def:hyper-markov}), 
 $\eigvec$ is an eigenvector of the matrix $\paren{I -  G_{\eigvec} D^{-1}}$ with eigenvalue $\eig$. 
Using \prettyref{rem:hyper-irregular}, we get $D^{-1/2} \eigvec$ is an eigenvector of the matrix
$\paren{I -  D^{-1/2} G_{\eigvec} D^{-1/2}}$ with eigenvalue $\eig$.
Therefore, 
\[  0 \leq \frac{ \paren{D^{-1/2} \eigvec}^T \paren{I -  D^{-1/2} G_{\eigvec} D^{-1/2}} \paren{D^{-1/2} \eigvec} }
		{ \paren{D^{-1/2} \eigvec}^T \paren{D^{-1/2} \eigvec} } 
	=  \frac{ \paren{D^{-1/2} \eigvec}^T \paren{ \eig \paren{D^{-1/2} \eigvec}}  }
		{ \paren{D^{-1/2} \eigvec}^T \paren{D^{-1/2} \eigvec} } = \eig  \]
where the first inequality follows from the folklore fact that the symmetric matrix 
$\paren{I -  D^{-1/2} G_{\eigvec} D^{-1/2}} \succeq 0$. Hence, the proposition follows.
\end{proof}

We start by showing that $L$ has at least one non-trivial eigenvalue. 
\begin{theorem}
\label{thm:hyper-2nd}
Given a hypergraph $H$, there exists a non-zero vector $v \in \R^n$ and a $\lambda \in \R$ such that 
$\inprod{v,\mustat} = 0$ and $L(v) = \lambda\, v $.
\end{theorem}
Given that a non-trivial eigenvector exists, we can define the second smallest eigenvalue $\eig_2$
as the smallest eigenvalue from \prettyref{thm:hyper-2nd}.
We define $\eigvec_2$ to be the corresponding eigenvector.
 
It is not clear if $L$ has any other eigenvalues. We again remind the reader that in general, 
non-linear operators can have very few eigenvalues or sometimes even have no eigenvalues at all. 
We leave as an open problem the task of investigating if other eigenvalues exist.  
We study the eigenvalues of $L$ when restricted to certain subspaces. We prove the following theorem
(see \prettyref{thm:hyper-eigs-subspace} for formal statement).
\begin{theorem}[Informal Statement]
\label{thm:hyper-eigs-subspace-informal}
Given a hypergraph $H$, for every subspace $S$ of $\ \R^n$, the operator 
$\projsymb_S L$  has an eigenvector, i.e. 
there exists a non-zero vector $\eigvec \in S$ and a $\eig \in \R$ such that 
\[ \projsymb_S L(\eigvec) = \eig\, \eigvec \mper \]
\end{theorem}

Given that $L$ restricted to any subspace has an eigenvalue, we can now define higher eigenvalues 
of $L$ \`{a} la Principal Component Analysis (PCA).
\begin{definition}
\label{def:hyper-eigs}
Given a hypergraph $H$, we define its $k^{th}$ smallest eigenvalue $\eig_k$ and the corresponding 
eigenvector $\eigvec_k$ recursively as follows.
The basis of the recursion is $\eigvec_1 = \mustat$ and $\eig_1 = 0$.
Now, let $S_k := {\sf span} \paren{ \set{\eigvec_i : i \in [k]}}$. 
We define $\eig_k$ to be the smallest non-trivial\footnote{
By non-trivial eigenvalue of $\projsymb_{S_{k-1}}^{\perp} L$, we mean vectors in $\R^n \setminus S_{k-1}$
as guaranteed by \prettyref{thm:hyper-eigs-subspace-informal}.} 
eigenvalue of $\projsymb_{S_{k-1}}^{\perp} L$ and $\eigvec_k$
to be the corresponding eigenvector.
\end{definition}

We will often use the following formulation of these eigenvalues.
\begin{proposition}
\label{prop:hyper-eigs-ral}
The eigenvalues defined in \prettyref{def:hyper-eigs} satisfy
\[  \eig_k = \min_{ X} \frac{ X^T \projsymb_{S_{k-1}}^{\perp} L(X)  }{X^T \projsymb_{S_{k-1}}^{\perp} X} 
   = \min_{X \perp \eigvec_1, \ldots, \eigvec_{k-1}} \ral{X}\mper 
\]
\[ \eigvec_k = \argmin_{ X} \frac{ X^T \projsymb_{S_{k-1}}^{\perp} L(X)  }{X^T \projsymb_{S_{k-1}}^{\perp}  X} 
= \argmin_{X \perp \eigvec_1, \ldots, \eigvec_{k-1}} \ral{X} 
  \mper \]

\end{proposition}

\subsection{Hypergraph Dispersion Processes}
A Dispersion Process on a vertex set $V$ starts with some distribution of mass on the 
vertices, and moves mass around according to some predefined rule. Usually mass 
moves from vertex having a higher concentration of mass to a vertex having a lower 
concentration of mass.
A random walk on a graph is a dispersion process, as it can be viewed as a process moving
{\em probability-mass} along the edges of the graph. 
We define the canonical dispersion process based on the hypergraph Markov operator
(\prettyref{def:hyper-randomwalk}).

\begin{mybox}
\begin{definition}[Continuous Time Hypergraph Dispersion Process]~Given a hypergraph $H = (V,E,w)$, 
a starting probability distribution $\mu^0$ on $V$, 
we (recursively) define the probability distribution on the vertices at time $t$ according to the  
following heat equation
\[  \frac{ \diff \mu^t  }{\dt} = - L(\mu^t)   \mper  \]
Equivalently, for an infinitesimal time duration $\dt$, 
the distribution at time $t + \dt$ is defined as a function of the distribution 
at time $t$ as follows
\[ \mu^{t + \dt} =  \paren{(1-\dt)I + \dt\, M } \circ \mu^t \mper \]
\label{def:hyper-randomwalk}
\end{definition}
\end{mybox}

This dispersion process can be viewed as the hypergraph analogue of the heat kernel on graphs;
indeed, when all hyperedges have cardinality $2$ (i.e. the hypergraph is a graph), the
action of the hypergraph Markov operator $M$ on a vector $X$ is equivalent to the 
action of the (normalized) adjacency matrix of the graph on $X$. 
This process can be used as an algorithm to estimate size of a hypergraph
and for sampling vertices from it, 
in the same way as random walks are used to accomplish these tasks in graphs. 
We further believe that this dispersion process will have numerous applications in counting/sampling
problems on hypergraphs, in the same way that random walks on graphs have applications in counting/sampling
problems on graphs.

A fundamental parameter associated with the dispersion processes 
is its {\em Mixing Time}.
\begin{definition}[Mixing Time]
Given a hypergraph $H = (V,E,w)$, a probability distribution $\mu$ is said to be $(1-\delta)$-mixed if 
\[ \normo{\mu - \mustat} \leq \delta \mper   \]
Given a starting probability distribution $\mu^0$, we define its {\em Mixing time} 
$\tmix{\mu^0}$ as the smallest time $t$ such that  
\[ \normo{\mu^t - \mustat } \leq \delta \]
where the $\mu^t$ are as given by the hypergraph Dispersion Process (\prettyref{def:hyper-randomwalk}). 
\end{definition}
We will show that in some hypergraphs on $2^k$ vertices, the mixing time can be $\bigo{ {\sf poly}(k)}$
(\prettyref{thm:hyperwalk-upper}).
We believe that this fact will have applications in counting/sampling problems on hypergraphs
\`{a} la MCMC (Markov chain monte carlo) algorithms on graphs.

\subsection{Summary of Results}
\label{sec:hyper-results}

Our first result is that assuming the \sse~ hypothesis, there is no linear operator (i.e. a matrix) whose 
eigenvalues can be used to estimate $\phi_H$ in a Cheeger like manner.
See \prettyref{sec:hyper-eigs-lower} for a definition of \sse~ hypothesis (\prettyref{hyp:sse}).
\begin{theorem}
\label{thm:hyper-nonlinear}
Given a hypergraph $H=(V,E,w)$, assuming the \sse~ hypothesis, there exists no
polynomial time algorithm to compute a matrix $A \in \R^{V \times V}$, such that 
\[  c_1 \lambda \leq \phi_H \leq c_2 \sqrt{\lambda}   \]
where $\lambda$ is any polynomial time computable function of the eigenvalues of $A$
and $c_1, c_2 \in \R^+$ are absolute constants.
\end{theorem}

Next, we show that our Laplacian operator $L$ has eigenvalues (see 
\prettyref{thm:hyper-2nd}, \prettyref{thm:hyper-eigs-subspace-informal} and \prettyref{prop:hyper-eigs-ral}).
We relate the hypergraph eigenvalues to other properties of hypergraphs as follows.

\subsubsection{Spectral Gap of Hypergraphs}
\begin{definition}
\label{def:hyper-expansion}
Given a hypergraph $H = (V,E,w)$, and a set $S \subset V$, we denote by $E(S, V \setminus S)$, the edges 
which have at least one end point in $S$, and at least one end point in $V \setminus S$, i.e.
\[ E(S,V \setminus S) \defeq \set{ e \in E : e \cap S \neq \eset \textrm{ and } e \cap (V \setminus S) \neq \eset } \mper   \]
We define the expansion of $S$ as
\[ \phi(S) \defeq \frac{ \sum_{e \in E(S, V \setminus S)} w(e) }{ \min \set{  \sum_{i \in S} d_i ,  \sum_{i \in \bar{S}} d_i  }} \]
and define the expansion of the hypergraph $H$ as $\phi_H \defeq \min_{S \subset V} \phi(S)$.
\end{definition}

A basic fact in spectral graph theory is that a graph is disconnected if and only if $\lambda_2$, 
the second smallest eigenvalue of its normalized Laplacian matrix, is zero. 
Cheeger's Inequality is a fundamental inequality  which can be viewed as robust version of this fact.
\begin{theorem*}[Cheeger's Inequality \cite{am85,a86}]
Given a graph $G$, let $\lambda_2$ be the second smallest eigenvalue of its normalized Laplacian matrix. Then
\[ \frac{\lambda_2}{2} \leq \phi_G \leq \sqrt{2 \lambda_2} \mper  \]
\end{theorem*}
We prove a generalization of Cheeger's Inequality to hypergraphs.
\begin{theorem}[Hypergraph Cheeger's Inequality]
\label{thm:hyper-cheeger}
Given a hypergraph $H$,
\[ \frac{\lh}{2} \leq \phi_H \leq \sqrt{2 \lh} \mper  \]
\end{theorem}

\paragraph{Hypergraph Diameter}
A well known fact about graphs is that the diameter\footnote{See \prettyref{sec:hyper-notation}
for the definition of graph and hypergraph diameter.} of a graph $G$ is at most 
$\bigo{\log n/ \paren{\log (1/(1 - \lambda_2)) }}$
where $\lambda_2$ is the second smallest eigenvalue 
of the graph Laplacian. Here we prove a generalization of this fact to hypergraphs.
\begin{theorem}
\label{thm:hyper-diam}
Given a hypergraph $H = (V,E,w)$ with all its edges having weight $1$, its diameter is at most 
\[ \diam(H) \leq \bigo{\frac{\log \Abs{V}}{\log \frac{1}{1 - \lh}}} \mper \]
\end{theorem}
We note that this bound is slightly stronger than the bound of $\bigo{\log \Abs{V}/\eig_2 }$.

\subsubsection{Higher Order Cheeger Inequalities.}
\label{sec:higher-cheeger}

Given a parameter $k \in \N$, the small set expansion problem asks to compute the set of size
at most $\Abs{V}/k$ vertices having the least expansion.
This problem arose in the context of understanding the Unique
Games Conjecture and has a close connection to it \cite{rs10,abs10}. 
In recent work, higher eigenvalues of graph Laplacians were used to bound small-set expansion in graphs.
\cite{lrtv12,lot12} show that for a graph $G$ and a parameter $k \in \N$, there exists a set 
$S \subset V$ of size $\bigo{n/k}$
such that \[ \phi(S) \leq \bigo{\sqrt{\lambda_k \log k}} \mper \]
We prove a generalization of this bound to hypergraphs
(see \prettyref{thm:hyper-sse} for formal statement).
\begin{theorem}[Informal Statement]
\label{thm:hyper-sse-informal}
Given hypergraph $H = (V,E,w)$ and parameter $k < \Abs{V}$, there exists a set $S \subset V$ 
such that $\Abs{S} \leq \bigo{ \Abs{V}/k}$ satisfying 
\[ \phi(S) \leq \bigo{ \ssevalue \sqrt{ \eig_k}  }  \]
where $r$ is the size of the largest hyperedge in $E$.
\end{theorem}

Moreover, it was shown that a graph's $\lambda_k$ 
(the $k^{th}$ smallest eigenvalue of its normalized Laplacian matrix)
is small if and only if the graph has roughly $k$ sets eaching having small expansion.
This fact can be viewed as a generalization of the Cheeger's inequality to higher 
eigenvalues and partitions.

\begin{theorem*}\cite{lot12,lrtv12}
\label{thm:graph-higher-cheeger}
For any graph $G=(V,E,w)$  and any integer $k < \Abs{V}$, there exist
 $\Theta(k)$  non-empty disjoint sets $S_1, \ldots, S_{ck} \subset V$ such that 
\[ \max_{i \in [ck]} \phi(S_i) \leq \bigo{ \sqrt{\lambda_{k} \log k } } \mper \]
Moreover, for any $k$ disjoint non-empty sets $S_1, \ldots, S_k \subset V$
\[ \max_{i \in [k]} \phi(S_i) \geq \frac{\lambda_k}{2} \mper  \]
\end{theorem*}
We prove a slightly weaker generalization to hypergraphs.
\begin{theorem}
\label{thm:hyper-higher-cheeger}
For any hypergraph $H=(V,E,w)$  and any integer $k < \Abs{V}$, there exists
 $\Theta(k)$  non-empty disjoint sets $S_1, \ldots, S_{ck} \subset V$ such that 
\[ \max_{i \in [ck]} \phi(S_i) \leq \bigo{ \kspvalue  \sqrt{\eig_k} } \mper \]
Moreover, for any $k$ disjoint non-empty sets $S_1, \ldots, S_k \subset V$
\[ \max_{i \in [k]} \phi(S_i) \geq \frac{\eig_k}{2} \mper  \]

\end{theorem}

\subsubsection{Mixing Time Bounds}
A well known fact in spectral graph theory is that a random walk on graph mixes in time at most 
$\bigo{\log n/ \lambda_2}$ where $\lambda_2$ is the second smallest eigenvalue of graph Laplacian. 
Moreover, every graph has some vertex such that a random walk starting from that vertex takes 
at least  $\Omega(1/\lambda_2)$ time to mix
(For the sake of completeness we give a proof of this fact in \prettyref{thm:graph-mixing-lower}),
thereby proving that the dependence of the mixing time
on $\lambda_2$ is optimal.
We prove a generalization of the first fact to hypergraphs and a slightly weaker 
generalization of the second fact to hypergraphs. Both of them together show that dependence
of the mixing time on $\lh$ is optimal.
Further, we believe that \prettyref{thm:hyperwalk-upper} will have applications in counting/sampling problems on hypergraphs
\`{a} la  MCMC (Markov chain monte carlo) algorithms on graphs.
\begin{theorem}[Upper bound on Mixing Time]
\label{thm:hyperwalk-upper}
Given a hypergraph $H = (V,E,w)$, for all starting probability distributions $\mu^0 : V \to [0,1]$, 
the Hypergraph Dispersion Process satisfies
\[  \tmix{\mu^0} \leq \frac{\log (n/\delta)}{ \lh} \mper \]
\end{theorem}

\begin{theorem}[Lower bound on Mixing Time]
\label{thm:hyperwalk-lower}
Given a hypergraph $H=(V,E,w)$, there exists a probability distribution $\mu^0$ on $V$
such that $\normo{\mu^0 -\mustat} \geq 1/2$
and 
\[ \tmix{\mu^0} \geq \frac{\log (1/\delta)  }{16\, \lh} \mper \]
\end{theorem}
We view the condition in \prettyref{thm:hyperwalk-lower} that the starting distribution 
$\mu^0$ satisfy $\normo{\mu^0 -\mustat} \geq 1/2$ as the analogue of a random walk in a graph starting from
some vertex.

\subsubsection{Towards Local Clustering Algorithms for Hypergraphs}
We believe that the hypergraph dispersion process (\prettyref{def:hyper-randomwalk})  
will have numerous applications in computing combinatorial properties of graphs as well 
as in sampling problems related to hypergraphs, in a manner similar to applications 
of random-walks/heat-dispersion in graphs. 
As a concrete example, we show that the hypergraph dispersion process might be useful  
towards computing sets of vertices having small expansion.
We show that if the Hypergraph dispersion process mixes slowly, then
the hypergraph must contain a set of vertices having small expansion. 
This is analogous to the corresponding fact for graphs,
and can be used as a tool to certify upper bounds on hypergraph expansion.
\begin{theorem}
\label{thm:hyper-walk-cut}
Given a hypergraph $H = (V,E,w)$ and a probability distribution $\mu^0 : V \to [0,1]$, 
let $\mu^t$ denote the probability distribution at time $t$ according to the 
hypergraph dispersion process (\prettyref{def:hyper-randomwalk}).
Then there exists a set $S \subset V$ such that $\mustat(S) \leq 1/2$ and  
\[ \phi(S) \leq  \bigo{ \min_{ t \in [0, \tmix{\mu^0}/2 ] } \sqrt{  \frac{ \log \paren{  \norm{\mu^0}^2/ \norm{ \mu^t}^2   } }{t} }} \mper      \]
Moreover, such a set can be computed in time $\tbigo{ \Abs{E}\, \tmix{\mu^0} }$.
\end{theorem}
Therefore, the hypergraph dispersion process can be used as a tool to
certify an upper bound on hypergraph expansion.
As in the case of graphs, 
this upper bound might be better than the guarantee obtained using an \sdp relaxation 
(\prettyref{cor:hyper-sparsest-informal}) in certain settings.

One could ask if the converse of the statement of \prettyref{thm:hyper-walk-cut} is true, i.e.,
if the hypergraph $H=(V,E,w)$ has a ``sparse cut'', then is there a polynomial time computable probability 
distribution $\mu^0 : V \to [0,1]$ such that the hypergraph dispersion process initialized with this
$\mu^0$ mixes ``slowly''? \prettyref{thm:hyperwalk-lower} shows that there exists such a distribution 
$\mu^0$, but it is known if such a distribution can be computed in polynomial time. We leave this as 
on open problem.

\subsubsection{Computing Eigenvalues}
Computing the eigenvalues of the hypergraph Markov operator is intractable,
as the operator is non-linear. 
We give an exponential time algorithm to compute all the eigenvalues and eigenvectors of $M$ and $L$; 
see \prettyref{thm:hyper-eigs-exp}.
We give a polynomial time $\bigo{k \log r}$-approximation algorithm to compute the $k^{th}$ smallest eigenvalue,
where $r$ is the size of the largest hyperedge. 
\begin{theorem}
\label{thm:hyper-eigs-alg}
There exists a randomized polynomial time algorithm that given a hypergraph $H = (V,E,w)$
and a parameter $k < \Abs{V}$, outputs $k$ orthonormal vectors $u_1, \ldots, u_k$ such that
\[ \ral{u_i} \leq \bigo{i \log r\ \eig_i}  \]
\whp, where $r$ is the size of the largest hyperedge.
\end{theorem}

Complimenting this upper bound, we prove a lower bound of $\log r$ for the computing the eigenvalues.
See \prettyref{sec:hyper-eigs-lower} for a definition of \sse~ hypothesis (\prettyref{hyp:sse})
and see \prettyref{thm:hyper-eigs-lower} for a formal statement of the lower bound.
\begin{theorem}[Informal Statement]
\label{thm:hyper-eigs-lower-informal}
Given a hypergraph $H$ and a parameter $k > 1$, it is \sse-hard to get better than a $\bigo{\log r}$-approximation
to $\eig_k$ in polynomial time.
\end{theorem}

\subsubsection{Approximation Algorithms for Hypergraph Partitioning}
For a hypergraph $H$, computing $\phi_H$ is a natural optimization problem in its own right.
\prettyref{thm:hyper-cheeger} gives a bound on $\phi_H$ in terms of $\eig_2$. 
Obtaining a $\bigo{\log r}$-approximation to $\eig_2$ from \prettyref{thm:hyper-eigs-alg}
gives us the following result directly. See \prettyref{cor:hyper-sparsest-formal} for a formal
statement.
\begin{corollary}[Informal Statement] 
\label{cor:hyper-sparsest-informal}
There exists a randomized polynomial time algorithm that given a hypergraph $H = (V,E,w)$,
outputs a set $S \subset V$ such that $\phi(S) = \bigo{\sqrt{\phi_H \log r}}  $
\whp, where $r$ is the size of the largest hyperedge in $E$.
\end{corollary}
We note that \prettyref{cor:hyper-sparsest-informal} also follows directly from \cite{lm14b}.

One could ask if this bound can be improved. We show that this bound is optimal (up to constant factors)
under \sse~ (see \prettyref{thm:hyper-expansion-hardness} for a formal statement of the lower bound).
\begin{theorem}[Informal Statement]
\label{thm:hyper-expansion-hardness-informal}
Given a hypergraph $H$, it is \sse-hard to get better than a $\bigo{\sqrt{\phi_H \log r}}$ bound
on hypergraph expansion in polynomial time.
\end{theorem}

Many theoretical and practical applications require multiplicative approximation
guarantees for hypergraph sparsest cut. 
In a seminal work,
Arora, Rao and Vazirani \cite{arv09} gave a $\bigo{\sqrt{\log n}}$-approximation
algorithm for the (uniform) sparsest cut problem in graphs. 
\cite{lm14b} gave a $\bigo{\sqrt{\log n }}$-approximation algorithm for hypergraph expansion.

\paragraph{Sparsest Cut with General Demands}
In an instance of the Sparsest Cut with General Demands, we are given a hypergraph $H = (V,E,w)$
and a set of demand pairs $(s_1,t_1), \ldots, (s_k,t_k) \in V \times V$ and demands $D_1, \ldots, D_k \geq 1$.
We think of the $s_i$ as {\em sources}, the $t_i$ as as {\em targets}, and the value
$D_i$ as the {\em demand} of the terminal pair $(s_i,t_i)$ for some commodity
$i$. The generalized expansion of $H$ w.r.t. $D$ is defined as 
\[ \Phi_H \defeq \min_{S \subset V} \frac{ w(E(S,\bar{S})) }{ \sum_{i = 1}^k \Abs{  \chi_S(s_i) - \chi_S(t_i)  } } \mper \]
Arora, Lee and Naor \cite{aln05}
$\bigo{\sqrt{\log k} \log \log k }$-approximation algorithm for the sparsest cut in graphs with general demands.
We give a similar bound for the sparsest cut in hypergraphs with general demands.
\begin{theorem}
\label{thm:hyper-sparsest-nonuniform}
There exists a randomized polynomial time algorithm that given 
an instance of the hypergraph Sparsest Cut problem with general demands $H=(V,E,D)$, 
outputs a set $S \subset V$ such that 
\[ \Phi(S) \leq \bigo{\sqrt{\log k \log r} \log \log k } \Phi_H  \]
\whp, where $k = \Abs{D}$ and $r = \max_{e \in E} \Abs{e} $.
\end{theorem}

\subsubsection{Vertex Expansion in Graphs and Hypergraph Expansion}
\label{sec:vertex-results}
Given a graph $G = (V,E,w)$ having maximum vertex degree $d$ and a 
set $S \subset V$, its internal boundary $\Nin(S)$, and external boundary $\Nout(S)$
is defined as follows.
\[ \Nin(S) \defeq \set{v \in S : \exists u \in \bar{S} \textrm{such that} \set{u,v} \in E } \qquad
  \Nout(S) \defeq \set{v \in \bar{S} : \exists u \in S \textrm{such that} \set{u,v} \in E } \mper  \]
The vertex expansion of this set $\phiv(S)$ is defined as 
\[ \phiv(S) \defeq \frac{ \Abs{\Nin(S)} + \Abs{\Nout(S)}  }{\Abs{S}} \mper \]
Vertex expansion is a fundamental graph parameter that has 
has applications both as an algorithmic primitive and as tool to proving communication lower bounds
\cite{lt80,l80,bt84,ak95,sm00}.

There is a well known reduction from vertex expansion in graphs to hypergraph expansion.
\begin{mybox}
\begin{reduction}~
\label{red:hyper-vert}

{\sf Input}: Graph $G=(V,E)$ having maximum degree $d$.

We construct hypergraph $H = (V,E')$ as follows. For every vertex $v \in V$, we add the
hyperedge $\set{v} \cup \Nout(\set{v})$ to $E'$.
\end{reduction}
\end{mybox}

\begin{theorem}
\label{thm:hyper-vert-exp}
Given a graph $G=(V,E,w)$ of maximum degree $d$ and minimum degree $c_1 d$ (for some constant $c_1$), 
the hypergraph $H = (V,E')$ obtained from \prettyref{red:hyper-vert}
has hyperedges of cardinality at most $d+1$ and,
\[ c_1 \phi_H(S) \leq  \frac{1}{d} \cdot \phiv(S) \leq \phi_H(S)  \qquad \forall S \subset V \mper \]
\end{theorem}
We refer the reader to \cite{lm14b} for a proof of this theorem.

\begin{remark}
The dependence on the degree in \prettyref{thm:hyper-vert-exp} is only because vertex expansion and hypergraph
expansion are normalized differently : the vertex expansion of a set $S$ is defined as the number of vertices in the boundary 
of $S$ divided by the cardinality of $S$, whereas the hypergraph expansion of a set $S$ is defined as the number 
hyperedges crossing $S$ divided by the sum of the degrees of the vertices in $S$. 
\end{remark}

We define a Markov operator $\mvert$ on graphs similar to the hypergraph Markov operator 
(see \prettyref{def:vert-markov} for formal statement).
Using this Markov operator on graphs, the analogs of all our results for hypergraphs 
can be proven for vertex expansion in graphs.
More formally, we have a Markov operator $\mvert$ and a Laplacian operator $\lapvert \defeq I - \mvert$, 
whose eigenvalues  satisfy the vertex expansion (in graphs) analogs of 
\prettyref{thm:hyper-cheeger}\footnote{A Cheeger-type Inequality for vertex expansion in graphs
was also proven by \cite{bht00}.},  
\prettyref{thm:hyper-diam}, 
\prettyref{thm:hyper-sse-informal}, \prettyref{thm:hyper-higher-cheeger},
\prettyref{thm:hyperwalk-upper}, \prettyref{thm:hyperwalk-lower}, 
\prettyref{thm:hyper-eigs-alg}, 
and \prettyref{thm:hyper-sparsest-nonuniform}.

Bobkov \etal~ \cite{bht00} defined a Poincair\'e-type  functional graph parameter called $\linf$ and
related it to vertex expansion in a {\em Cheeger-like} manner (see \prettyref{sec:vert-exp} for details).
We show that $\linf$ coincides with the second smallest eigenvalue of $\lapvert$.
\begin{theorem}
\label{thm:linf-eig}
For a graph $G$, $\linf$ is the second smallest eigenvalue of $\lapvert \defeq I - \mvert$.
\end{theorem}

\subsubsection{Discussion}

We stress that none of our bounds have a polynomial dependence on $r$, the size of the largest hyperedge
(\prettyref{thm:hyper-sse-informal} has a dependence on $\tbigo{\min \set{r,k}}$) . 
In many of the practical
applications, the typical instances have $r = \Theta(n^{\alpha})$ for some $\alpha = \Omega(1)$;
in such cases have bounds of ${\sf poly}(r)$ would not be of any practical utility.

We also stress that all our results generalize the corresponding results for graphs.

\subsection{Organization}
We begin with an overview of the proofs in \prettyref{sec:hyper-overview}.
We prove the existance on hypergraph eigenvalues (\prettyref{thm:hyper-2nd}, \prettyref{thm:hyper-eigs-subspace-informal}, 
formally \prettyref{thm:hyper-eigs-subspace}, and  \prettyref{prop:hyper-eigs-ral}) in \prettyref{sec:hyper-walk}.
We prove \prettyref{thm:hyper-walk-cut} in \prettyref{sec:hyper-walk}.
We prove the hypergraph Cheeger Inequality (\prettyref{thm:hyper-cheeger}), 
and bound on the hypergraph diameter (\prettyref{thm:hyper-diam}) in \prettyref{sec:hyper-cheeger}.
We study the higher order Cheeger inequalities  (\prettyref{thm:hyper-higher-cheeger} and \prettyref{thm:hyper-sse-informal}) 
in \prettyref{sec:hyper-higher-cheeger}. 
We prove our bounds on the mixing time (\prettyref{thm:hyperwalk-upper} and \prettyref{thm:hyperwalk-lower})
in \prettyref{sec:hyper-walk}.
We give an exponential time algorithm for computing our hypergraph eigenvalues 
(\prettyref{thm:hyper-eigs-exp}) in \prettyref{sec:hyper-eigs-exp}.
We give our approximation algorithm for computing hypergraph eigenvalues (\prettyref{thm:hyper-eigs-alg}) 
in \prettyref{sec:hyper-eigs-poly-alg}.
We prove our hardness results for computing hypergraph eigenvalues (\prettyref{thm:hyper-eigs-lower-informal}) 
and for hypergraph expansion ( \prettyref{thm:hyper-expansion-hardness-informal}),
and that no linear hypergraph operator exists (\prettyref{thm:hyper-nonlinear})  in \prettyref{sec:hyper-eigs-lower}.
We present our algorithm for hypergraph expansion (\prettyref{cor:hyper-sparsest-informal}, formally 
\prettyref{cor:hyper-sparsest-formal}) in \prettyref{sec:hyper-sparsest}, 
and we present our algorithm for sparsest cut with general demands 
(\prettyref{thm:hyper-sparsest-nonuniform}) in \prettyref{sec:hyper-sparsestcut}.

\section{Overview of Proofs}
\label{sec:hyper-overview}
\paragraph{Hypergraph Eigenvalues.} To prove that hypergraph eigenvalues exist
(\prettyref{thm:hyper-eigs-subspace-informal} and  \prettyref{prop:hyper-eigs-ral}),
we study the hypergraph dispersion process in
a more general setting (\prettyref{def:hyper-projected-randomwalk}).
We start the dispersion process with an arbitrary vector $\mu^0 \in \R^n$.
Our main tool here is to show that the Rayleigh quotient (as a function of the time) 
monotonically decreases with time. More formally, we show that the Rayleigh quotient
of $\mu^{t+ \dt}$, the vector at time $t + \dt$ (for some infinitesimally small $\dt$), is not larger than the Rayleigh
quotient of $\mu^t$, the vector at time $t$. If the under lying matrix $A_{\mu^{\sf t}}$
did not change between times $t$ and $t+\dt$, then this fact can be shown using simple 
linear algebra. 
If the under lying matrix $A_{\mu^{\sf t}}$ changes between $t$ and $t + \dt$, then proof
requires a lot more work. 
Our proof involves studying the limits of the Rayleigh quotient 
in the neighborhoods of the time instants at which the support matrix changes,  
and exploiting the continuity properties of the process. 

To show that eigenvectors exist, we start with a candidate eigenvector, say $X$, 
that satisfies the conditions of \prettyref{prop:hyper-eigs-ral}. 
We study a slight variant of hypergraph dispersion process starting with this vector $X$.
We use the monotonicity of the Rayleigh quotient to conclude that $\forall t \geq 0$, 
the vector at time $t$ of this process, say $X^t$, also satisfies the conditions of  
\prettyref{prop:hyper-eigs-ral}. Then we use the fact that the 
number of possible support matrices $\Abs{ \set{  A_Y : Y \in \R^n }} < \infty$
to argue that there exists a time interval of positive Lebesgue measure 
during which the support matrix does not change. 
We use this to conclude
 that the vectors $X^t$ during that time interval must also not change 
(the proof of this uses the previous conclusion that all $X^t$ the conditions of \prettyref{prop:hyper-eigs-ral})
and hence must be an eigenvector.

\paragraph{Mixing Time Bounds.}
To prove a lower bound on the mixing time of the Hypergraph Dispersion process
(\prettyref{thm:hyperwalk-lower}), we need to exhibit a probability distribution that 
is far from being mixed and takes a long time to mix. 
To show that a distribution $\mu$ takes a long time to mix, it would suffice to show that
$\mu - \mustat$ is ``close'' to $\eigvec_2$, as we can then use our previous assertion
about the monotonicity of the Rayleigh quotient to prove a lower bound on the mixing time.
As a first attempt at constructing such a distribution, one might be tempted to 
consider the vector $\mustat + \eigvec_2$. But this
vector might not even be a probability distribution if $\eigvec_2(i) < - \mustat(i) $ for some coordinate $i$.
A simple fix for this would to consider the vector $\mu \defeq \mustat + \eigvec_2/(n\, \normi{\eigvec_2})$.
But then $\normo{\mu - \mustat} = \normo{\eigvec_2/(n\, \normi{\eigvec_2})}$ which could be very small
depending on $\normi{\eigvec_2}$. 
Our proof involves starting with $\eigvec_2$ and carefully ``chopping'' of the vector at some points 
to control its infinity-norm while maintaining that its Rayleigh quotient is still $\bigo{\eig_2}$.
We show that this suffices to prove the desired lowerbound on the mixing time.

The main idea used in proving the upper bound on the mixing time of 
(\prettyref{thm:hyperwalk-upper}) is that the support matrix at any time $t$ has a spectral 
gap of at least $\eig_2$. Therefore, after every unit of the time, the component of the vector
$\mu^t$ that is orthogonal to $\mustat$, decreases in $\ell_2$-norm by a factor of at least $1 - \eig_2$
(irrespective of the fact that the support matrix might be changing infinitely many times 
during that time interval). 

\paragraph{Hypergraph Diameter.}
Our proof strategy for \prettyref{thm:hyper-diam} is as follows. 
Let $M' \defeq I/2 + M/2$ be a lazy version of $M$.
Fix some vertex $u \in V$. Consider the vector  
$M'(\chi_u)$. This vector will have non-zero values at exactly those coordinates 
which correspond to vertices that are at a distance of at most $1$ from $u$. Building on this
idea, it follows that the vector $M'^t (\chi_u)$ will have non-zero values at exactly those 
coordinates which correspond to vertices that are at a distance of at most $t$ from $u$. 
Therefore, the diameter of $H$ is the smallest value $t \in \N$ such that the vectors
$\set{ M'^t(\chi_u) : u \in V }$ have non-zero entries in all coordinates. We will 
upper bound the value of such a $t$. 
The key insight in this step is that the support matrix $A_X$ of any vector 
$X \in \R^n$ has a spectral gap of at least $\eig_2$, irrespective of what the vector $X$ is.

\paragraph{Hypergraph Cheeger's Inequality.} We appeal to the formulation of 
eigenvalues in \prettyref{prop:hyper-eigs-ral} to prove \prettyref{thm:hyper-cheeger}.  
\[  \eig_2 = \min_{X \perp \one}  \frac{X^T L (X) }{X^T X} = 
\frac{ \sum_{e \in E} w(e) \max_{i,j \in E} (X_i - X_j)^2   }{ d \sum_i X_i^2 }  \mper  \]
First, observe that if all the entries of the vector $X$ were in $\set{0,1}$, then
the support of this vector $X$, say $S$, will have expansion equal to $\ral{X}$. 
Building on this idea, we start with the vector $\eigvec_2$, and it use to construct
a line-embedding of the vertices of the hypergraph, such that the average 
``distortion'' of the hyperedges is at most $\bigo{\sqrt{\eig_2}}$.
Next, we represent this average distortion as an average over cuts in the hypergraph
and conclude that at least one of these cuts must have expansion at most this average value.  
Overall, we follow the strategy of proving Cheeger's Inequality for graphs. However,
we need some new ideas to handle hyperedges.

\paragraph{Higher Order Cheeger's Inequalities.}
Proving our bound for hypergraph small-set expansion (\prettyref{thm:hyper-sse-informal})
requires a lot more work.
We start with the spectral embeddings, the canonical embedding of the vertex set into $\R^k$ given  by the top
$k$ eigenvectors.
As a first attempt, one might try to ``round'' this embedding using the rounding algorithms
for small set expansion on graphs, namely the algorithms 
of \cite{bfk11} or \cite{rst10} or \cite{lm14b}. However, the rounding algorithm of \cite{bfk11}
uses the fact that the vectors should satisfy $\ell_2^2$-triangle inequality.
More crucially the algorithms of \cite{bfk11} and \cite{lm14b}
 use the fact that the inner product between any two vectors is 
 non-negative. Neither of these properties are satisfied by the spectral embedding\footnote
{If the $v_i$'s are the spectral embedding vectors, then one could also try to round the
vectors $v_i \otimes v_i$.  This will have the property 
$\inprod{v_i \otimes v_i, v_j \otimes v_j } \geq 0$. However, by rounding these vectors
one can only hope to prove a $\bigo{\sqrt{ \eig_{k^2} {\sf polylog}\, k}}$ (see \cite{lrtv11}).}.
The rounding algorithm of \cite{rst10} crucially uses the fact that the Rayleigh
quotient of the vector $X_l$ obtained by picking the $l^{th}$ coordinate from each vector  
of the spectral embedding be ``small'' for at least one coordinate $l$. It is easy to show 
that this fact holds for graphs, but this is not true for hypergraphs because of the 
``$\max$'' in the definition of the eigenvalues.

Our proof starts with the spectral embedding and uses a simple random sampling algorithm 
due to \cite{lm14b} to sample a set of vectors, say $S$, whose corresponding unit vectors are ``close'' togethor.
We use this set to construct a line-embedding of the hypergraph, where 
a vertex $u \in S$ is mapped to the value equal to the length of the vector corresponding to $u$ in the
spectral embedding, and vertices not in $S$ get mapped to $0$. 
This step is similar to the rounding algorithm of \cite{lm14}, who studied a variant of small-set expansion
in graphs.
We then bound the length\footnote{Length of an edge $e$ under $X$ is defined
as $\max_{i,j \in e} \Abs{X_i - X_j}$.} of the hyperedges under this line-embedding. 
We handle the hyperedges whose vertices have roughly equal lengths by bounding the
probability of them being split in the random sampling step, in a manner similar to \cite{lm14b} . 
We handle the hyperedges
whose vertices have very large disparity in lengths by showing that they must be
having a large contribution to the Rayleigh quotient (in other words, such hyperedges are ``already paid for'').   
This suffices to bound the expansion of the set obtained by our rounding algorithm 
(\prettyref{alg:hyper-sse}). To show that the set is small, 
we use a combination of the techniques studied in \cite{lrtv12} and \cite{bfk11}.
This gives uses the desired bound for small-set expansion.
To get a bound on hypergraph multi-partitioning (\prettyref{thm:hyper-higher-cheeger}),
at a high level, we use a stronger form of our hypergraph small-set expansion bound together 
 with the framework of \cite{lm14}.

\paragraph{Computing Eigenvalues.} We show that the exact computation of the eigenvalues 
of our Laplacian operator is intractable (\prettyref{thm:hyper-eigs-lower-informal}).
\cite{lrv13} showed a lower bound of $\Omega(\sqrt{\opt \log d})$ for the computation of 
vertex expansion on graphs of maximum degree $d$. 
The reduction from vertex expansion to hypergraph expansion (\prettyref{thm:hyper-vert-exp})
implies a lower bound of $\Omega(\sqrt{\opt \log r})$ for the computation of 
hypergraph expansion of hypergraphs having hyperedges of cardinality at most $r$.
This immediately implies that one can not get better than an $\Omega(\log r)$
approximation to the eigenvalues of $L$ in polynomial time, as any $o(\log r)$-approximation
for the eigenvalues of $L$ will imply a $o(\sqrt{\opt \log r})$ bound for hypergraph expansion
via the Hypergraph Cheeger's Inequality (\prettyref{thm:hyper-cheeger}).
Building on this, we can show that there is no linear operator whose spectra 
captures hypergraph expansion in a  Cheeger-like manner.

We give a $\bigo{k \log r}$-approximation algorithm for $\eig_k$ (\prettyref{thm:hyper-eigs-alg}).
Our algorithm proceeds inductively.
We assume that we have computed $k-1$ orthonormal vectors $u_1, \ldots, u_{k-1}$ 
such that $ \ral{u_i} \leq \bigo{i  \log r\, \eig_i}$, and show how to compute an approximation to $\eig_k$.
Our main idea is to show that there exists a unit vector
$X \in {\sf span}\set{\eigvec_1, \ldots, \eigvec_{k}}$ which is orthogonal to 
${\sf span} \set{u_1, \ldots, u_{k-1}}$ and has small Rayleigh quotient.
Note that unlike the case of matrices, for an 
$X \in {\sf span} \set{\eigvec_1, \ldots, \eigvec_k}$, we can not bound 
$X^T L(X)$ by $\max_{i \in [k]} \eigvec_i^T L( \eigvec_i)$. The operator $L$ 
is non-linear, and there is no reason to believe that something like the 
celebrated {\em Courant-Fischer Theorem}  for matrices holds for this operator.
In general, for an $X \in {\sf span} \set{\eigvec_1, \ldots, \eigvec_k}$,
the Rayleigh quotient can be much larger than $\eig_k$.
We will show that for such an $X$, $\ral{X} \leq k\, \eig_k$.
However, we still do not have a way to compute such a vector $X$.
We given an $\sdp$ relaxation and a rounding
algorithm to compute an ``approximate'' $X$.

\paragraph{Sparsest Cut with General Demands.}
To prove \prettyref{thm:hyper-sparsest-nonuniform},
we start with a suitable $\sdp$ relaxation
together with $\ell_2^2$-triangle inequality constraints. 
Let the $\sdp$ vectors be denoted by $\set{\U}_{u \in V}$.
Arora, Lee and Naor \cite{aln05} gave a way to embed any 
$n$ point negative-type metric space into $\ell_1$ while incurring a distortion of at most 
$\bigo{\sqrt{\log n} \log \log n}$.
We use this construction to embed the $\sdp$ vectors into $\ell_1$.
Let us denote these $\ell_1$ vectors by $\set{f(u) }_{u \in V}$.
Till this point, this proof is the same as the corresponding proof for 
sparsest cut with general demands in graphs. 

Picking a certain coordinate, say $i$, gives an embedding of the vertices onto the line
where vertex $u \mapsto f(u)(i)$. From each such line embedding we can recover a
 cut having expansion proportional the average distortion of edges under this line embedding.
In the case of graphs, we can proceed by enumerating over all line embeddings obtained from the
coordinates of $\set{f(u)}_{u \in V}$, and outputting the best cut. This cut can be shown 
to be an $\bigo{\sqrt{\log k} \log \log k }$-approximation.

However, this approach will not work in the case of hypergraphs because of the more complicated
objective function for the $\sdp$ relaxation of sparsest cut in hypergraphs. Therefore, we must proceed differently. We show that a simple  
random projection of the $\set{f(u)}_{u \in V}$ vectors does not increase the ``length'' of the
edges by too much, while still keeping the vectors spread out on average. 
We use this to obtain a $\bigo{  \sqrt{\log r} \cdot \sqrt{\log k} \log \log k}$-approximation to 
$\Phi_H$.

\section{The Hypergraph Dispersion Process}
\label{sec:hyper-walk}

In this section we will prove \prettyref{thm:hyper-2nd},
\prettyref{thm:hyper-eigs-subspace-informal}, \prettyref{prop:hyper-eigs-ral},
\prettyref{thm:hyperwalk-upper} and \prettyref{thm:hyperwalk-lower}.
For the sake of simplicity, we assume that the hypergraph is regular. All our proofs easily 
extend to the general case. 

\begin{mybox}
\begin{definition}[Projected Continuous Time Hypergraph Dispersion Process]~
\label{def:hyper-projected-randomwalk}
Given a hypergraph $H = (V,E,w)$, 
 a projection operator $\p_S : \R^n \to \R^n$ for some subspace $S$ of $\R^n$
and a function $\omega^0 : V \to \R$ such that $\omega^0 \in S$, 
we (recursively) define the functions on the 
vertices at time $t$ according to the  following heat equation
\[  \frac{\diff \omega^t }{\dt} = - \p_S L (\omega^t) \]
Equivalently, for an infinitesimal time duration $\dt$, 
the function at time $t + \dt$ is defined as 
\[ \omega^{t + \dt} \defeq  \p_S \, \paren{(1-\dt)I + \dt\, M } \circ \omega^t \mper \]
\end{definition}
\end{mybox}

\begin{remark}
\label{rem:hyper-walk-ties}
We make a remark about the matrices $A_X$ for vectors $X \in \R^n$ 
in \prettyref{def:hyper-markov} when being used in the
continuous time processes of \prettyref{def:hyper-randomwalk} and 
\prettyref{def:hyper-projected-randomwalk}. For a hyperedge $e\in E$, 
we compute the pair of vertices 
\[ (i_e,j_e) = \argmax_{i,j \in e} \paren{X_i - X_j} \]
and add an edge between them in the graph $G_X$. 
If the pair is not unique, then we define 
\[ S^t_e \defeq \set{i \in e : \omega^t(i) = \max_{j \in e} \omega^t(j)} \qquad \textrm{and} \qquad 
   R^t_e \defeq \set{i \in e : \omega^t(i) = \min_{j \in e} \omega^t(j)} \]
and add to $G_X$ a complete weighted bipartite graph 
on $S^t_e \times R^t_e$ with each edge having weight $w(e)/\paren{\Abs{S^t} \Abs{R^t}} $.

A natural thing one would try first is to pick a vertex, say $i_1$, from $S^t_e$ and a vertex, say $j_1$,
from $R^t_e$ and add an edge between $\set{i_1,j_1}$. However, in such a case, after $1$ infinitesimal time unit,
the pair $(i_1,j_1)$ will no longer have the largest difference in values of $X$ among the pairs
in $e \times e$, and we will need to pick some other suitable pair from $S^t_e \times R^t_e \setminus \set{(i_1,j_1)}$. 
We will have to repeat this process of picking a different pair of vertices after each infinitesimal time unit.
Moreover, each of these infinitesimal time units will have Lebesgue measure $0$.
Therefore, we avoid this difficulty by adding a suitably weighted complete graph on
$S^t_e \times R^t_e$ without loss of generality.
 
\end{remark}

Note that when $\p_S = I$, then \prettyref{def:hyper-projected-randomwalk} is the same as 
\prettyref{def:hyper-randomwalk}. We need to study the Dispersion
Process in this generality to prove \prettyref{thm:hyper-eigs-subspace-informal}
and \prettyref{prop:hyper-eigs-ral}.

\begin{lemma}[Main Technical Lemma]
\label{lem:hyper-walk-prop}
Given a hypergraph $H = (V,E,w)$, and a function $\omega^0 : V \to \R$,  
the Dispersion process in \prettyref{def:hyper-projected-randomwalk}
satisfies the following properties.
\begin{enumerate}
%
\item
\label{eq:hyper-walk-prop-3}
\begin{equation}
\frac{ \diff \norm{\omega^t}^2 }{ \dt } =  - 2\, \ral{\omega^t} \norm{\omega^t}^2  \qquad \forall t \geq 0 \mper
\end{equation}
\item For any $t \geq 0$\footnote{See \prettyref{sec:hyper-notation} for definition of $\frac{\diff^+ }{ \dt } f$. } 
\label{eq:hyper-walk-prop-2}
\begin{equation}
\frac{\diff^+ \ral{\omega^t}}{ \dt } \leq  0 \qquad \textrm{and} \qquad \frac{\diff^- \ral{\omega^t}}{ \dt } \leq  0   \mper
\end{equation}
\end{enumerate}
\end{lemma}

\begin{proof}
Fix a time $t \geq 0$.  
\begin{enumerate}
\item
Let 
\[ A \defeq A_{\omega^{{\sf t}}} \qquad \textrm{and} \qquad   A' = (1 - \dt)I + \dt\, A \mper \]
Then, for $\dt \to 0$, 
\[ \norm{\omega^t}^2  -  \norm{\omega^{t+\dt}}^2 = 
\inprod{\omega^t - \omega^{t+\dt}, \omega^t + \omega^{t+\dt}}  = (\omega^t)^T( I - \p_S A' )(I+ \p_S A') \omega^t \]
Now, $\lim_{\dt \to 0} (I+ \p_S A') = I + \p_S$.
By construction, we have $\omega^t \in S$.
Therefore, 
\[ \norm{\omega^t}^2  -  \norm{\omega^{t+\dt}}^2 = 2 \dt\, (\omega^t)^T( I - A) \omega^t \mper \]
Therefore
\[ \frac{\diff  \norm{\omega^t}^2}{\dt } = - 2\, \ral{\omega^t}  \norm{\omega^t}^2 \mper \]


\item
We will show that 
\[ \frac{\diff^+ \ral{\omega^t}}{ \dt } \leq  0 \mper \] 
The proof of $\frac{\diff^- \ral{\omega^t}}{ \dt } \leq  0 $ can be done similarly.

Fix a time $t \in \R^+$.
Note that to show $\frac{\diff^+ \ral{\omega^t}}{ \dt } \leq  0$, it suffices to show that
for some infinitesimally small interval $[0,\dt]$
\[ \ral{\omega^{t + t'}} \leq  \ral{\omega^t} \qquad \forall t' \in [0, \dt] \mper \]

First, let us consider the case when 
there exists a time interval of positive Lebesgue measure $[0,a)$ such that
\[ A_{\omega^{{\sf t}  }} =  A_{\omega^{{\sf t + t'}  }}  \qquad \forall t' \in  [0, a) \mper \]
Let 
\begin{equation} 
\label{eq:A1def}
A_1 \defeq  A_{\omega^{{\sf t}}} \qquad A_1' \defeq (1 - \dt)I + \dt\, A_1 \qquad L_1 \defeq I - A_1    \mper
\end{equation}
In such a case, the heat equation in \prettyref{def:hyper-projected-randomwalk} is time-invariant in the interval $[t, t+a)$, and hence
can be solved using folklore methods (see \cite{chung97,mt06,lpw09} for a comprehensive discussion) to give
\[ \omega^{t + t'} =  e^{- t' \p_S L_1 }  \omega^t  \qquad  \forall  t' \in [0,a) \mper \]
Then,
\begin{align*}
 \frac{\diff^+ \ral{\omega^t}}{ \dt } & = 
\lim_{\dt \to 0} \frac{ \ral{\omega^{t+\dt}} - \ral{\omega^t} }{\dt} 
  \\ & =  \lim_{\dt \to 0} \frac{1}{\dt} \paren{ \frac{ \paren{e^{- \dt\, \p_s L_1 } \omega^t}^T  (I - A_1) \paren{e^{- \dt\, \p_s L_1 } \omega^t }}
		{\paren{e^{- \dt\, \p_s L_1 } \omega^t}^T \paren{ e^{- \dt\, \p_s L_1 } \omega^t } }
  - \frac{  (\omega^t)^T (I - A_1) \omega^t  }{ (\omega^t)^T \omega^t  }}   \mper 
\end{align*}
Using the matrix exponential expansion
\[ e^{ -t L } =  \sum_{n = 0}^{\infty} \frac{(-t\, L)^n}{n!}   \]
and using $\dt \to 0$, we get
\[
\frac{\diff^+ \ral{\omega^t}}{ \dt }  =  
\lim_{\dt \to 0} \frac{1}{\dt} \paren{
	\frac{ \paren{(I - \dt\, \p_S L_1) \omega^t}^T  (I - A_1) \paren{(I - \dt\, \p_S L_1) \omega^t}  }
	{ \paren{(I - \dt\, \p_S L_1) \omega^t}^T \paren{(I - \dt\, \p_S L_1) \omega^t} }
		- \frac{  (\omega^t)^T (I - A_1) \omega^t  }{ (\omega^t)^T \omega^t  }} 
\]
Next, using $(I - A_1) \dt = I - A_1'$, and that $\omega^t \in S$ we get
\begin{align*}  
\frac{\diff^+ \ral{\omega^t}}{ \dt }    
 & =  \lim_{\dt \to 0} \frac{1}{\dt^2} \paren{
	\frac{ (\omega^t)^T  A_1'\p_S (I - A_1') \p_S A_1' \omega^t   }{ (\omega^t)^T A_1'\p_S \p_S A_1' \omega^t}
		- \frac{  (\omega^t)^T (I - A_1') \omega^t  }{ (\omega^t)^T \omega^t  }} \\ 
& = \lim_{\dt \to 0} \frac{1}{\dt^2} \paren{
	\frac{ (\omega^t)^T \paren{ \p_s A_1'\p_S} \paren{I - \paren{ \p_s A_1'\p_S} } \paren{ \p_s A_1'\p_S} \omega^t   }
		{ (\omega^t)^T \paren{ \p_s A_1'\p_S}\paren{ \p_s A_1'\p_S} \omega^t}
		- \frac{  (\omega^t)^T \paren{I - \paren{ \p_s A_1'\p_S} } \omega^t  }{ (\omega^t)^T \omega^t  }} & \\ 
	& \qquad \qquad \qquad  \ \ \textrm{(Using $\omega^t \in S$)} \\ 
	& \leq 0  \qquad \qquad \quad \textrm{(\prettyref{prop:matrix-ineq-1})} \mper  
\end{align*}

Next, we consider the case when 
\[ A_{\omega^t} \neq A_{\omega^{t+t'}} \qquad \forall t' \in  (0,a]  \]
for some sufficiently small interval $(0,a]$ of positive Lebesgue measure. 
Let us also choose this interval $(0,a]$ such that,
\[  A_{\omega^{t + t_1}} = A_{\omega^{t + t_2}}  \qquad \forall t_1, t_2 \in (0,a] \mper   \]
This can be done without loss of generality.
Then, using the argument in the previous case, we get that
\[ \ral{\omega^{t+a}} \leq  \ral{\omega^{t + t_1}}  \qquad \forall t_1 \in (0,a] \mper  \]
This implies that
\begin{equation}
\label{eq:ral-mon-2}
\ral{\omega^{t+a}} \leq \lim_{\alpha \to 0} \ral{\omega^{t + \alpha}} \mper
\end{equation}
Therefore, to finish the proof, we need only show that 
\[ \lim_{\alpha \to 0} \ral{\omega^{t + \alpha}} = \ral{\omega^t} \mper   \]

Recall from \prettyref{rem:hyper-walk-ties} that
\[ S_e^t \defeq \set{i \in e : \omega^t(i) = \max_{j \in e} \omega^t(j)} \qquad \textrm{and} \qquad 
   R_e^t \defeq \set{i \in e : \omega^t(i) = \min_{j \in e} \omega^t(j)} \mper \]
The contribution of  $e$ to the numerator of $\ral{\omega^t}$ is
\[ f_e(t) \defeq \frac{w(e)}{\Abs{S_e^t} \Abs{R_e^t}} \sum_{i \in S_e^t, j \in R_e^t} \paren{ \omega^t(i) - \omega^t(j) }^2 \mper  \]
We make the following claim.
\begin{Claim}
\label{claim:hyper-ral-cont}
$f_e(t)$ is a continuous function of the time $t$ $\forall t \geq 0$. 
\end{Claim}
\begin{proof}
This follows from the definition of process.
The projection operator $\p_S$, being a linear operator, is continuous. Being a projection operator, 
it has operator norm at most $1$.
For a fixed edge $e$, and vertex $v \in e$, the rate of change of mass at $v$ due to edge $e$
is at most $ \omega^t(v)/d$ (from \prettyref{def:hyper-projected-randomwalk}). 
Since, $v$ belongs to at most $d$ edges, the total rate of change of mass
at $v$ is at most $\omega^t(v)$.

Since, there are at most $r$ vertices in $e$, we get that 
for any time $t$ and for every $\e > 0$, 
\[ \Abs{ f_e(t + \alpha) - f_e(t) } \leq \e  \qquad \forall \Abs{\alpha} < \frac{\e}{2 r } \mper  \]
Therefore, $f_e(t)$ is a continuous function.
\end{proof}
\prettyref{claim:hyper-ral-cont} implies that
\begin{equation}
\label{eq:hyper-ral-helper-3}
\lim_{\alpha \to 0} \sum_{e \in E}  f_e(t + \alpha) =   \sum_{e \in E}  f_e(t) \mper  
\end{equation}
Next, from the continuity of the dispersion process (\prettyref{def:hyper-projected-randomwalk}),
we get 
\begin{equation}
\label{eq:hyper-ral-helper-4}
\lim_{\alpha \to 0} \norm{ \omega^{t + \alpha} }^2 = \norm{\omega^t}^2
\end{equation}
Therefore, from \prettyref{eq:hyper-ral-helper-3} and \prettyref{eq:hyper-ral-helper-4}
we get that 
\[  \lim_{\alpha \to 0 } \ral{\omega^{t + \alpha}}  = \ral{\omega^t}   \]
and hence, this finishes the proof of the lemma.

\end{enumerate}
\end{proof}

\subsection{Bottlenecks for the Hypergraph Dispersion Process}
In this section we prove that if the hypergraph dispersion process mixes slowly, then
it must have a set of vertices having small expansion (\prettyref{thm:hyper-walk-cut}).

\begin{theorem}[Stronger form of \prettyref{thm:hyper-walk-cut}]
Given a hypergraph $H = (V,E,w)$ and a function $\omega^0 : V \to \R$ such that $\inprod{\omega^0,\mustat} = 0$, 
let $\omega^t$ denote the probability distribution at time $t$ according to the 
hypergraph dispersion process (\prettyref{def:hyper-projected-randomwalk} with $S = V$).
Then, for any $T > 0$, there exists a set $S \subset V$ such that $\mustat(S) \leq 1/2$ and  
\[ \phi(S) \leq  \bigo{ \min_{ t \in [0, T ] } \sqrt{  \frac{ \log \paren{  \norm{\omega^0}^2/ \norm{ \omega^t}^2   } }{t} }} \mper      \]
Moreover, such a set can be computed in time $\tbigo{T \, \Abs{E}}$.
\end{theorem}

\begin{proof}
Fix a time $t > 0$.
Using \prettyref{lem:hyper-walk-prop} \prettyref{eq:hyper-walk-prop-3} we get
\[  \frac{ \diff \norm{\omega^t}^2  }{\dt} = - 2 \ral{\omega^t} \norm{\omega^t}^2 \mper    \]
Integrating with respect to $t$ from $0$ to $t$ and using 
\prettyref{lem:hyper-walk-prop} \prettyref{eq:hyper-walk-prop-2} we get
\[ \log \paren{ \frac{ \norm{\omega^t}^2  }{ \norm{\omega^0 }^2 } }   \leq - 2 \ral{\omega^t} t \mper \]
Rearranging, we get 
\[ \ral{\omega^t} \leq  \frac{ \log \paren{  \norm{\omega^0}^2/ \norm{ \omega^t}^2   } }{2 t} \mper  \]
Using a proposition that we will prove in \prettyref{sec:hyper-cheeger} (\prettyref{prop:hyper-sweep-rounding}),
we can conclude that there exists a set $S \subset V$ such that  
\[ \phi(S)  \leq \bigo{\sqrt{ \ral{\omega^t}} } \leq 
\bigo{ \sqrt{  \frac{ \log \paren{  \norm{\omega^0}^2/ \norm{ \omega^t}^2   } }{t} }} \mper  \]
This completes the proof of this theorem.

To prove \prettyref{thm:hyper-walk-cut} as stated, we invoke this theorem with
$\omega^0 \defeq  \mu^0  - \paren{\inprod{\mu^0,\mustat}/\norm{\mustat}^2 } \mustat$ and observe that
\[   \frac{1}{4} \norm{\mu^t}^2 \leq   \norm{\omega^t}^2 \leq \norm{\mu^t}^2 \qquad \forall t \in [0, \tmix{\mu^0}/2 ]  \mper \]

\end{proof}

\subsection{Eigenvalues in Subspaces}

\begin{theorem}[Formal statement of  of \prettyref{thm:hyper-eigs-subspace-informal}]
\label{thm:hyper-eigs-subspace}
Given a hypergraph $H$, for every subspace $S$ of $\ \R^n$, the operator 
$\projsymb_S L$  has a eigenvector, i.e. 
there exists a non-zero vector $\eigvec \in S$ and a $\eig \in \R$ such that 
\[ \p_S L(\eigvec) = \eig\, \eigvec \qquad \textrm{and} \qquad 
 \eig = \min_{X \in S} \frac{ X^T \p_S L (X)  }{X^T X}  \mper \]
\end{theorem}

\begin{proof}
Fix a subspace $S$ of $\R^n$. 
Using \prettyref{lem:hyper-walk-prop} \prettyref{eq:hyper-walk-prop-2} and the compactness of the
unit ball, $\eig$ exists and is well defined. 
We define the set of vectors $U_{\eig}$ as follows.
\begin{equation}
\label{eq:hyper-eig-subspace-helper1}
U_{\eig} \defeq \set{ X \in S : X^T X = 1 \textrm{ and } X^T \p_S L (X) = \eig  }  \mper 
\end{equation}

From the definition of $\eig$, we get that $U_{\eig}$ is non-empty.
Now, the set $U_{\eig}$ could potentially have many vectors. We will show that at least one of them will
be an eigenvector. As a warm up, let us first consider the case when $\Abs{U_{\eig}} = 1$.
Let $\eigvec$ denote the unique vector in $U_{\eig}$. 
We will show that $\eigvec$ is an eigenvector of $\p_S L$.
To see this, we define the unit vector $\eigvec'$ as follows.  
\[ \eigvec' \defeq  \frac{ \p_S M(\eigvec) }{ \norm{ \p_S M(\eigvec) }  } \mper \]
Since $\eigvec$ is the vector in $S$
having the smallest value of $\ral{\cdot}$, we get 
\[ \ral{\eigvec} \leq \ral{\eigvec'} \mper  \]
But from \prettyref{lem:hyper-walk-prop}\prettyref{eq:hyper-walk-prop-2}, we get the 
$\ral{\cdot}$ is a monotonic function, i.e. $\ral{\eigvec'} \leq \ral{\eigvec}$  . Therefore
\[ \ral{\eigvec} = \ral{\eigvec'} \mper  \]
Therefore, $\eigvec'$ also belongs to $U_{\eig}$. But we assumed that $\Abs{U_{\eig}} = 1$. Therefore,
$\eigvec' = \eigvec$, or in other words $\eigvec$ is an eigenvector of $\p_S L$.
\[ \p_S L (\eigvec) =  (1 - \norm{\p_S M(\eigvec)})\, \eigvec = \eig\, \eigvec \mper    \]

The general case when $\Abs{U_{\eig}} > 1$ requires more work, as the operator $L$ is non-linear.
We follow the general idea of the case when $\Abs{U_{\eig}} = 1$.
We let $\omega^0 \defeq \eigvec$ for any $\eigvec \in U_{\eig}$.
We define the set of unit vectors $\set{\omega^t}_{t \in [0,1]}$ recursively as follows (for an infinitesimally 
small $\dt$).
\begin{equation}
\label{eq:hyper-eig-subspace-helper2}
\omega^{t + \dt} \defeq \frac{  \paren{ (1 - \dt)I + \dt\, \p_S M}\circ \omega^t }{
	\norm{\paren{ (1 - \dt)I + \dt\, \p_S M}\circ \omega^t} } \mper
\end{equation}
As before, we get that 
\begin{equation}
\label{eq:hyper-eig-subspace-helper3}
\omega^t \in U_{\eig} \qquad \forall t \geq 0 \mper
\end{equation}

If for any $t$, $\omega^t = \omega^{t'}$  $\forall t' \in [t, t + \dt]$, then $\omega^{t} = \omega^{t'} \ \forall t' \geq t$, 
and we have that $\omega^t$ is an eigenvector 
of $\p_S M$, and hence also of $\p_S L$ (of eigenvalue $\eig$).
Therefore, let us assume that $\omega^t \neq \omega^{t + \dt}$ $\forall t \geq 0$.

Let $\cA_{\omega}$ be the set of support matrices of $\set{\omega^t}_{t \geq 0}$, i.e.
\[ \cA_{\omega} \defeq \set{A_{\omega^t} : t \geq 0 } \mper   \]
Note that unlike the set $\set{\omega^t}_{t \geq 0}$ which could potentially be of 
uncountably infinite cardinality, the $\cA_{\omega}$ is of finite size.
A matrix $A_X$ is only determined by the pair of vertices in each hyperedge which have the largest
difference in the values of $X$. Therefore,
\[  \Abs{ \cA_{\omega}} \leq  \paren{2^r}^m < \infty \mper   \]
Now, since $\Abs{ \cA_{\omega}}$ is finite,
(using \prettyref{lem:discrete-function-lebesgue}) there exists 
$p,q \in [0,1]$, $p < q$ such that 
\[ A_{ \omega^{{\sf t}}} = A_{\omega^{{\sf p}}} \qquad \forall t \in [p, q] \mper   \]  
For the sake of brevity let $A \defeq A_{\omega^{{\sf p}}}$ denote this matrix.

We now show that $\omega^{p}$ is an eigenvector of $\p_S L$. 
From \prettyref{eq:hyper-eig-subspace-helper3}, we get that for infinitesimally
small $\dt$ (in fact anything smaller than $q - p$ will suffice),
\[ \ral{\omega^{p}} - \ral{\omega^{p + \dt}} = 0 \mper  \]
Let $\alpha_1, \ldots, \alpha_n$ be the eigenvalues of $A' \defeq \paren{ (1 - \dt)I + \dt\, A}$ and let 
$v_1, \ldots, v_n$ be the corresponding eigenvectors.
Since $A$ is a stochastic matrix,
\begin{equation}
\label{eq:hyper-eig-subspace-helper4}
 A \succeq (1 - 2 \dt) I \succeq \frac{1}{2} I \qquad \textrm{or } \qquad \alpha_i \geq \frac{1}{2}\ \forall i \mper  
\end{equation}
Let $c_1, \ldots, c_n \in \R$ be appropriate constants such that 
\[  \omega^{p} = \sum_i c_i v_i \mper \]
Then using \prettyref{prop:matrix-ineq-1}, we get that 

\begin{eqnarray*}
0 & = & \ral{\omega^{p}} - \ral{\omega^{p + \dt}} \\
 & = & \frac{1}{\dt} \cdot \paren{  \frac{ (\omega^{p})^T (I - \p_S A') \omega^{p}  }{  (\omega^{p})^T \omega^{p} } -   
\frac{ (\omega^{p})^T A' \p_S (I - \p_S A') \p_S A' \omega^{p}  }{  (\omega^{p})^T A' \p_S A' \omega^{p} } } \\
 & = & \frac{1}{\dt}  
2 \frac{ \sum_{i,j} c_i^2 c_j^2 (\alpha_i - \alpha_j)^2 (\alpha_i + \alpha_j) }{\sum_i c_i^2 \sum_i c_i^2 \alpha_i^2 } 
\mper  
\end{eqnarray*}
Since, all $\alpha_i \geq 1/2$ (from \prettyref{eq:hyper-eig-subspace-helper4}), the last term
can be zero if and only if for some eigenvalue $\alpha \in \set{\alpha_i : i \in [n]}$,
\[ c_i \neq 0 \textrm{ if and only if } \alpha_i = \alpha \mper \]
Or equivalently, $\omega^{p}$ is an eigenvector of $A$, and 
$\omega^t = \omega^{p}\ \forall t \in [p, q]$. Hence, by recursion
\[ \omega^t = \omega^{p} \qquad \forall t \geq p \mper  \]
Therefore, 
\[ \p_S L(\omega^{p}) = \paren{ \frac{1 - \alpha}{\dt}} \omega^{p}    \]
Since we have already established that $\ral{\omega^{p}} = \eig$, this finishes the proof of the theorem.

\end{proof}

\prettyref{prop:hyper-eigs-ral} follows from \prettyref{thm:hyper-eigs-subspace} as a corollary.
\begin{proof}[Proof of \prettyref{prop:hyper-eigs-ral} ]
We will prove this by induction on $k$. The proposition is trivially true of $k =1 $.
Let us assume that the proposition holds for $k-1$. We will show that it holds for 
$k$.
Recall that $ \eigvec_k$ is defined as  
\[ \eigvec_k = \argmin_{ X} \frac{ X^T \p_{S_{k-1}}^{\perp} L(X)  }{X^T \projsymb_{S_{k-1}}^{\perp}  X} 
  \mper \]
Then from \prettyref{thm:hyper-eigs-subspace}, we get that $\eigvec_k$ is indeed an 
eigenvector of $\p_{S_{k-1}}^{\perp} L$ with eigenvalue   
\[ \eig_k = \min_{ X} \frac{ X^T \p_{S_{k-1}}^{\perp} L(X)  }{X^T \projsymb_{S_{k-1}}^{\perp}  X} 
  \mper \]
\end{proof}

We now show that \prettyref{thm:hyper-2nd} follows almost directly from \prettyref{thm:hyper-eigs-subspace}.
\begin{theorem}[Restatement of \prettyref{thm:hyper-2nd}]
Given a hypergraph $H = (V,E,w)$, there exists a non-zero vector $v \in \R^n$ and a $\lambda \in \R$ such that 
$\inprod{v,\mustat} = 0$ and $L(v) = \lambda\, v $.
\end{theorem}

\begin{proof}
Considering the subspace of vectors orthogonal to $\mustat$,
from \prettyref{thm:hyper-eigs-subspace} we get that there exists a vector $\eigvec \in \R^n$  
and a $\lambda \in \R$ such that 
\[ \inprod{\eigvec, \mustat} = 0 \qquad \textrm{and} \qquad \p_{\set{\mustat}}^{\perp}  L (\eigvec) = \lambda\, \eigvec \mper   \]
Since $L_{\eigvec}$ is a Laplacian matrix, the vector $\mustat$ is an eigenvector with eigenvalue $0$.
Therefore,
\[ L (\eigvec) =   \p_{\set{\mustat}}^{\perp}  L (\eigvec) = \lambda\, \eigvec \mper   \]
This finishes the proof of the theorem.
\end{proof}

\subsection{Upper bounds on the Mixing Time}

\begin{theorem}[Restatement of \prettyref{thm:hyperwalk-upper}]
Given a hypergraph $H = (V,E,w)$, for all starting probability distributions $\mu^0 : V \to [0,1]$, 
the Hypergraph Dispersion Process (\prettyref{def:hyper-randomwalk}) satisfies
\[  \tmix{\mu^0} \leq \frac{\log (n/\delta)}{ \lh} \mper \]
\end{theorem}

\begin{proof}
Fix a probability distribution $\mu^0$ on $V$.
For the sake of brevity, let $A_t$ denote $A_{\mu^{{\sf t}}}$
and let $A_t'$ denote $\paren{ (1-\dt)I + \dt\, A_{\mu^{{\sf t}}}}$.
We first note that  
\begin{equation}
\label{eq:psd-walk-matrix}
A_t' \succeq (1 - 2 \dt)I + \succeq 0 \qquad \forall t \mper
\end{equation}
This follows from the fact that $A_t$ being a stochastic matrix, satisfies $I \succeq A_t \succeq - I$.
Let $1 \geq \alpha_2 \geq \ldots \geq \alpha_n$ be the eigenvalues of $A_t$ and let 
$\one/\sqrt{n}, v_2, \ldots, v_n$ be the corresponding eigenvectors.
Let $\alpha_i' \defeq (1 - \dt) + \dt \, \alpha_i$ for $i \in [n]$ be the eigenvalues of $A_t'$.
Writing $\mu^t$ in this eigen-basis, 
let $c_1, \ldots, c_n \in \R$ be appropriate constants such that $\mu^t = \sum_i c_i v_i$.
Since $\mu^t$ is a probability distribution on $V$, its component along the first eigenvector 
$v_1 = \one/\sqrt{n}$ is
\[ c_1 v_1 = \inprod{\mu^t,\frac{ \one}{\sqrt{n}}  } \frac{ \one}{\sqrt{n}} = \frac{ \one}{n} \mper \]
Then, using the fact that $\alpha_1' = (1 - \dt) + \dt \cdot 1 = 1$.
\begin{equation}
\label{eq:hyper-mix-helper3}
 \mu^{t + \dt} = A_t'\, \mu^t = \sum_{i=1}^n \alpha_i' c_i v_i = \frac{ \one}{n} 
+ \sum_{i=2}^n \alpha_i' c_i v_i \mper 
\end{equation}
Note that at all times $t \geq 0$, the component of $\mu^t$ along $\one$ (i.e. $c_1 v_1$)
remains unchanged.
Since for regular hypergraphs $\mustat = \one/n$, 
\begin{equation}
\label{eq:hyper-mix-helper1}
\norm{\mu^{t+\dt} - \mustat } = \norm{\mu^{t+\dt} - \one/n}  = \norm{\sum_{i = 2}^n \alpha_i' c_i v_i} 
	= \sqrt{ \sum_{i=2}^n \alpha_i'^2 c_i^2  } \mper 
\end{equation}
Since all the $\alpha_i' \geq 0$ (using \prettyref{eq:psd-walk-matrix}) and $\alpha_2 \geq \alpha_i$ 
$\forall i \geq 2$, $\alpha_2'^2 \geq \alpha_i'^2$ $\forall i \geq 2$. Therefore, from 
\prettyref{eq:hyper-mix-helper1}
\begin{equation}
\label{eq:hyper-mix-helper2}
\norm{\mu^{t+\dt} - \one/n} \leq \alpha_2' \sqrt{ \sum_{i=2}^n c_i^2} = \alpha_2' \norm{\mu^{t} - \one/n}
\mper 
\end{equation}
We defined $\eig_2$ to the second smallest eigenvalue of $L$. Therefore, from the definition of $L$,
it follows that $(1 - \eig_2)$ is the second largest eigenvalue of $M$. In this context, this implies
that
\[ \alpha_2 \leq 1 - \eig_2 \mper  \]
Therefore, from the definition of $\alpha_2'$
\[ \alpha_2' = (1 - \dt) + \dt \, \alpha_2 \leq  (1 - \dt) + \dt \, (1 - \eig_2) =  1 -\dt \, \eig_2 \mper \]
Therefore, from \prettyref{eq:hyper-mix-helper2},
\[  \norm{\mu^{t+\dt} - \one/n} \leq (1 - \dt\, \eig_2) \norm{\mu^{t} - \one/n} 
\leq e^{ -\dt \, \eig_2 } \norm{\mu^{t} - \one/n} \mper \]
Integrating with respect to time, from time $0$ to $t$,
\[ \norm{\mu^{t} - \one/n} \leq e^{ - \eig_2 t  } \norm{\mu^{0} - \one/n}  \leq 2 e^{ - \eig_2 t  } \mper \] 
Therefore, for $t \geq \log (n/\delta)/\eig_2 $, 
\[ \norm{\mu^{t} - \one/n} \leq \frac{\delta}{\sqrt{n}}   
\qquad \textrm{and} \qquad 
 \normo{\mu^{t} - \one/n} \leq \sqrt{n} \cdot \norm{\mu^{t} - \one/n} \leq \delta \mper  \]
Therefore,
\[  \tmix{\mu^0} \leq \frac{\log (n/\delta)}{ \lh} \mper \]

\end{proof}

\begin{remark}
\prettyref{thm:hyperwalk-upper} can also be proved directly 
by using \prettyref{lem:hyper-walk-prop},
but we believe that this proof is more intuitive.
\end{remark}

\subsection{Lower bounds on Mixing Time}

Next we prove \prettyref{thm:hyperwalk-lower}
\begin{theorem}[Restatement of \prettyref{thm:hyperwalk-lower}]
Given a hypergraph $H=(V,E,w)$, there exists a probability distribution $\mu^0$ on $V$
such that $\normo{\mu^0 -\one/n} \geq 1/2$
and 
\[ \tmix{\mu^0} \geq \frac{\log (1/\delta)  }{16\, \lh} \mper \]
\end{theorem}
In an attempt to motivate why \prettyref{thm:hyperwalk-lower} is true, 
we first prove the following (weaker) lower bound.
\begin{theorem}
Given a hypergraph $H=(V,E,w)$, there exists a probability distribution $\mu^0$ on $V$
such that $\normo{\mu^0 -\one/n} \geq 1/2$
and 
\[ \tmix{\mu^0} \geq \frac{\log (1/\delta) }{\phi_H} \mper \]
\end{theorem}

\begin{proof}[Proof Sketch]
Let $S \subset V$ be the set which has the least value of $\phi_H(S)$. Let 
$\mu^0 : V \to [0,1]$ be the probability 
distribution supported on $S$ that is stationary on $S$, i.e.
\[ \mu^0(i) = \begin{cases} \frac{1}{\Abs{S}} & i \in S \\ 0 & i \notin S  \end{cases} \]
Then, for an infinitesimal time duration $\dt$, only the edges in $E(S,\bar{S})$  
will be active in the dispersion process, and for each edge $e \in E(S,\bar{S})$,
the vertices in $e \cap S$ will be sending $1/d$ fraction of their mass to the vertices
in $e \cap \bar{S}$. Therefore, 
\[ \mu^0(S) - \mu^{\dt}(S) = \sum_{e \in E(S,\bar{S})} \frac{1}{d} \cdot \frac{1}{\Abs{S}}\, \dt
=  \frac{ \Abs{E(S,\bar{S})}  }{d \Abs{S}}\, \dt = \phi_H\, \dt \mper  \]

In other words, mass escapes from $S$ at the rate of $\phi_H$ initially. 
It is easy to show that the rate at which mass escapes from $S$ is a non-increasing function of time.
Therefore, it will take
at least $\Omega(1/\phi_H)$ units of time to remove $1/2$ of the mass from the $S$. Thus the lower
bound follows.

\end{proof}

Now, we will work towards proving \prettyref{thm:hyperwalk-lower}.

\begin{lemma}
\label{lem:hyperwalk-lower}
For any hypergraph $H=(V,E,w)$ and any probability distribution $\mu^0$ on $V$,
let $\alpha = \norm{\mu^0 - \one/n}^2$. Then
\[ \tmix{\mu^0} \geq \frac{ \log (\alpha/ \delta)  }{  4 \ral{\mu^0 - \one/n  } }  \mper \]
\end{lemma}

\begin{proof}
For a probability distribution $\mu^t$ on $V$, let $\omega^t$ be its component 
orthogonal to $\mustat = \one/\sqrt{n}$
\[ \omega^t \defeq \mu^t - \inprod{\mu^t, \frac{\one}{\sqrt{n}} } \frac{\one}{\sqrt{n}} 
= \mu^t - \frac{\one}{n} \mper \]
As we saw before (in \prettyref{eq:hyper-mix-helper3}), only $\omega^t$, the component of $\mu^t$ 
orthogonal to $\one$, changes with time; the component of $\mu^t$ along $\one$ does not change with
time. 
For the sake of brevity, let $\lambda = \ral{\mu^0 - \one/n}$.
Then, using \prettyref{lem:hyper-walk-prop}\prettyref{eq:hyper-walk-prop-2} and the definition of $\omega$, we get that
\[  \ral{\omega^t} \leq \ral{\omega^0} = \lambda  \qquad \forall t \geq 0 \mper \]
Now, using this and \prettyref{lem:hyper-walk-prop}\prettyref{eq:hyper-walk-prop-3} we get 
\[  \frac{ \diff \norm{\omega^t}^2  }{\norm{\omega^t}^2} = -2 \, \ral{\omega^t} \dt \geq - 2 \lambda\, \dt \mper  \] 
Integrating with respect to time from $0$ to $t$, we get  
\[ \log \norm{\omega^t}^2 - \log \norm{\omega^0}^2 \geq - 2 \lambda\, t \mper \]
Therefore
\[ e^{- 2 \lambda t} \leq  \frac{ \norm{\omega^t}^2  }{\norm{\omega^0}^2}
 = \frac{ \norm{\mu^t  -\one/n}^2 }{  \norm{ \mu^0 - \one/n}^2 } = \frac{ \norm{\mu^t  -\one/n}^2 }{ \alpha }
  \qquad \forall t \geq 0 \mper  \]
Hence
\[  \normo{\mu^t  -\one/n} \geq \norm{\mu^t  -\one/n} \geq 2 \delta \qquad \textrm{for } 
t \leq \frac{ \log (\alpha/ \delta)  }{  4 \lambda }  \mper    \]
Thus
\[ \tmix{\mu^0} \geq \frac{ \log (\alpha/ \delta)  }{  4 \ral{\mu^0 - \one/n  } }  \mper \]
\end{proof}

\begin{lemma}
\label{lem:hyperwalk-lower-ral}
Given a hypergraph $H = (X,E)$ and a vector $X \in \R^V$, there
exists a polynomial time algorithm to compute a probability distribution $\mu$ on $V$
satisfying 
\[ \normo{\mu-\one/n} \geq \frac{1}{2} \qquad \textrm{and} \qquad   
\ral{\mu - \one/n} \leq 4  \ral{X - \inprod{X,\one} \one/n  } \mper \] 
\end{lemma}

\begin{proof}
For the sake of building intuition, let us consider the case when $\inprod{X,\one} = 0$.
As a first attempt, one might be tempted to consider the vector $\one/n + X$. This
vector might not be a probability distribution if $X(i) < - 1/n$ for some coordinate $i$.
A simple fix for this would to consider the vector $\mu' \defeq \one/n + X/(n\, \normi{X})$.
This is clearly a probability distribution on the vertices, but  
\[  \normo{\mu' - \frac{\one}{n}} = \normo{  \frac{X}{n\, \normi{X}} } = \frac{ \normo{X} }{n\, \normi{X}}  \]
and $\normo{X}/( n\, \normi{X} ) \ll 1/2$ depending on $X$, for e.g. when $X$ is very sparse.
Therefore, we must proceed differently.

Since we only care about $\ral{ X - \inprod{X,\one} \one /n}$, 
w.l.o.g. we may assume that $\Abs{\supp(X^+)} = \Abs{\supp(X^-)}$ 
by simply setting $X := X + c \one$ for some appropriate constant $c$.
W.l.o.g. we may also assume that $\norm{X^+} \geq \norm{X^-}$.
Let $\omega$ be the component of $X^+$ orthogonal to $\one$
\[ \omega \defeq X^+ - \frac{\inprod{X^+,\one}}{n}\one  =   X^+ - \frac{\normo{X^+} }{n}\one \mper \] 
By definition, we get that $\inprod{\omega,\one} = 0$. 
Now,
\begin{equation}
\label{eq:hyperwalk-lower-helper1}
\normo{ \omega} \geq  \sum_{ i \in \supp(\omega^-)} \abs{\omega(i)}
\geq  \sum_{ i \in \supp(X^-)} \abs{\omega(i)} 
\geq  \frac{n}{2} \frac{\normo{X^+}}{n} \geq \frac{\normo{X^+}}{2 } \mper  
\end{equation}
We now define the probability distribution $\mu$ on $V$ as follows.
\[ \mu \defeq \frac{\one}{n} + \frac{\omega}{2 \normo{\omega}}  \mper \]
We now verify that $\mu$ is indeed a probability distribution, i.e. $\mu(i) \geq 0\ \forall i \in V$.
If vertex $i \in \supp(X^+)$, then clearly $\mu(i) \geq 0$. Lets consider an $i \in \supp(X^-)$.
\[ \frac{\omega(i) }{ 2 \normo{\omega}} =  \frac{ - \Abs{X^+} /n }{2 \normo{\omega}} \geq - \frac{1}{n} 
\qquad \textrm{(Using \prettyref{eq:hyperwalk-lower-helper1})} \mper \]
Therefore, $\mu(i) = 1/n + \omega(i)/(2 \normo{\omega}) \geq 0$ in this case as well. 
Thus, $\mu$ is a probability distribution on $V$.
Next, we work towards bounding $\ral{\mu - \one/n}$. 
\begin{equation}
\label{eq:hyperwalk-lower-helper2}
\sum_{e} w(e) \max_{i,j \in e} \paren{ \mu(i) - \mu(j)  }^2 
= \frac{1}{ 4 \normo{\omega}^2}\cdot \sum_{e} w(e) \max_{i,j \in e} \paren{ \omega(i) - \omega(j)  }^2
\leq \frac{1}{ 4 \normo{\omega}^2}\cdot \sum_{e} w(e) \max_{i,j \in e} \paren{ X(i) - X(j)  }^2 \mper
\end{equation}

We now bound $\norm{\omega}_2$.
\begin{equation}
\label{eq:hyperwalk-lower-helper4}
 \norm{\omega}_2^2 = \norm{ X^+ - \inprod{X^+,\one}\one/n}^2 = \norm{X^+}^2 - \frac{\inprod{X^+,\one}^2}{n} 
= \norm{X^+}^2 - \frac{\normo{X^+}^2 }{n}  \mper 
\end{equation}
Since $\Abs{\supp(X^+)} \leq n/2$, 
\[ \normo{X^+}^2 \leq \frac{n}{2} \norm{X^+}^2 \mper   \]
Combining this with \prettyref{eq:hyperwalk-lower-helper4}, and using our assumption that $\norm{X^+} \geq \norm{X^-}$,
 we get 
\[ \norm{\omega}_2^2 = \norm{X^+}^2 - \frac{\normo{X^+}^2 }{n}
\geq  \frac{\norm{X^+}^2}{2} \geq  \frac{\norm{X}^2}{4}  \mper \]
Therefore,
\begin{equation}
\label{eq:hyperwalk-lower-helper3}
\norm{\mu - \one/n}^2 = \frac{\norm{ \omega}^2 }{4 \normo{\omega}^2}  
\geq \frac{1}{4 \normo{\omega}^2} \cdot \frac{\norm{X}^2}{4} \geq
\frac{1}{4 \normo{\omega}^2} \cdot \frac{\norm{X - \inprod{X, \one }\one/n}^2   }{4}  \mper
\end{equation}

Therefore, using \prettyref{eq:hyperwalk-lower-helper2} and \prettyref{eq:hyperwalk-lower-helper3}, we get 
\[ \ral{\mu - \one/n} \leq 4  \ral{X - \inprod{X,\one} \one/n  }   \]
and by construction
\[ \normo{\mu - \one/n} = \normo{\frac{\omega}{2 \normo{\omega}}} = \frac{1}{2}  \]

\end{proof}

We are now ready to prove \prettyref{thm:hyperwalk-lower}.

\begin{proof}[Proof of \prettyref{thm:hyperwalk-lower}]~

Let $X = \eigvec_2$. Using \prettyref{lem:hyperwalk-lower-ral},
there exists a probability distribution $\mu$ on $V$ such that 
\[ \normo{\mu-\one/n} \geq \frac{1}{2} \qquad \textrm{and} \qquad   \ral{\mu - \one/n} \leq 4 \lh  \] 
and for this distribution $\mu$, using \prettyref{lem:hyperwalk-lower}, we get 
\[ \tmix{\mu} \geq \frac{\log (1/\delta)}{16\, \lh} \mper \]
\end{proof}

\begin{remark}
The distribution in \prettyref{thm:hyperwalk-lower} is not known to be computable in polynomial time.
We can  compute a probability distribution $\mu$ in polynomial time such 
\[ \normo{\mu-\one/n} \geq \frac{1}{2} \qquad \textrm{and} \qquad \tmix{\mu} \geq \frac{\log (1/\delta)}{c \lh \log r}  \]
for some absolute constant $c$.
Using \prettyref{thm:hyper-eigs-alg}, we get a vector $X \in \R^n$ such that $\ral{X} \leq c_1 \lh \log r$
for some absolute constant $c_1$. 
Using \prettyref{lem:hyperwalk-lower-ral},
we compute a probability distribution $\nu$ on $V$ such that 
\[ \normo{\nu-\one/n} \geq \frac{1}{2} \qquad \textrm{and} \qquad   \ral{\nu - \one/n} \leq 4 c_1 \lh \log r \mper \]
and for this distribution $\nu$, using \prettyref{lem:hyperwalk-lower}, we get
\[ \tmix{\nu} \geq \frac{\log (1/\delta)}{4 c_1 \lh \log r} \mper \]

\end{remark}

\section{Spectral Gap of Hypergraphs}

We define the {\em Spectral Gap} of a hypergraph to be $\eig_2$, the second smallest eigenvalue of 
its Laplacian operator.

\subsection{Hypergraph Cheeger's Inequality }
\label{sec:hyper-cheeger}

In this section we prove the hypergraph Cheeger's Inequality \prettyref{thm:hyper-cheeger}.
\begin{theorem}[Restatement of \prettyref{thm:hyper-cheeger}]
Given a hypergraph $H$,
\[ \frac{\lh}{2} \leq \phi_H \leq \sqrt{2 \lh} \mper  \]
\end{theorem}
Towards proving this theorem, 
we first show that a {\em good} line-embedding of the hypergraph  
suffices to upper bound the expansion.

\begin{proposition}
\label{prop:hyper-1d}
Let $H = (V,E,w)$ be a hypergraph with edge weights $w : E \to \R^+$  
 and let $Y \in [0,1]^{\Abs{V}}$ . 
Then there exists a set $S \subseteq {\sf supp}(Y)$ such that 
\[ \phi(S) \leq \frac{\sum_{e \in E} w(e) \max_{i,j \in e} \Abs{Y_i - Y_j} }{  \sum_i d_i Y_i }  \]
\end{proposition}

\begin{proof}
We define a family of functions $\set{F_r : [0,1] \to \set{0,1} }_{r \in [0,1]}$ as follows.
\[ F_r(x) = \begin{cases} 1 & x \geq r \\ 0 & \textrm{otherwise}     \end{cases} \]
Let $S_r$ denote the support of the vector $F_r(Y)$.
For any $a \in [0,1]$ it is easy to see that 
\begin{equation}
\label{eq:hyper-1d-helper1}
 \int_0^1 F_r(a)\, \dr  = a \mper 
\end{equation}
Now, observe that if $a - b \geq 0$, then $F_r(a) - F_r(b) \geq 0 \ \forall r \in [0,1]$ and similarly
if $a - b \leq 0$ then  $F_r(a) - F_r(b) \leq 0 \ \forall r \in [0,1]$. Therefore, 
\begin{equation}
\label{eq:hyper-1d-helper2}
 \int_0^1 \Abs{F_r(a) - F_r(b)} \dr = \Abs{ \int_0^1 F_r(a) \dr - \int_0^1 F_r(b) \dr }  =  \Abs{a-b} \mper 
\end{equation}
Also, for a hyperedge $e = \set{a_i : i \in [r]}$ if $\Abs{a_1 - a_2} \geq \Abs{a_i - a_j} \forall a_i,a_j \in e$, then 
\begin{equation}
\label{eq:hyper-1d-helper3}
 \Abs{F_r(a_1) - F_r(a_2)} \geq \Abs{ F_r(a_i) - F_r(a_j) } \quad \forall r \in [0,1] \textrm{ and }   \forall a_i, a_j \in e \mper 
\end{equation}
Therefore, 
\begin{align*}
\frac{\int_0^1  \sum_e w(e) \max_{i,j \in e} \Abs{F_r(Y_i) - F_r(Y_j)} \dr }{ \int_0^1 \sum_i d_i F_r(Y_i) \dr }
& =   \frac{ \sum_e w(e) \max_{i,j \in e} \int_0^1 \Abs{F_r(Y_i) - F_r(Y_j)} \dr }{ \int_0^1 \sum_i d_i F_r(Y_i) \dr } 
	& \textrm{(Using \prettyref{eq:hyper-1d-helper3})}	\\
& =   \frac{ \sum_e w(e) \max_{i,j \in e} \Abs{ \int_0^1 F_r(Y_i) - \int_0^1 F_r(Y_j)} \dr}{  \sum_i d_i \int_0^1 F_r(Y_i) \dr } 
	& \textrm{(Using \prettyref{eq:hyper-1d-helper2})}	\\
& =   \frac{ \sum_e w(e) \max_{i,j \in e} \Abs{ Y_i - Y_j} }{  \sum_i d_i Y_i } 
	& \textrm{(Using \prettyref{eq:hyper-1d-helper1})}	\mper
\end{align*}
Therefore, $\exists r' \in [0,1]$ such that 
\[ \frac{  \sum_e w(e) \max_{i,j \in e} \Abs{F_{r'}(Y_i) - F_{r'}(Y_j)} }{  \sum_i d_i F_{r'}(Y_i) } 
 \leq \frac{ \sum_e w(e) \max_{i,j \in e} \Abs{ Y_i - Y_j} }{  \sum_i d_i Y_i} \mper  \]
Since $F_{r'}(\cdot)$ is a value in $\set{0,1}$, we have 
\[ \frac{  \sum_e w(e) \max_{i,j \in e} \Abs{F_{r'}(Y_i) - F_{r'}(Y_j)} }{  \sum_{i \in V} d_i F_{r'}(Y_i) }
 = \frac{  \sum_e w(e) \Ind{ e \textrm{ is cut by } S_{r'} } }{ \sum_{i \in S_{r'}} d_i}  =  \phi(S_{r'}) \mper \] 
Therefore, 
\[ \phi(S_{r'}) \leq  \frac{ \sum_e w(e) \max_{i,j \in e} \Abs{ Y_i - Y_j} \dr}{  \sum_i d_i Y_i \dr} 
\qquad \textrm{and} \qquad S_{r'} \subset \supp(Y) \mper   \]

\end{proof}

\begin{proposition}
\label{prop:hyper-sweep-rounding}
Given a hypergraph $H = (V,E,w)$ and a vector $Y \in \R^{\Abs{V}}$ such that $\inprod{Y,\mustat} = 0$, 
there exists a set $S \subset V$ such that 
\[ \phi(S) \leq \ral{Y} +  2 \sqrt{\frac{ \ral{Y}}{\rmin} }  \mper  \]

\end{proposition}

\begin{proof}
Since $\inprod{Y,\mustat} = 0$, we have 
\[ \ral{Y} =  \frac{ \sum_{e \in E} w(e) \max_{i,j \in e} (Y_i - Y_j)^2 }{ \sum_i d_i Y_i^2  - (\sum_{i} d_i Y_i)^2/(\sum_i d_i) } 
= \frac{ \sum_{e \in E} w(e) \max_{i,j \in e} (Y_i - Y_j)^2 }{ \sum_{i,j} d_i d_j \paren{ Y_i - Y_j}^2 /(\sum_i d_i) }
\mper   \]

Let $X = Y + c \one$ for an appropriate $c \in \R$ such that 
$\Abs{\supp(X^+)} = \Abs{\supp(X^-)} = n/2$. 
Then we get 
\[ \ral{Y} =  \frac{ \sum_{e \in E} w(e) \max_{i,j \in e} (X_i - X_j)^2 }{ \sum_{i,j} d_i d_j \paren{ X_i - X_j}^2 /(\sum_i d_i) }
= \frac{ \sum_{e \in E} w(e) \max_{i,j \in e} (X_i - X_j)^2 }{ \sum_i d_i X_i^2  - (\sum_{i} d_i X_i)^2/(\sum_i d_i) }
\geq \ral{X} \mper  \]

For any $a,b \in R$, we have 
\[ (a^+ - b^+)^2 + (a^- - b^-)^2 \leq  (a - b)^2    \]
Therefore we have 
\begin{eqnarray*}
\ral{Y}  & \geq & \ral{X} = \frac{ \sum_{e \in E} w(e) \max_{i,j \in e} (X_i - X_j)^2 }{ \sum_i d_i X_i^2 } \\ 
	& \geq & \frac{ \left( \sum_{e \in E} w(e) \max_{i,j \in e} (X_i^+ - X_j^+)^2  \right)
	+ \left( \sum_{e \in E} w(e) \max_{i,j \in e} (X_i^- - X_j^-)^2  \right) }{ \sum_i d_i (X_i^+)^2   + \sum_i d_i (X_i^-)^2 } \\ 
 & \geq & \min \set{ \frac{ \sum_{e \in E} w(e) \max_{i,j \in e} (X_i^+ - X_j^+)^2 }{\sum_i d_i (X_i^+)^2} , 
	\frac{ \sum_{e \in E} w(e)\max_{i,j \in e} (X_i^- - X_j^-)^2 }{\sum_i d_i (X_i^-)^2} } \\
\end{eqnarray*}
Let $Z \in \set{X^+, X^-}$ be the vector corresponding the minimum in the previous inequality. Then
\begin{align*}
 \sum_{e \in E} w(e) \max_{i,j \in e} \Abs{Z_i^2 - Z_j^2} 
& = \sum_{e \in E} w(e) \max_{i,j \in e} \Abs{Z_i - Z_j}(Z_i + Z_j) \\
& = \sum_{e \in E} w(e) \max_{i,j \in e} (Z_i - Z_j)^2 +2 \sum_{e \in E} w(e) \min_{i \in e} Z_i \max_{i,j \in e} \Abs{Z_i - Z_j} \\  
& \leq \sum_{e \in E} w(e) \max_{i,j \in e} (Z_i - Z_j)^2 
		+2 \sqrt{ \sum_{e \in E} w(e) \max_{i,j \in e} (Z_i - Z_j)^2} \sqrt{\sum_{e \in E} w(e) \frac{ \sum_{i \in e} Z_i^2}{\rmin} } \\
& = \sum_{e \in E} w(e) \max_{i,j \in e} (Z_i - Z_j)^2
+ 2 \sqrt{ \sum_{e \in E} w(e) \max_{i,j \in e} (Z_i - Z_j)^2} \sqrt{ \frac{ \sum_{i \in V} d_i Z_i^2}{\rmin} } 
\end{align*}

Using $\ral{Z} \leq \ral{Y}$,
\[ \frac{\sum_{e \in E} w(e) \max_{i,j \in e} \Abs{Z_i^2 - Z_j^2}}{ \sum_i d_i Z_i^2} 
\leq \ral{Z} + 2 \sqrt{ \frac{  \ral{Z}}{\rmin} } \leq  \ral{Y} +  2 \sqrt{\frac{ \ral{Y}}{\rmin} } \mper  \] 
Invoking \prettyref{prop:hyper-1d} with vector $Z^2$, we get that there exists a set $S \subset \supp \paren{Z}$ such that
\[ \phi(S) \leq \ral{Y} + 2 \sqrt{\frac{ \ral{Y}}{\rmin} }  \qquad \textrm{and } \qquad  \Abs{S} \leq \Abs{\supp \paren{Z}} \leq \frac{n}{2} \mper \]

\end{proof}

We are now ready to prove \prettyref{thm:hyper-cheeger}.
\begin{proof}[Proof of \prettyref{thm:hyper-cheeger}]~

\begin{enumerate}
\item 
Let $S \subset V$ be any set such that $\vol(S) \leq \vol(V)/2$, and let $X \in \R^n$ be the indicator vector 
of $S$. Let $Y$ be the component of $X$ orthogonal to $\mustat$.
Then 
\begin{align*}
\lh & \leq \frac{ \sum_e w(e) \max_{i,j \in e} (Y_i - Y_j)^2  }{ \sum_i d_i Y_i^2 }
 = \frac{ \sum_e w(e) \max_{i,j \in e} (X_i - X_j)^2  }{ \sum_i d_i X_i^2 -  (\sum_i d_i X_i)^2/(\sum_i d_i) } \\
  & = \frac{ w(E(S, \bar{S} )) }{ \vol(S) - \vol(S)^2/\vol(V)} = \frac{\phi(S)}{ 1 - \vol(S)/\vol(V) } \\
  & \leq 2 \phi(S) \mper  
\end{align*}
Since the choice of the set $S$ was arbitrary, we get 
\[ \frac{\lh}{2} \leq \phi_H \mper   \]

\item 
Invoking \prettyref{prop:hyper-sweep-rounding} with $\eigvec_2$ we get that
\[ \phi_H \leq \ral{\eigvec_2} +  \sqrt{\frac{ \ral{\eigvec_2}}{\rmin} } 
 = \eig_2 + 2 \sqrt{ \frac{\eig_2}{\rmin}  } \leq \sqrt{2\, \lh} \mper \]
\end{enumerate}
\end{proof}

\subsection{Hypergraph Diameter }
\label{sec:hyper-diam}
In this section we prove \prettyref{thm:hyper-diam}.
\begin{theorem}[Restatement of \prettyref{thm:hyper-diam}]
Given a hypergraph $H = (V,E,w)$ with all its edges having weight $1$, its diameter is at most 
\[ \diam(H) \leq \bigo{\frac{\log n}{\log \frac{1}{1 - \lh}}} \mper \]
\end{theorem}

\begin{remark}
A weaker bound on the diameter follows from \prettyref{thm:hyperwalk-upper}
\[ \diam(H) \leq \bigo{ \frac{\log n}{\eig_2} } \mper \]
\end{remark}

We start by defining the notion of operator powering.
\begin{definition}[Operator Powering]
For a $t \in \N$, and an operator $M : \R^n \to \R^n$, for a vector $X \in \R^n$
we define $M^t(X)$ as follows
\[ M^t(X) \defeq M( M^{t-1}(X))  \qquad \textrm{and } \qquad M^1(X) \defeq M(X) \mper \]
\end{definition}

Next, we state bound the norms of powered operators.
\begin{lemma}
\label{lem:hyper-higher-norms}
For vector $\omega \in \R^n$, such that $\inprod{\omega,\one}=0$,
\[  \norm{M^t(\omega)} \leq (1 - \eig_2)^{t/2} \norm{\omega} \mper    \]
\end{lemma}

\begin{proof}
We prove this by induction on $t$. 
Let $v_1, \ldots, v_n$ be the eigenvectors of $A_{\omega}$ and let $\lambda_1, \ldots, \lambda_n$ be the 
the corresponding eigenvalues. 
Let $\omega = \sum_{i = 1}^n c_i v_i $ for appropriate constants $c_i \in \R$. Then, for $t=1$,  
\begin{align}
\frac{ \norm{ M(\omega)}}{\norm{\omega} } & =  \frac{ \norm{ A_{\omega}\, \omega }}{\norm{\omega} } 
		= \sqrt{ \frac{ \sum_i c_i^2 \lambda_i^2  }{ \sum_i c_i^2  } } 
  \leq \sqrt{ \frac{ \sum_i c_i^2 \lambda_i  }{ \sum_i c_i^2  } } 
	& \textrm{(Since each $ \lambda_i \in [0,1]$, $\lambda_i^2 \leq \lambda_i$)} \nonumber \\ 
 & = \sqrt{ \frac{ \omega^T M (\omega) }{ \omega^T \omega  } } \leq \sqrt{1 - \eig_2} 
	& \textrm{(From the definition of $\eig_2$)} 
\end{align}

Similarly, for $t > 1$.
\[ \norm{M^t(\omega)} =  \norm{M( M^{t-1}(\omega))} \leq (1 - \eig_2)^{1/2} \norm{ M^{t-1}(\omega) } 
\leq (1- \eig_2)^{t/2} \norm{\omega}    \]
where the last inequality follows from the induction hypothesis. 
\end{proof}

\begin{proof}[Proof of \prettyref{thm:hyper-diam}]
For the sake of simplicity, we will assume that the hypergraph is regular.
Our proof easily extends to the general case.
We define the operator $M' \defeq I/2 + M/2$.
Then the eigenvalues of $M'$ are $1/2 + \eig_i/2$,
and the corresponding eigenvectors are $\eigvec_i$, 
 for $i \in [n]$.

Our proof strategy is as follows. Fix some vertex $u \in V$. Consider the vector  
$M'(\chi_u)$. This vector will have non-zero values at exactly those coordinates 
which correspond to vertices that are at a distance of at most $1$ from $u$ (see also
\prettyref{rem:hyper-walk-ties}). Building on this
idea, it follows that the vector $M'^t (\chi_u)$ will have non-zero values at exactly those 
coordinates which correspond to vertices that are at a distance of at most $t$ from $u$. 
Therefore, the diameter of $H$ is the smallest value $t \in \N$ such that the vectors
$\set{ M'^t(\chi_u) : u \in V }$ have non-zero entries in all coordinates. We will 
upper bound the value of such a $t$. 

Fix two vertices $u,v \in V$. Let $\chi_u, \chi_v$ be their respective characteristic vectors and 
let $\omega_u, \omega_v$ be the components of $\chi_u, \chi_v$ orthogonal to $\one$ respectively 
\[  \omega_u \defeq \chi_u - \frac{\one}{n} \qquad \textrm{and} \qquad   \omega_v \defeq \chi_v - \frac{\one}{n} \mper  \]
Then 
\begin{equation}
\label{eq:hyper-diam-helper3}
\norm{\omega_u} = \sqrt{ \paren{\chi_u - \frac{\one}{n}}^T   \paren{\chi_u - \frac{\one}{n}} }
 = \sqrt{  1 - \frac{1}{n} - \frac{1}{n} + \frac{n}{n^2}  } = \sqrt{1 - \frac{1}{n} } \mper 
\end{equation}
Since $\one$ is invariant under $M'$ we get 
\begin{align*}
  \chi_u^T M'^t (\chi_v) & = \paren{ \frac{\one}{n}  + \omega_u }^T M'^t \paren{ \frac{\one}{n}  + \omega_v}      
 = \paren{ \frac{\one}{n}  + \omega_u }^T  \paren{ \frac{\one}{n}  + M'^t( \omega_v) } \\
& = \frac{1}{n} + 0 + \frac{1}{n} \one^T M'^t( \omega_v) + \omega_u^T M'^t( \omega_v) \mper 
\end{align*}
Now since $M'$ is a dispersion process, if $\inprod{\omega_u,\one} = 0$, then $\inprod{ M'(\omega_u), \one  } = 0 $
and hence $\inprod{ M'^t(\omega_u), \one  } = 0 $. Therefore, 
\begin{equation}
\label{eq:hyper-diam-helper1}
 \chi_u^T M'^t \chi_v = \frac{1}{n} + \omega_u^T M'^t( \omega_v) \mper 
\end{equation}
Now,
\[ \Abs{ \omega_u^T M'^t( \omega_v) } \leq \norm{ \omega_u}\, \norm{M'^t( \omega_v)}
 \leq \paren{\frac{1 - \eig_2}{2}}^{t/2} \norm{ \omega_u}\, \norm{ \omega_v} \qquad \textrm{(Using \prettyref{lem:hyper-higher-norms}).} \]
Therefore, from \prettyref{eq:hyper-diam-helper1} and \prettyref{eq:hyper-diam-helper3}, 
\begin{equation}
\label{eq:hyper-diam-helper2}
\chi_u^T M'^t \chi_v \geq \frac{1}{n} - \paren{\frac{1 - \eig_2}{2}}^{t/2} \norm{ \omega_u}\, \norm{ \omega_v}  
\geq \frac{1}{n} - \paren{\frac{1 - \eig_2}{2}}^{t/2} \paren{1  - \frac{1}{n}  } \mper 
\end{equation}
Therefore, for 
\[  t \geq  \frac{ 2 \log (n/2)}{  \log \paren{   \frac{2}{1 - \eig_2} } },         \]
we have 
$\chi_u^T M'^t \chi_v > 0$.
Therefore, 
\[  \diam(H) \leq  \frac{  \log n }{  \log \paren{   \frac{1}{1 - \eig_2} } } \mper \]

\end{proof}

\section{Higher Eigenvalues and Hypergraph Expansion}
\label{sec:hyper-higher-cheeger}

In this section we will prove \prettyref{thm:hyper-sse-informal}
and \prettyref{thm:hyper-higher-cheeger}.

\subsection{Small Set Expansion}
\label{sec:hyper-sse}

\begin{theorem}[Formal Statement of \prettyref{thm:hyper-sse-informal}]
\label{thm:hyper-sse}
There exists an absolute constant $C$ such that
 every hypergraph $H = (V,E,w)$ and parameter $k < \Abs{V}$, there exists a set $S \subset V$ 
such that $\Abs{S} \leq 24 \Abs{V}/k$ satisfying 
\[ \phi(S) \leq C \ssevalue \sqrt{ \eig_k}    \]
where $r$ is the size of the largest hyperedge in $E$. 
\end{theorem}

Our proof will be via a simple randomized polynomial time algorithm (\prettyref{alg:hyper-sse}) 
to compute a set $S$ satisfying the conditions of the theorem. We will use the following rounding
step as a subroutine.
\begin{lemma}[\cite{lm14b}\footnote{We remark that the algorithm from \cite{lm14b} 
can not directly be used here as the vectors 
$\set{\tilde{u}_i}_{i \in V}$ need not have non-negative inner product.}]
\label{lem:gen-orth-sep}
There exists a randomized polynomial time algorithm that given a set of unit vectors 
$\set{\U}_{u \in V}$, a parameters $\beta \in (0,1)$ and $m \in \Z^+$ outputs a random set 
$S \subset \set{\U}_{u \in V} $ such that 
\begin{enumerate}
\item $\Pr{\U \in S}  = 1/m$.

\item For every $\U,\V$ such that $\inprod{\U,\V} \leq \beta$,
\[ \Pr{ \U \in S \textrm{ and } \V \in S } \leq 1/m^2 \mper     \]

\item For any $e \subset \set{\U}_{u \in V}$ 
\[ \Pr{ e \textrm{ is ``cut'' by  } S } \leq \frac{c_1}{\sqrt{1 - \beta}} \log m \log \log m \sqrt{\log \Abs{e}} \max_{\U,\V \in e}\, \norm{\U - \V}    \]
for some absolute constant $c_1$.
\end{enumerate}
\end{lemma}

\begin{mybox}
\begin{algorithm}~
\begin{enumerate}
\item {\bf Spectral Embedding}. We first construct a mapping of the vertices in $\R^k$ using the 
first $k$ eigenvectors. We map a vertex $i \in V $ to the vector $u_i$ defined as follows.
\[   u_i(l) = \frac{1}{\sqrt{d_i}} \eigvec_l(i) \mper \]
In other words, we map the vertex $i$ to the vector formed by taking the $i^{th}$ coordinate
from the first $k$ eigenvectors.

\item {\bf Random Projection}. 
Using \prettyref{lem:gen-orth-sep}, sample a random set $S$ from the set of vectors $\set{\tilde{u}_i}_{i \in V}$
with $\beta = 99/100$ and $m = k$, 
and define the vector $X \in \R^n$ as follows. 
\[  X(i) \defeq \begin{cases}  \norm{u_i}^2 & \textrm{if } \tilde{u}_i \in S \\
		0 & \textrm{otherwise}   \end{cases} \mper \]
\label{step:hyper-algstep-helper1}
\item {\bf Sweep Cut}. Sort the entries of the vector $X$ in decreasing order and output
the level set having the least expansion (See \prettyref{prop:hyper-1d}).

\end{enumerate}
\label{alg:hyper-sse}
\end{algorithm}
\end{mybox}

We first prove some basic facts about the Spectral Embedding (\prettyref{lem:hyper-spectral-embedding}). 
The analogous facts for graphs are well known (folklore).

\begin{lemma}[Spectral embedding]~
\begin{enumerate}
\item
\[  \frac{ \sum_{e \in E} \max_{i,j \in e} w(e) \norm{u_i - u_j}^2 }{ \sum_i d_i \norm{u_i}^2 }  \leq \eig_k \mper  \]
\item
\[ \sum_{i \in V} d_i \norm{u_i}^2 = k \mper \]
\item
\[ \sum_{i,j \in V} d_i d_j \inprod{u_i,u_j}^2 = k \mper \]
\end{enumerate}
\label{lem:hyper-spectral-embedding}
\end{lemma}

\begin{proof} 
\begin{enumerate}
\item Follows directly from the fact that $\set{u_i}_{i \in V}$ were constructed using
the $k$ vectors, each having Rayleigh quotient at most $\eig_k$.
\item Follows from the fact that each eigenvector is of length $1$.
\item 
\begin{align*}
\sum_{i,j} d_i d_j \inprod{u_i,u_j}^2 & =  \sum_{i,j} d_i d_j \left(\sum_{t = 1}^k u_i(t) u_j(t)\right)^2 \\
& =  \sum_{i,j} d_i d_j \sum_{t_1, t_2} u_i(t_1) u_j(t_1) u_i(t_2) u_j(t_2) 
 =  \sum_{t_1, t_2} \sum_{i,j} d_i d_j u_i(t_1) u_j(t_1) u_i(t_2) u_j(t_2) \\
& =  \sum_{t_1, t_2} \left( \sum_i d_i u_i(t_1) u_i(t_2) \right)^2 
\end{align*}

Since $\sqrt{d_i} u_i(t_1)$ is the entry to corresponding to vertex $i$ in the $t_1^{th}$ eigenvector,
$\sum_i d_i u_i(t_1) u_i(t_2)$ is equal to the inner product of the $t_1^{th}$ and $t_2^{th}$ eigenvectors
of $L$, which is equal to $1$ only when $t_1 = t_2$ and is equal to $0$ otherwise.
Therefore,
\[ \sum_{i,j} d_i d_j \inprod{u_i,u_j}^2 =  \sum_{t_1, t_2} \left( \sum_i d_i u_i(t_1) u_i(t_2) \right)^2 
 =  \sum_{t_1, t_2} \Ind{t_1 = t_2} = k \mper \]

\end{enumerate}
\end{proof}

For the sake of brevity let $\tau$ denote
\begin{equation}
\label{eq:taudef}
\tau \defeq \lmvalue \mper 
\end{equation}

\paragraph{Main Analysis.} To prove that \prettyref{alg:hyper-sse} outputs a set which meets the 
requirements of \prettyref{thm:hyper-sse}, we will show that the vector $X$ meets the requirements
of \prettyref{prop:hyper-sweep-rounding}. We will need an upper bound on the numerator of  
{\em cut-value} of the vector $X$ (\prettyref{lem:hyper-sse-num}), and a lower bound on the denominator of the 
{\em cut-value} of the vector $X$ (\prettyref{lem:hyper-sse-denom}).

\begin{lemma}
\label{lem:hyper-sse-num}
\[ \Ex{ \sum_{ e \in E} w(e) \max_{i,j \in e} \Abs{X_i - X_j}} \leq 8  c_1 \tau\, \sqrt{ \eig_k } \mper \]
\end{lemma}

\begin{proof}

For an edge $e \in E$ we have 
\begin{eqnarray}
\label{eq:hyper-sse-helper16}
\Ex{ \max_{i,j \in e} \Abs{X_i - X_j} }& \leq & \max_{i,j \in e} \Abs{ \norm{u_i}^2 - \norm{u_j}^2}  
		\Pr{\tilde{u}_i \in S \ \forall i \in e} 
   + \max_{i \in e} \norm{u_i}^2 \Pr{  e \textrm{ is cut by } S} 
\end{eqnarray}
The first term can be bounded by 
\begin{equation}
\label{eq:hyper-sse-helper10}
\frac{1}{k} \max_{i,j \in e} \Abs{ \norm{u_i}^2 - \norm{u_j}^2} \leq
 \frac{1}{k} \max_{i,j \in e} \norm{u_i - u_j} \cdot \norm{ u_i + u_j } \leq 
2 \frac{1}{k} \max_{i,j \in e} \norm{u_i - u_j} \max_{i \in e} \norm{u_i}  \mper 
\end{equation}

To bound the second term in  \prettyref{eq:hyper-sse-helper16}, 
we will divide the edge set $E$ into two parts $E_1$ and $E_2$ as follows.
\[ E_1 \defeq \set{ e \in E :  \max_{i,j \in e} \frac{ \norm{u_i}^2 }{  \norm{u_j}^2  } \leq 2   } 
\quad \textrm{and} \quad   
E_2 \defeq \set{ e \in E :  \max_{i,j \in e} \frac{ \norm{u_i}^2 }{  \norm{u_j}^2  } > 2 } \mper \] 
$E_1$ is the set of those edges whose vertices have roughly equal lengths and 
$E_2$ is the set of those edges whose vertices have large disparity in lengths.
For a hyperedge  $e \in E_1$, 
using \prettyref{lem:2vecs} and  \prettyref{lem:gen-orth-sep}, the second term 
in  \prettyref{eq:hyper-sse-helper16} can be bounded by
\begin{equation}
\label{eq:hyper-sse-helper2}
\frac{2 c_1 \tau }{k}  \max_{l \in e} \norm{u_l}^2 \max_{i,j \in e} 
 \frac{ \norm{u_i - u_j} }{\sqrt{ \norm{u_i}^2 + \norm{u_j}^2 }}  
\leq \frac{2 c_1 \tau }{k}  \max_{l \in e} \norm{u_l} \max_{i,j \in e} \norm{u_i - u_j} \mper 
\end{equation}
Let us analyze the edges in $E_2$. Fix any $e \in E_2$. Let $e = \set{u_1, \ldots, u_r}$ such that
$\norm{u_1} \geq \norm{u_2} \geq \ldots \geq \norm{u_r}$. Then from the definition of $E_2$ we have that
\[   \frac{\norm{u_1}^2 }{ \norm{u_r}^2 } > 2 \mper \]
Rearranging, we get
\begin{align*} 
\norm{u_1}^2 & \leq 2 \paren{ \norm{u_1}^2 - \norm{u_r}^2 } = 2 \inprod{u_1 - u_r , u_1 + u_r} 
  \leq 2 \norm{u_1 + u_r} \, \norm{u_1 - u_r} \\
 & \leq  2\sqrt{2} \max_{i \in e} \norm{u_i} \max_{i,j \in e} \norm{u_i - u_j} \mper
\end{align*}
Therefore for an edge $e \in E_2$, using this and \prettyref{lem:gen-orth-sep},
the second term in  \prettyref{eq:hyper-sse-helper16} can be bounded by
\begin{equation}
\label{eq:hyper-sse-helper3}
\frac{4c_1 \tau }{k} 
\max_{i \in e} \norm{u_i} \max_{i,j \in e} \norm{u_i - u_j}  \mper
\end{equation}
Using \prettyref{eq:hyper-sse-helper16}, \prettyref{eq:hyper-sse-helper10},  
\prettyref{eq:hyper-sse-helper2} and\prettyref{eq:hyper-sse-helper3}
we get 
\begin{equation}
\label{eq:hyper-sse-helper1}
\Ex{\max_{i,j \in e} \Abs{X_i - X_j}} \leq \frac{8 c_1 \tau }{k} \max_{l \in e} \norm{u_l} \max_{i,j \in e} 
  \norm{u_i - u_j}  \mper
\end{equation}
\begin{align*}
\Ex{\sum_{e \in E} w(e) \max_{i,j \in e} \Abs{X_i - X_j}} & \leq \frac{8 c_1 \tau }{k} 
\sum_{e \in E} w(e) \max_{i \in e} \norm{u_i} \max_{i,j \in e} \norm{u_i - u_j} \\
& \leq \frac{8 c_1 \tau }{k} 
\sqrt{ \sum_{e \in E} w(e) \max_{i \in e} \norm{u_i}^2 } \sqrt{ \sum_{e \in E} w(e) \max_{i,j \in e} \norm{u_i - u_j}^2 } \\
& \leq \frac{8 c_1 \tau  }{k} 
\sqrt{  \sum_{i \in V} d_i \norm{u_i}^2 } \sqrt{ \sum_{e \in E} w(e) \max_{i,j \in e} \norm{u_i - u_j}^2 } \\
& \leq 8 c_1 \tau \sqrt{\eig_k } \qquad \textrm{(Using \prettyref{lem:hyper-spectral-embedding})} 
\end{align*}

\end{proof}

\begin{lemma}
\label{lem:hyper-sse-denom}
\[ \Pr{ \sum_{i \in V} d_i X_i  > \frac{1}{2} } \geq \frac{1}{12} \mper  \]

\end{lemma}

\begin{proof}
For the sake of brevity, we define $D \defeq \sum_{i \in V} d_i X_i$.
We first bound $\Ex{D}$ as follows. 
\begin{align*}
\Ex{D} & = \sum_{i \in V} d_i \norm{u_i}^2 \Pr{ \tilde{u}_i \in S } \\
 & = \sum_{i \in V} d_i \norm{u_i}^2 \cdot \frac{1}{k} & \textrm{(From \prettyref{lem:gen-orth-sep})} \\
 & =  k \cdot \frac{1}{k} = 1  &  \textrm{(Using \prettyref{lem:hyper-spectral-embedding})} \mper
\end{align*}

Next we bound the variance of $D$. 
\begin{align*}
\Ex{D^2} & = \sum_{i,j} d_i d_j \norm{u_i}^2 \norm{u_j}^2 \Pr{ \tilde{u}_i, \tilde{u}_i \in S } \\
 & \leq  \sum_{ \substack{i,j \\ \inprod{ \tilde{u}_i, \tilde{u}_j} \leq \beta }  } 
		d_i d_j \norm{u_i}^2 \norm{u_j}^2  \Pr{ \tilde{u}_i, \tilde{u}_i \in S } 
	+ 	\sum_{ \substack{i,j \\ \inprod{ \tilde{u}_i, \tilde{u}_j} > \beta }  } 
		d_i d_j \norm{u_i}^2 \norm{u_j}^2  \Pr{ \tilde{u}_i, \tilde{u}_i \in S }  
\end{align*}
We use \prettyref{lem:gen-orth-sep} to bound the first term, and use the trivial bound of 
$1/k$ to bound $\Pr{ \tilde{u}_i, \tilde{u}_i \in S }$ in the second term. Therefore, 
\begin{align*}
\Ex{D^2} & \leq \sum_{ \substack{i,j \\ \inprod{ \tilde{u}_i, \tilde{u}_j} \leq \beta }  } 
		d_i d_j \norm{u_i}^2 \norm{u_j}^2 \frac{1}{k^2}
	+ \sum_{ \substack{i,j \\ \inprod{ \tilde{u}_i, \tilde{u}_j} > \beta }  } 
		d_i d_j \norm{u_i}^2 \norm{u_j}^2 \frac{ \inprod{ \tilde{u}_i, \tilde{u}_j}^2 }{\beta^2} \frac{1}{k} \\
& \leq \sum_{i,j} d_i d_j \paren{   \frac{\norm{u_i}^2 \norm{u_j}^2  }{k^2}  + \frac{1}{\beta^2 k} \inprod{u_i,u_j}^2  } \\
 & =  \frac{1}{k^2} \paren{ \sum_{i} d_i \norm{u_i}^2 }^2 + \frac{1}{\beta^2 k} \sum_{i,j} d_i d_j \inprod{u_i,u_j}^2  \\
 & =  \frac{1}{k^2} \cdot k^2 + \frac{1}{\beta^2 k} \cdot k  = 1 + \frac{1}{\beta^2} \leq  3  
	& \textrm{(Using \prettyref{lem:hyper-spectral-embedding})} \mper
\end{align*}

Since $D$ is a non-negative random variable, we get using the Paley-Zygmund inequality that 
\[ \Pr{D \geq \frac{1}{2} \Ex{D} } \geq \paren{\frac{1}{2}}^2 \frac{ \Ex{D}^2  }{\Ex{D^2}} 
= \frac{1}{4} \cdot \frac{1}{3} = \frac{1}{12} \mper \]
This finishes the proof of the lemma.

\end{proof}

We are now ready finish the proof of \prettyref{thm:hyper-sse}.
\begin{proof}[Proof of \prettyref{thm:hyper-sse}]
By definition of \prettyref{alg:hyper-sse},
\[ \Ex{ \Abs{ \supp(X)} } = \frac{n}{k} \mper \]
Therefore, by Markov's inequality,
\begin{equation}
\label{eq:hyper-sse-helper4}
 \Pr{ \Abs{ \supp(X)}  \leq 24 \frac{n}{k}   } \geq 1 - \frac{1}{24}     \mper
\end{equation}

Using Markov's inequality and \prettyref{lem:hyper-sse-num},
\begin{equation}
\label{eq:hyper-sse-helper5}
\Pr{ \sum_{ e \in E} w(e) \max_{i,j \in e} \Abs{X_i - X_j} \leq 384  c_1 \tau \sqrt{ \eig_k  } }
\geq 1 - \frac{1}{48}  \mper
\end{equation}

Therefore, using a union bound over \prettyref{eq:hyper-sse-helper4}, \prettyref{eq:hyper-sse-helper5}
and \prettyref{lem:hyper-sse-num}, we get that
\[  \Pr{ \frac{ \sum_{e \in E} w(e) \max_{i,j \in e} \Abs{X_i - X_j} }{ \sum_i d_i X_i } 
\leq 1000 c_1 \tau \sqrt{\eig_k}  \textrm{ and }  \Abs{\supp(X)} \leq 24 \frac{n}{k} } \geq \frac{1}{48} \mper    \]
Invoking \prettyref{prop:hyper-sweep-rounding} on this vector $X$, we get that with probability at least $1/48$,
\prettyref{alg:hyper-sse} outputs a set $S$ such that
\begin{equation}
\label{eq:hyper-sse-helper21}
 \phi(S) \leq 1000 c_1 \tau \sqrt{\eig_k} \qquad \textrm{and} \qquad \Abs{S} \leq 24 \frac{n}{k} \mper  
\end{equation}

Also, from every hypergraph $H = (V,E,w)$, we can obtain a graph $G = (V,E',w')$ as follows.
We replace every $e \in E$ by a constant degree expander graph on $\Abs{e}$ vertices and 
set the weights of the new edges to be equal to $w(e)$.
By this construction, it is easy to see that the $k^{th}$ smallest eigenvalue of the normalized 
Laplacian of $G$ is at most $r\, \eig_k$. Therefore, using \cite{lrtv12,lot12} we get a set $S \subset V$ such that
\begin{equation}
\label{eq:hyper-sse-helper23}
\phi_H(S) \leq \phi_G(S) \leq \bigo{\sqrt{r\, \eig_k \, \log k  }} \qquad 
\textrm{and} \qquad \Abs{S} \leq 2 \frac{n}{k} \mper
\end{equation}
\prettyref{eq:hyper-sse-helper21} and \prettyref{eq:hyper-sse-helper23}
finish the proof of the theorem.
\end{proof}

\subsection{Hypergraph Multi-partition}
In this section we only give a sketch of the proof of \prettyref{thm:hyper-higher-cheeger},
as this theorem can be proven by essentially using \prettyref{thm:hyper-sse} and
the ideas studied in \cite{lm14}. 

\begin{theorem}[Restatement  of \prettyref{thm:hyper-higher-cheeger}]
For any hypergraph $H=(V,E,w)$  and any integer $k < \Abs{V}$, there exists
 $ck$  non-empty disjoint sets $S_1, \ldots, S_{ck} \subset V$ such that 
\[ \max_{i \in [ck]} \phi(S_i) \leq \bigo{ \kspvalue  \sqrt{\eig_k} } \mper \]
Moreover, for any $k$ disjoint non-empty sets $S_1, \ldots, S_k \subset V$
\[ \max_{i \in [k]} \phi(S_i) \geq \frac{\eig_k}{2} \mper  \]

\end{theorem}

\begin{proof}[Proof Sketch]
The first part of the theorem can be proved in a manner similar to \prettyref{thm:hyper-sse},
additionally using techniques from \cite{lm14}. As before, we will start
with the spectral embedding and then round it to get $k$-partition where each piece has
small expansion (\prettyref{alg:hyper-higher-cheeger}). 
Note that \prettyref{alg:hyper-higher-cheeger} 
can be viewed as a recursive application of \prettyref{alg:hyper-sse};
the algorithm computes a  ``small'' set having small expansion,
removes it and recurses on the remaining graph.


Note that \prettyref{step:hyper-algstep-helper2} of \prettyref{alg:hyper-higher-cheeger}
is somewhat different from \prettyref{step:hyper-algstep-helper1} of \prettyref{alg:hyper-sse}. 
Nevertheless, with some more work, we can bound the expansion of the set obtained at the end of 
\prettyref{step:hyper-algstep-helper3} 
by\footnote{Similar to \prettyref{eq:taudef}, $\tau' \defeq \kspvalue$. }
$\bigo{\tau' \sqrt{\eig_k}}$.
The proof of this bound on expansion follows from stronger forms of 
\prettyref{lem:hyper-sse-num} and \prettyref{lem:hyper-sse-denom}.

Once we have this, we can finish the proof of this theorem in a manner similar to 
\cite{lm14}. \cite{lm14} studied $k$-partitions in graphs and gave an alternate proof
of the graph version of this theorem (\prettyref{thm:graph-higher-cheeger}). They 
implicitly show how to use an algorithm for computing 
small-set expansion  to compute a $k$-partition in graphs where each piece has small expansion. 
A similar analysis can be used for hypergraphs as well.

\begin{mybox}
\begin{algorithm}~
Define $k' \defeq 10^5 k$.

\begin{enumerate}
\item Initialize $t := 1$ and $V_t := V$ and $C := \phi$.

\item {\bf Spectral Embedding}. We first construct a mapping of the vertices in $\R^k$ using the 
first $k$ eigenvectors. We map a vertex $i \in V $ to the vector $u_i$ defined as follows.
\[   u_i(l) = \frac{1}{\sqrt{d_i}} \eigvec_l(i) \mper \]

\item While $l \leq 10^5 k$
\begin{enumerate}

	\item {\bf Random Projection}. 
	Using \prettyref{lem:gen-orth-sep}, sample a random set $S$ from the set of vectors $\set{\tilde{u}_i}_{i \in V}$
	with $\beta = 99/100$ and $m = k$, and define the vector $X \in \R^n$ as follows. 
	\[  X(i) \defeq \begin{cases}  \norm{u_i}^2 & \textrm{if } \tilde{u}_i \in S \textrm{ and } i \in V_l  \\
		0 & \textrm{otherwise}   \end{cases} \mper \]
	\label{step:hyper-algstep-helper2}
	\item {\bf Sweep Cut}. Sort the entries of the vector $X$ in decreasing order and compute
	the set $S$ having the least expansion (See \prettyref{prop:hyper-1d}). If 
	\[ \sum_{i \in S} \norm{u_i}^2 > 3 
	\qquad \textrm{or} \qquad 
	\phi(S) > 10^5 \tau' \sqrt{\eig_k}	\]
	then discard $S$, else 
	$ C \gets C \cup \set{S} \quad \textrm{and} \quad V_{l+1} \gets V_l \setminus S \mper$
	\label{step:hyper-algstep-helper3}
	\item  $l \gets l+1$ and repeat.
\end{enumerate}

\item Output $C$.

\end{enumerate}
\label{alg:hyper-higher-cheeger}
\end{algorithm}
\end{mybox}

\end{proof}

\section{Algorithms for Computing Hypergraph Eigenvalues}
\label{sec:hyper-eigs-alg}

\subsection{An Exponential Time Algorithm for computing Eigenvalues}
\label{sec:hyper-eigs-exp}

\begin{theorem}
\label{thm:hyper-eigs-exp}
Given a hypergraph $H=(V,E,w)$, there exists an algorithm running in time $\tbigo{2^{rm}}$
which outputs all eigenvalues and eigenvectors of $M$. 
\end{theorem}

\begin{proof}
Let $X$ be an eigenvector $M$ with eigenvalue $\eig$.
Then
\[ \eig\, X = M(X) = A_X X \mper  \]
Therefore, $X$ is also an eigenvector of $A_X$. 
Therefore, the set of eigenvalues of $M$ is a subset of the set of eigenvalues 
of all the support matrices $\set{A_X : X \in \R^n}$.
Note that a support matrix $A_X$ is only determined by the pairs of vertices in each hyperedge
which have the largest difference in values under $X$.
Therefore, 
\[ \Abs{ \set{A_X : X \in \R^n} } \leq \paren{2^r}^m \mper   \]
Therefore, we can compute all the eigenvalues and eigenvectors of $M$ by enumerating 
over all $2^{rm}$ matrices.

\end{proof}

\subsection{Polynomial Time Approximation Algorithm for Computing Hypergraph Eigenvalues}
\label{sec:hyper-eigs-poly-alg}

Since $L$ is a non-linear operator, computing its eigenvalues exactly is intractable. 
In this section we give a  $\bigo{k \log r}$-approximation algorithm for $\eig_k$.

\begin{theorem}[Restatement of \prettyref{thm:hyper-eigs-alg}]
There exists a randomized polynomial time algorithm that, given a hypergraph $H = (V,E,w)$
and a parameter $k < \Abs{V}$, outputs $k$ orthonormal vectors $u_1, \ldots, u_k$ such that
\[ \ral{u_i} \leq \bigo{i \log r\, \eig_i}  \]
\whp
\end{theorem}

We will prove this theorem inductively.
We already know that $\eig_1 = 0$ and $u_1 = \mustat$.
Now, we assume that we have computed $k-1$ orthonormal vectors $u_1, \ldots, u_{k-1}$ 
such that   $ \ral{u_i} \leq \bigo{i  \log r\, \eig_i} $.
We will now show how to compute $u_k$. 

Our main idea is to show that there exists a unit vector
$X \in {\sf span}\set{\eigvec_1, \ldots, \eigvec_{k}}$ which is orthogonal to 
${\sf span} \set{u_1, \ldots, u_{k-1}}$. 
We will show that for such an $X$, $\ral{X} \leq k\, \eig_k$ (\prettyref{prop:hyper-CF}).
Then we give an $\sdp$ relaxation (\prettyref{sdp:eig-k}) and a rounding
algorithm (\prettyref{alg:hyper-eigs-rounding}, \prettyref{lem:hyper-eigs-rounding}) to compute an 
``approximate'' $X'$.

\begin{proposition}
\label{prop:hyper-CF}
Let $u_1, \ldots, u_{k-1}$ be arbitrary orthonormal vectors. Then
\[ \min_{X \perp u_1, \ldots, u_{k-1} } \ral{X} \leq k\, \eig_k  \mper \]

\end{proposition}

\begin{proof}

Consider subspaces $S_1 \defeq {\sf span}\set{u_1, \ldots, u_{k-1}}$
and $S_2 \defeq {\sf span}\set{\eigvec_1, \ldots, \eigvec_k} $.
Since ${\sf rank}(S_2) > {\sf rank}(S_1)$, there exists $X \in S_2$
such that $X \perp S_1$. 
We will now show that this $X$ satisfies $\ral{X} \leq \bigo{ k\, \eig_k}$, which will finish this proof.
Let $X = c_1 \eigvec_1 + \ldots + c_k \eigvec_k$ for scalars $c_i \in \R$ such that $\sum_i c_i^2 = 1$.

Recall that $\eig_k$ is defined as
\[ \eig_k \defeq \min_{Y \perp \eigvec_1, \ldots, \eigvec_{k-1}}  \frac{Y^T L_Y Y }{Y^T Y} \mper  \]
We can restate the definition of $\eig_k$ as follows,
\[ \eig_k = \min_{Y \perp \eigvec_1, \ldots, \eigvec_{k-1}} \max_{Z \in \R^n} \frac{Y^T L_Z Y }{Y^T Y} \mper  \]
Therefore,
\begin{equation}
\label{eq:eig-norm-helper}
\eig_k = \eigvec_k^T L_{\eigvec_k} \eigvec_k \geq \eigvec_k^T L_X \eigvec_k \qquad \forall X \in \R^n \mper
\end{equation}

The Laplacian matrix $L_X$, being positive semi-definite, has a Cholesky Decomposition into matrices $B_X$
such that $L_X = B_X B_X^T$. 

\begin{align*}
\ral{X} & = X^T L_X X = \sum_{i,j \in [k] } c_i c_j \eigvec_i^T B_X B_X^T \eigvec_j & \textrm{(Cholesky Decomposition of $L_X$ )} \\
 & \leq \sum_{i,j \in [k] } c_i c_j \norm{ B_X \eigvec_i }\cdot \norm{ B_X \eigvec_i } & \textrm{(Cauchy-Schwarz inequality)}  \\
 & = \sum_{i,j \in [k] } c_i c_j \sqrt{\eigvec_i^T L_X \eigvec_i } \sqrt{\eigvec_j^T L_X \eigvec_j } 
  \leq  \sum_{i,j \in [k] } c_i c_j \sqrt{\eig_1 \eig_j } & \textrm{(Using \prettyref{eq:eig-norm-helper})} \\
 & \leq \paren{\sum_i c_i}^2 \max_{i,j} \sqrt{\eig_i \eig_j} \leq k\, \eig_k \mper \\
\end{align*}

\end{proof}

Next we present an $\sdp$ relaxation (\prettyref{sdp:eig-k}) to compute the vector 
orthogonal $u_1, \ldots, u_{k-1}$ having the least Rayleigh quotient. 
The vector $\I$ is the relaxation of the $i^{th}$ coordinate of the 
vector $u_k$ that we are trying to compute.
The objective function of the $\sdp$ and \prettyref{eq:hyper-sdp-norm} seek to minimize
the Rayleigh quotient; \prettyref{prop:hyper-CF} shows that the objective value of this $\sdp$ is at most 
$k\, \eig_k$.
\prettyref{eq:hyper-sdp-orth} demands the solution be orthogonal to $u_1, \ldots, u_{k-1}$.

\begin{mybox}
\begin{SDP}
\label{sdp:eig-k}
\[ \sdpval \defeq \min \sum_{e \in E} w(e) \max_{i,j \in e} \norm{\I - \J}^2 \mper \] 
\subjectto
\begin{equation}
\label{eq:hyper-sdp-norm}
 \sum_{i \in V} \norm{\I}^2  =  1 
\end{equation}

\begin{equation}
\label{eq:hyper-sdp-orth}
\sum_{i \in V} u_l(i)\, \I  = 0 \qquad \forall l \in [k-1] 
\end{equation}

\end{SDP}
\end{mybox}

\begin{mybox}
\begin{algorithm}[Rounding Algorithm for Computing Eigenvalues]~
\begin{enumerate}
\item Solve \prettyref{sdp:eig-k} on the input hypergraph $H$ with the previously computed 
 $k-1$ vectors $u_1, \ldots, u_{k-1}$. 

\item Sample a random Gaussian vector $g \sim \cN(0,1)^n$. Set $X_i \defeq \inprod{\I,g}$.

\item Output $X/\norm{X}$.

\end{enumerate}

\label{alg:hyper-eigs-rounding}
\end{algorithm}
\end{mybox}

\begin{lemma}
\label{lem:hyper-eigs-rounding}
With constant probability \prettyref{alg:hyper-eigs-rounding} outputs a vector $u_k$ such that 
\begin{enumerate}
\item $u_k \perp u_l$ $\forall l \in [k-1]$.
\item $\ral{u_k} \leq  192 \log r\, \sdpval$.
\end{enumerate}
\end{lemma}

\begin{proof}

We first verify condition (1). For any $l \in [k-1]$, we using \prettyref{eq:hyper-sdp-orth}
\[ \inprod{X,u_l} = \sum_{i \in V} \inprod{\I,g} u_l(i) =  \inprod{\sum_{i \in V} u_l(i)\, \I,g}  = 0 \mper \]

We now prove condition (2). To bound $\ral{X}$
we need an upper bound on the numerator and a lower bound on the numerator of the $\ral{\cdot}$
expression.
For the sake of brevity let $L$ denote $L_{X}$. Then
\begin{align*}
\Ex{X^T L X} & \leq \sum_{e \in E} w(e)\Ex{  \max_{i,j \in e} (X_i - X_j)^2 } 
\leq 4 \log r\, \sum_{e \in E} w(e) \max_{i,j \in e} \norm{\I - \J}^2 & \textrm{(Using \prettyref{fact:appGauss})}  \\
& = 4 \log r\, \sdpval \\
\end{align*}
Therefore, by Markov's Inequality,
\begin{equation}
\label{eq:sdp-hyper-num}
\Pr{ X^T L X  \leq 96 \log r\, \sdpval   } \geq 1 - \frac{1}{24} \mper
\end{equation}

For the denominator, using linearity of expectation, we get 
\[ \Ex{ \sum_{i\in V} X_i^2} = \sum_i \Ex{ \inprod{\I,g}^2} = \sum_i \norm{\I}^2 = 1 \qquad \textrm{(Using \prettyref{eq:hyper-sdp-norm})} \mper  \]

Now applying
\prettyref{lem:squaregaussian} to the denominator we conclude

\begin{equation} 
\label{eq:sdp-hyper-denom}
\Pr{ \sum_i X_i^2 \geq \frac{1}{2} } \geq \frac{1}{12} \mper  
\end{equation}

Using Union-bound on \prettyref{eq:sdp-hyper-num}  and \prettyref{eq:sdp-hyper-denom} we get that 
\[ \Pr{ \ral{X} \leq 192\, \sdpval }  \geq \frac{1}{24} .\]

\end{proof}

\begin{lemma} 
\label{lem:squaregaussian}
Let $z_1,\ldots, z_m$ be standard normal random variables (not necessarily independent) 
such $ \Ex{ \sum_i z_i^2} = 1$ then
\[ \Pr{\sum_i z_i^2 \geq \frac{1}{2}} \geq \frac{1}{12} \mper \]
\end{lemma}
\begin{proof}
We will bound the variance of the random variable $R = \sum_i z_i^2$
as follows,
\begin{align*}
\Ex{R^2} & = \sum_{i,j} \Ex{z_i^2 z_j^2} 
 \leq \sum_{i,j} \paren{ \Ex{z_i^4} }^{\frac{1}{2}} \paren{\Ex{z_j^4} }^{\frac{1}{2}} \\
& = \sum_{i,j} 3 \Ex{z_i^2} \Ex{z_j^2}  & \textrm{ (Using } \Ex{g^4} = 3 \paren{\Ex{g^2}}^2 \textrm{ for gaussians )} \\
& = 3\paren{ \sum_{i}  \Ex{z_i^2} }^2 = 3  
\end{align*}
By the Paley-Zygmund inequality,
\[ \Pr{ R \geq \frac{1}{2} \Ex{R}} \geq \left(\frac{1}{2}\right)^2 \frac{\Ex{R}^2}{\Ex{R^2}} \geq \frac{1}{12} \mper \]

\end{proof}

We now have all the ingredients to prove \prettyref{thm:hyper-eigs-alg}.

\begin{proof}[Proof of \prettyref{thm:hyper-eigs-alg}]
We will prove this theorem inductively. For the basis of induction, we have the first eigenvector 
$u_1 = \eigvec_1 = \one/\sqrt{n}$. We assume that we have computed $u_1,\ldots, u_{k-1}$ satisfying 
$\ral{u_i} \leq \bigo{i \log r\, \eig_i }$. We now show how to compute $u_k$. 

\prettyref{prop:hyper-CF} implies that for \prettyref{sdp:eig-k},
\[ \sdpval \leq k\, \gamma_k \mper \]
Therefore, from \prettyref{lem:hyper-eigs-rounding}, we get that \prettyref{alg:hyper-eigs-rounding}
will output a unit vector which is orthogonal to all $u_i$ for $i \in [k-1]$ and 
\[ \ral{u_k} \leq 192\, k \log r\, \eig_k \mper \]

\end{proof}

\subsection{Approximation Algorithm for Hypergraph Expansion}
\label{sec:hyper-sparsest}

Here we show how to use our algorithm for computing hypergraph eigenvalues (\prettyref{thm:hyper-eigs-alg}) 
to compute an approximation for hypergraph expansion.

\begin{corollary}[Formal statement of \prettyref{cor:hyper-sparsest-informal}]
\label{cor:hyper-sparsest-formal}
There exists a randomized polynomial time algorithm that given a hypergraph $H = (V,E,w)$,
outputs a set $S \subset V$ such that 
\[ \phi(S) = \bigo{\sqrt{ \frac{1}{\rmin} \phi_H \log r}}  \]
\whp
\end{corollary}

\begin{proof}
\prettyref{thm:hyper-eigs-alg} gives a randomized polynomial time algorithm to
compute a vector $X \in \R^n$ such that $\ral{X} \leq \bigo{\eig_2 \log r}$. 
Invoking \prettyref{prop:hyper-sweep-rounding} with this vector $X$, we get a set 
$S \subset V$ such that 
\[ \phi(S) = \bigo{\sqrt{ \frac{1}{\rmin} \ral{X} } } =  \bigo{\sqrt{ \frac{1}{\rmin} \eig_2 \log r } }
=  \bigo{\sqrt{ \frac{1}{\rmin} \phi_H \log r } } \mper \]
Here the last inequality uses $ \eig_2/2 \leq \phi_H  $ from \prettyref{thm:hyper-cheeger}.
\end{proof}

\section{Sparsest Cut with General Demands}
\label{sec:hyper-sparsestcut}

In this section we study polynomial time (multiplicative) approximation algorithms
for hypergraph expansion problems. 
We study the Sparsest Cut with General Demands problem and given an approximation algorithm for it
(\prettyref{thm:hyper-sparsest-nonuniform}).

\begin{theorem}[Restatement of \prettyref{thm:hyper-sparsest-nonuniform}]
There exists a randomized polynomial time algorithm that given 
an instance of the hypergraph Sparsest Cut problem with general demands $H=(V,E,D)$, 
outputs a set $S \subset V$ such that 
\[ \Phi(S) \leq \bigo{\sqrt{\log k \log r} \log \log k } \Phi_H   \]
\whp, where $k = \Abs{D}$  and $r = \max_{e \in E} \Abs{e}$.
\end{theorem}

\begin{proof}
We prove this theorem by giving an \sdp relaxation for this problem (\prettyref{sdp:hgensc})
and a rounding algorithm for it (\prettyref{alg:hgensc}).
We introduce a variable $\U$ for each vertex $u \in V$. 
Ideally, we would want all vectors $\U$ to be in the set $\set{0,1}$ so that we can identify the cut,
in which case $\max_{u,v \in e} \norm{\U - \V}^2 $ will indicate whether the edge $e$ is cut 
or not.
Therefore, our objective function will be $\sum_{e \in E} w(e) \max_{u,v \in e} \norm{\U - \V}^2$.
Next, we add \prettyref{eq:hyper-gensc-1} as a scaling constraint. 
Finally, we add $\ell_2^2$ triangle inequality constraints between all triplets of vertices
\prettyref{eq:hyper-gensc-3}, 
as all integral solutions of the relaxation will trivially satisfy this. 
Therefore \prettyref{sdp:hgensc} is a relaxation of $\Phi_H$. 

\begin{mybox}
\begin{SDP}
\label{sdp:hgensc}
\[  \min \sum_{e \in E} w(e) \max_{u,v \in e} \norm{ \U - \V  }^2   \]
\subjectto 
\begin{equation}
\label{eq:hyper-gensc-1}
\sum_{u,v \in D} \norm{\U - \V}^2   = 1  
\end{equation}
\begin{equation}
\label{eq:hyper-gensc-3}
\norm{\U - \V}^2 + \norm{\V - \W}^2 \geq  \norm{\U - \W}^2 \qquad \forall u,v,w \in V 
\end{equation}
\end{SDP}
\end{mybox}

Our main ingredient is the following theorem due to \cite{aln05}. 
\begin{theorem}[\cite{aln05}]
\label{thm:aln}
Let $(X,d)$ be an arbitrary metric space, and let $D \subset X$ be any $k$-point subset.
If the space $(D,d)$ is a metric of the negative type, then there exists a $1$-Lipschitz map
$f : X \to L_2$ such that the map $f|_D : D \to L_2$ has distortion $\bigo{\sqrt{\log k} \log k \log k}$.
\end{theorem}

\begin{mybox}
\begin{algorithm}~
\label{alg:hgensc}
\begin{enumerate}
\item Solve \prettyref{sdp:hgensc}.

\item Compute the map $f : V \to \R^n$ using \prettyref{thm:aln}.

\item Pick $g \sim \cN(0,1)^n$ and define $x_i \defeq \inprod{g,f(v_i)}$ for each $v_i \in V$.

\item Arrange the vertices of $V$ as $v_1, \ldots, v_n$
such that $x_j \leq x_{j+1}$  for each  $1\leq j \leq n-1$. 
Output the sparsest cut of the form 
\[ \paren{ \set{v_1, \ldots, v_i} , \set{v_{i+1}, \ldots, v_n }  } \mper  \]
\label{step:hgensc-four}

\end{enumerate}
\end{algorithm}
\end{mybox}

W.l.o.g. we may assume that the map $f$ is such that $f|_D$ has the least distortion among all $1$-Lipschitz
maps $f:V \to L_2$ (\cite{aln05} give a polynomial time algorithm to compute such a map.)
For the sake of brevity, let $\Lambda = \bigo{\sqrt{\log k} \log \log k}$ denote the distortion factor
guaranteed in \prettyref{thm:aln}.
Since \prettyref{sdp:hgensc} is a relaxation of $\Phi_H$, we also get that objective value of the 
\sdp is at most $\Phi_H$.

Now, using \prettyref{fact:appGauss}, we get 
\[ \Ex{\max_{u,v \in e} \Abs{x_u - x_v}} \leq 2 \sqrt{\log r} \max_{u,v \in e} \norm{f(u) - f(v) } \mper  \]
Therefore, using Markov's inequality 
\begin{equation}
\label{eq:hgensc-help-1}
 \Pr{ \sum_e w(e) \max_{u,v \in e} \Abs{x_u - x_v} \leq 48 \sqrt{\log r}\, \Phi_H }  \geq 1 - \frac{1}{24} \mper     
\end{equation}

Next,
\[ \Ex{ \sum_{u,v \in D} \Abs{x_u - x_v} }  = \sum_{u,v \in D} \norm{ f(u) -f(v)} 
\geq \frac{1}{\Lambda} \sum_{u,v \in D} \norm{\U - \V}^2  = \frac{1}{\Lambda} \mper \]
Here the last equality follows from \prettyref{eq:hyper-gensc-1}.
Now, using \prettyref{lem:squaregaussian}, we get 
\begin{equation}
\label{eq:hgensc-help-2}
\Pr{ \sum_{u,v \in D} \Abs{x_u - x_v} \geq \frac{1}{2 \Lambda}  } \geq \frac{1}{12}
\end{equation}

Using, \prettyref{eq:hgensc-help-1} and \prettyref{eq:hgensc-help-2} we get that 
with probability at least $1/24$
\[ \frac{\sum_e w(e) \max_{u,v \in e} \Abs{x_u - x_v}}{ \sum_{u,v \in D} \Abs{x_u - x_v} }
\leq 96 \sqrt{\log r}\, \Lambda\, \Phi_H \mper \]

Using an analysis similar to \prettyref{prop:hyper-1d}, we get that the set output in
\prettyref{step:hgensc-four} satisfies 
\[ \Phi(S) \leq 96 \sqrt{\log r}\, \Lambda \Phi_H = \bigo{\sqrt{\log k \log r}  \log \log k } \Phi_H \mper  \]

\end{proof}

\section{Lower Bound for Computing Hypergraph Eigenvalues}
\label{sec:hyper-eigs-lower}
We now use \prettyref{thm:hyper-vert-exp} to prove \prettyref{thm:hyper-eigs-lower-informal}
and \prettyref{thm:hyper-expansion-hardness-informal}.
We begin by describing the \SSEH proposed by Raghavendra and Steurer \cite{rs10}.
\begin{hypothesis}[\SSEH, \cite{rs10}]
\label{hyp:sse}
  For every constant $\eta > 0$, there exists sufficiently small $\delta>0$
  such that given a graph $G$ it is NP-hard to distinguish the cases,
  \begin{description}\item[\yes:] 
    there exists a vertex set $S$ with volume $\mu(S)=\delta$ and expansion 
    $\phi(S)\le \eta$,
  \item[\no:]  all vertex sets $S$  with volume  $\mu(S)=\delta$ have expansion
     $\phi(S)\ge 1-\eta$.
  \end{description}
\end{hypothesis} 

\paragraph{\SSEH.}
Apart from being a natural optimization problem, the small-set expansion problem is closely tied to the Unique
Games Conjecture.  Recent work by Raghavendra-Steurer
\cite{rs10} established the reduction from the small-set expansion problem to the well known Unique
Games problem, thereby showing that \SSEH implies the Unique Games Conjecture.  
We refer the reader to \cite{rst12} for a comprehensive discussion on the 
implications of \SSEH.

\begin{theorem}[Formal statement of \prettyref{thm:hyper-expansion-hardness-informal}]
\label{thm:hyper-expansion-hardness}
For every $\eta > 0$, there exists an absolute constant $C$ such that $\forall \e>0 $ it is \sse-hard to distinguish 
between the following two cases for a given hypergraph $H = (V,E,w)$ with maximum hyperedge size $r \geq 100/\e$
and $\rmin \geq c_1 r$ (for some absolute constant $c_1$).
\begin{description}
\item[\yes] : There exists a set $S \subset V$ such that 
	\[ \phi_H(S) \leq \e  \]
\item[\no] : For all sets $S \subset V$, 
	\[  \phi_H(S) \geq  \min \set{10^{-10}, C \sqrt{\frac{c_1}{r} \e \log r}} - \eta \]
\end{description}
\end{theorem}

\begin{proof}
We will use the following theorem due to \cite{lrv13}.
\begin{theorem*}[\cite{lrv13}]
For every $\eta > 0$, there exists an absolute constant $C_1$ such that $\forall \e>0 $ it is \sse-hard to distinguish 
between the following two cases for a given graph $G = (V,E,w)$ with maximum degree $d \geq 100/\e$ and minimum 
degree $c_1 d$ (for some absolute constant $c_1$).
\begin{description}
\item[\yes] : There exists a set $S \subset V$ of size $\Abs{S} \leq \Abs{V}/2$ such that 
	\[ \phiv(S) \leq \e  \]
\item[\no] : For all sets $S \subset V$, 
	\[  \phiv(S) \geq  \min \set{10^{-10}, C_2 \sqrt{\e \log d}} - \eta \]
\end{description}
\end{theorem*}
Using this and the reduction from vertex expansion in graphs to hypergraph expansion 
(\prettyref{thm:hyper-vert-exp}), finishes the proof of this theorem.

\end{proof}

\begin{theorem}[Formal statement of \prettyref{thm:hyper-eigs-lower-informal}]
\label{thm:hyper-eigs-lower}
For every $\eta > 0$, there exists an absolute constant $C$ such that $\forall \e>0 $ it is \sse-hard to distinguish 
between the following two cases for a given hypergraph $H = (V,E,w)$ with maximum hyperedge size $r \geq 100/\e$
and $\rmin \geq c_1 r$ (for some absolute constant $c_1$).
\begin{description}
\item[\yes] : There exists an  $X \in \R^n$ such that $\inprod{X,\mustat} = 0$ and  
	\[ \ral{X} \leq \e  \]
\item[\no] : For all $X \in \R^n$ such that $\inprod{X,\mustat} = 0$, 
	\[  \ral{X} \geq  \min \set{10^{-10}, C \e \log r} - \eta \]
\end{description}
\end{theorem}

\begin{proof}

For the \yes case, if there exists a set $S \subset V$ such that $\phi_H(S) \leq \e/2$,
then for the vector 
\[ X \defeq  \chi_S - \frac{\inprod{\chi_S,\mustat}}{\norm{\mustat}^2}   \mustat \qquad 
\textrm{we have } \qquad \ral{X} \leq \e \mper \]
 
For the \no case, \prettyref{prop:hyper-sweep-rounding} says that given a vector $X \in \R^n$ such that 
$\inprod{X,\mustat} = 0$, 
we can find a set $S \subset V$ such that $\phi(S) \leq 2 \sqrt{ \ral{X}/ \rmin }$. 

This
combined with \prettyref{thm:hyper-expansion-hardness} finishes the proof of this theorem..

\end{proof}

\subsection{Nonexistence of Linear Hypergraph Operators}

\begin{theorem}[Restatement of \prettyref{thm:hyper-nonlinear}]
Given a hypergraph $H=(V,E,w)$, assuming the \sse~ hypothesis, there exists no
polynomial time algorithm to compute a matrix $A \in \R^{V \times V}$, such that 
\[  c_1 \lambda \leq \phi_H \leq c_2 \sqrt{\lambda}   \]
where $\lambda$ is any polynomial time computable function of the eigenvalues of $A$
and $c_1, c_2 \in \R^+$ are absolute constants.
\end{theorem}

\begin{proof}
For the sake of contradiction, suppose there existed a polynomial time algorithm to compute such 
a matrix $A$ and there existed a polynomial time algorithm to compute a $\lambda$ from the eigenvalues
of $A$ such that \[  c_1 \lambda \leq \phi_H \leq c_2 \sqrt{\lambda} \mper \]
Then this would yield a $\bigo{\sqrt{\OPT}}$ approximation for $\phi_H$.
But \prettyref{thm:hyper-expansion-hardness} says that this is not possible assuming the
\sse~ hypothesis. Therefore, no such polynomial time algorithm to compute such a matrix exists.

\end{proof}

\section{Vertex Expansion in Graphs and Hypergraph Expansion}
\label{sec:vert-exp}
Bobkov \etal~ defined a Poincair\'e-type  functional graph parameter called $\linf$ as follows.
\begin{definition}[\cite{bht00}]
For an un-weighted graph $G = (V,E)$, $\linf$ is defined as follows.
\[ \linf \defeq \min_{ X \in \R^n} 
\frac{ \sum_{u \in V} \max_{v \sim u} \paren{X_u - X_v}^2 }{\sum_{u \in V} X_u^2 - \frac{1}{n} \paren{\sum_{u \in V}X_u }^2 } \mper \]
\end{definition}
They showed that $\linf$ captures the vertex expansion of a graph in a {\em Cheeger-like} manner.

\begin{theorem*}[\cite{bht00}]
For an un-weighted graph $G = (V,E)$,
\[  \frac{\linf}{2} \leq  \phiv_G \leq \sqrt{2 \linf } \mper  \]
\end{theorem*}
The computation of $\linf$ is not known to be tractable.
For graphs having maximum vertex degree $d$, 
\cite{lrv13} gave a $\bigo{\log d}$-approximation algorithm for computing $\linf$, and 
showed that there exists an absolute constant $C$ such that is $\sse$-hard 
to get better than a $C \log d$ approximation to $\linf$.

We first show that $\eig_2$ of the hypergraph obtained from $G$ via the reduction from
vertex expansion in graphs to hypergraph expansion, is within a factor four of $\linf$.
\begin{theorem}
\label{thm:linf-eig2}
Let $G = (V,E)$ be a un-weighted $d$-regular graph, and let $H = (V,E')$ be the hypergraph obtained
from $G$ using \prettyref{thm:hyper-vert-exp}. Then
\[  \frac{\eig_2}{4} \leq  \frac{\linf}{d} \leq \eig_2 \mper    \]
\end{theorem}

\begin{proof}
Using \prettyref{thm:hyper-vert-exp}, $\eig_2$ of $H$ can be reformulated as
\[  \eig_2 = \min_{X \perp \one} 
	\frac{ \sum_{u \in V}  \max_{  i,j \in \paren{ \set{u} \cup N(u)}} \paren{X_i - X_j}^2 }{d \sum_{u \in V} X_u^2} \mper    \]
Therefore, it follows that $  \linf/d \leq \eig_2$.
Next, using $(x + y)^2 \leq 4 \max \set{x^2,y^2}$ for any $x,y \in \R$, we get  
\[  \eig_2 = \min_{X \perp \one} 
	\frac{ \sum_{u \in V}  \max_{  i,j \in \paren{ \set{u} \cup N(u)}} \paren{X_i - X_j}^2  }{d \sum_{u \in V} X_u^2} 
\leq \min_{X \perp \one} 
	\frac{ \sum_{u \in V} 4 \max_{ v \sim u} \paren{X_i - X_j}^2 }{d \sum_{u \in V} X_u^2}
	= 4 \frac{\linf}{d}  \mper    \]
\end{proof}
\prettyref{thm:linf-eig2} shows that $\linf$ of a graph $G$ is an ``approximate eigenvalue'' of
the hypergraph markov operator for the hypergraph obtained from $G$ using the reduction from
vertex expansion in graphs to hypergraph expansion (\prettyref{thm:hyper-vert-exp}).

We now define a markov operator for graphs, similar to \prettyref{def:hyper-markov},
for which $(1 - \linf)$ is the second largest eigenvalue.

\begin{mybox}
\begin{definition}[The Vertex Expansion Markov Operator]~

 Given a vector $X \in \R^n$, $\mvert(X)$ is computed as follows.

	\begin{enumerate}
	\item 	For each vertex $u \in V$, let $j_u := {\sf argmax}_{v \sim u} \Abs{X_u - X_v}$, breaking 
	ties randomly (See \prettyref{rem:hyper-walk-ties}). 	
	\item We now construct the weighted graph $G_X$ on the vertex set $V$ as follows.
		We add edges $\set{ \set{u,j_u} : u \in V}$ having weight $w(\set{u,j_u}) := 1/d$ to $G_X$. 
		Next, to each vertex $v$ we add self-loops of sufficient weight such that its weighted degree in $G_X$ is equal 
		to $1$.  	
	\item	We define $A_X$ to be the (weighted) adjacency matrix of $G_X$.
	
	\end{enumerate}
Then, 
	\[ \mvert(X) \defeq A_X X \mper \]

\label{def:vert-markov}
\end{definition}
\end{mybox}

\begin{theorem}[Restatement of \prettyref{thm:linf-eig}]
For a graph $G$, $\linf$ is the second smallest eigenvalue of $\lapvert \defeq I - \mvert$.
\end{theorem}

The proof of \prettyref{thm:linf-eig} is similar to the proof of \prettyref{thm:hyper-2nd},
and hence is omitted.

\section{Conclusion and Open Problems}
In this paper we introduced a new hypergraph Markov operator as a generalization of the 
random-walk operator on graphs.
We proved many spectral properties about this operator and hypergraphs, 
which can be viewed as generalizations of the analogous properties of graphs.

\paragraph{Open Problems.}
Many open problems remain.
In short, we ask what properties of graphs and random walks generalize to
hypergraphs and this Markov operator?
More concretely, we present a few exciting (to us) open problems.

\begin{problem}
Given a hypergraph $H = (V,E,w)$ and a parameter $k$, do there exists $k$ non-empty
disjoint subsets $S_1, \ldots, S_k$ of $V$ such that 
\[ \max_{i} \phi(S_i) \leq \bigo{ \sqrt{\eig_k \log k \log r}  }  ?  \]
\end{problem}

\begin{problem}
Given a hypergraph $H = (V,E,w)$ and a parameter $k$, is there a randomized polynomial
time algorithm to obtain a $\bigo{ {\sf polylog}\, k\, {\sf polylog}\, r}$-approximation to $\eig_k$ ?
\end{problem}

\begin{problem}
Is there a $\bigo{\sqrt{\log k}\log \log k}$-approximation algorithm for sparsest cut with general demands
in hypergraphs ?
\end{problem}

\paragraph{Acknowledgements.}
The dispersion process associated with our markov operator was suggested to us by Prasad Raghavendra
in the context of understanding vertex expansion in graphs, and was the
starting point of this project.
The author would like to thank Ravi Kannan,  
Konstantin Makarychev, Yury Makarychev, 
Yuval Peres, 
Prasad Raghavendra, Nikhil Srivastava, Piyush Srivastava, Prasad Tetali, Santosh Vempala, 
David Wilson and Yi Wu for helpful discussions.

\bibliography{bibfile}
\bibliographystyle{amsalpha}
\appendix

\section{Hypergraph Tensor Forms}
\label{app:hyper-tensor}
Let $A$ be an $r$-tensor. 
For any suitable norm $\normb{\cdot}$, e.g. $\norm{.}^2_2$, $\norm{.}^r_r$, 
we define tensor eigenvalues as follows. 

\begin{definition}
We define $\lambda_1$, the largest eigenvalue of a tensor $A$  
as follows.
\[ \lambda_1 \defeq \max_{X \in \R^n} \frac{ \sum_{i_1, i_2, \ldots, i_r} A_{i_1 i_2 \ldots i_r} X_{i_1} X_{i_2} \ldots X_{i_r} }{\normb{X}}
\qquad  v_1 \defeq \argmax_{X \in \R^n} \frac{\sum_{i_1, i_2, \ldots, i_r} A_{i_1 i_2 \ldots i_r} X_{i_1} X_{i_2} \ldots X_{i_r}  }{\normb{X}}
  \]
We inductively define successive eigenvalues $\lambda_2 \geq \lambda _3 \geq \ldots$ as follows.
\[ \lambda_k \defeq \max_{X \perp \set{v_1, \ldots, v_{k-1}} } 
	\frac{ \sum_{i_1, i_2, \ldots, i_r} A_{i_1 i_2 \ldots i_r} X_{i_1} X_{i_2} \ldots X_{i_r} }{\normb{X}}
\qquad  v_k \defeq \argmax_{x \perp \set{v_1, \ldots, v_{k-1}} }
	\frac{ \sum_{i_1, i_2, \ldots, i_r} A_{i_1 i_2 \ldots i_r} X_{i_1} X_{i_2} \ldots X_{i_r}  }{\normb{X}}   \]

\end{definition}

Informally, the Cheeger's Inequality states that a graph has a sparse cut if and only if the 
gap between the two largest eigenvalues of the adjacency matrix is small; in particular, a graph
is disconnected if any only if its top two eigenvalues are equal.  In the case of the hypergraph tensors,
we show that there exist hypergraphs having 
no gap between many top eigenvalues while still being connected. This shows that the tensor 
eigenvalues are not relatable to expansion in a Cheeger-like manner.

\begin{proposition}
For any $k \in \N$, there exist connected hypergraphs such that $\lambda_1 = \ldots = \lambda_k$.
\end{proposition}

\begin{proof}
Let $r = 2^w$ for some $w \in \Z^+$. Let $H_1$ be a large enough complete $r$-uniform hypergraph. 
We construct $H_2$ from two copies of $H_1$, say $A$ and $B$, as follows. 
Let $a \in E(A)$ and $b \in E(B)$ be any two hyperedges.   
Let $a_1 \subset a $ (resp. $b_1 \subset b$) be a set of any $r/2$ vertices. 
We are now ready to define $H_2$.
\[ H_2 \defeq \left( V(H_1) \cup V(H_2), (E(H_1) \setminus \set{a}) \cup (E(H_2) \setminus \set{b}) 
\cup \set{ (a_1 \cup b_1), (a_2 \cup b_2)} \right)  \] 
Similarly, one can recursively define $H_i$ by joining two copies of $H_{i-1}$ (this can be done as  
long as $r > 2^{2i}$). The construction of $H_k$ can be viewed as a {\em hypercube of hypergraphs}. 

Let $A_H$ be the tensor form of hypergraph $H$.
For $H_2$, it is easily verified that $v_1 = \one$.
Let $X$ be the vector which has $+1$ on the vertices 
corresponding to $A$ and the $-1$ on the vertices corresponding to $B$. 
By construction, for any hyperedge $\set{i_1, \ldots, i_r} \in E$
\[ X_{i_1} \ldots X_{i_r} = 1    \]
and therefore, 
\[  \frac{ \sum_{i_1, i_2, \ldots, i_r} A_{i_1 i_2 \ldots i_r} X_{i_1} X_{i_2} \ldots X_{i_r} }{\normb{X}} = \lambda_1 \mper  \]
Since $\inprod{X,\one} = 0$, we get $\lambda_2 = \lambda_1$ and $v_2 = X$.
Similarly, one can show that $\lambda_1 = \ldots = \lambda_k$ for $H_k$. 
This is in sharp contrast to the fact that $H_k$ is, by construction,  a connected hypergraph. 
\end{proof}

\section{Omitted Proofs}
\label{sec:omit-proofs}

\begin{proposition}
\label{prop:matrix-ineq-1}

Let $A$ be a  symmetric $n\times n$ matrix with eigenvalues $\alpha_1, \ldots, \alpha_n$
and corresponding eigenvectors $v_1, \ldots, v_n$ such that $A \succeq 0$.
Then, for any $X \in \R^n$ 
\[  \frac{X^T (I - A) X }{X^T X} - \frac{X^T A^T (I - A) A X}{X^T A^T A X} = 
2 \frac{ \sum_{i,j} c_i^2 c_j^2 (\alpha_i - \alpha_j)^2 (\alpha_i + \alpha_j) }{\sum_i c_i^2 \sum_i c_i^2 \alpha_i^2 } \geq 0 \]
where $X = \sum_i c_i v_i$.
\end{proposition}

\begin{proof}
We first note that the eigenvectors of $I-A$ are also $v_1, \ldots, v_n$ with 
$1-\alpha_1, \ldots, 1 -\alpha_n$ being the corresponding eigenvalues.
\begin{align*}
\frac{X^T (I - A) X }{X^T X} - \frac{X^T A^T (I - A) A X}{X^T A^T A X}  
& = \frac{\sum_i c_i^2 (1- \alpha_i) }{\sum_i c_i^2} - 
\frac{\sum_i c_i^2 \alpha_i^2 (1 - \alpha_i) }{\sum_i c_i^2 \alpha_i^2} \\
& = 2 \frac{\sum_{i \neq j} c_i^2 c_j^2 \left( (1 - \alpha_i) \alpha_j^2 + (1 - \alpha_j) \alpha_i^2 -
(1 - \alpha_i) \alpha_i^2 - (1 - \alpha_j) \alpha_j^2 \right) }
{\sum_i c_i^2 \sum_i c_i^2 \alpha_i^2 } \\
& = 2 \frac{ \sum_{i,j} c_i^2 c_j^2 (\alpha_i - \alpha_j)^2 (\alpha_i + \alpha_j) }{\sum_i c_i^2 \sum_i c_i^2 \alpha_i^2 } 
\end{align*}

\end{proof}

\begin{lemma}
\label{lem:discrete-function-lebesgue}
Let $f: [0,1] \to \set{1, 2, \ldots, k}$ be any discrete function. Then there exists an
interval $(a,b) \subset [0,1]$, $a \neq b$, such that for some $\alpha \in \set{1, 2, \ldots, k}$
\[ f(x) = \alpha \qquad \forall x \in (a,b) \mper  \] 

\end{lemma}

\begin{proof}
Let $\upsilon(\cdot)$ denote the standard Lebesgue measure on the real line.
Then since $f$ is a discrete function on $[0,1]$ we have
\[ \sum_{i = 1}^k \upsilon \paren{ f^{-1}(i)  }  = 1 \mper \]
Then, for some $\alpha \in \set{1, 2, \ldots, k}$
\[    \upsilon \paren{ f^{-1}(\alpha) }  \geq \frac{1}{k} \mper \]
Therefore, there is some interval $(a,b) \subset f^{-1}(\alpha)$ such that 
\[ \upsilon\paren{ (a,b) } > 0 \mper  \]
This finishes the proof of the lemma.
\end{proof}

\begin{fact}
\label{fact:appGauss}
Let $Y_1, Y_2, \ldots, Y_d$ be $d$ standard normal random variables. Let $Y$ be the random
variable defined as $Y \defeq \max \set{Y_i | i\in [d]}$. Then
\[ \Ex{Y^2} \leq 4 \log d \qquad \textrm{ and } \qquad \Ex{Y} \leq 2 \sqrt{\log d} \mper \]
\end{fact}

\begin{proof}

For any $Z_1, \ldots, Z_d \in \R$ and any $p \in \Z^+$, we have $\max_i \Abs{Z_i} \leq (\sum_i Z_i^p)^{\frac{1}{p}}$.
Now $Y^2 = (\max_i X_i)^2 \leq \max_i X_i^2$.

\begin{eqnarray*}
\Ex{Y^2} & \leq & \Ex{ \left( \sum_i X_i^{2p} \right)^{\frac{1}{p}} }
  \leq  \left(\Ex{ \sum_i X_i^{2p} } \right)^{\frac{1}{p}} \quad \textrm{ ( Jensen's Inequality )} \\
 & \leq & \left( \sum_i \left( \Ex{X_i^2} \right) \frac{(2p)! }{(p)! 2^{p} } \right)^{\frac{1}{p}}
  \leq  2 p d^{\frac{1}{p}} \quad \textrm{(using $(2p)!/p! \leq (2p)^{p} $ )} \\
\end{eqnarray*}

Picking $p = \log d$ gives $\Ex{Y^2} \leq 2e \log d$.

Therefore $\Ex{Y} \leq \sqrt{\Ex{Y^2}} \leq \sqrt{2 e \log d} $.

\end{proof}

\begin{lemma}
\label{lem:2vecs}
For any two non zero vectors $u_i$ and $u_j$, if $\tilde{u_i} =
u_i/\norm{u_i}$ and $\tilde{u_j} = u_j/\norm{u_j}$ then
\[ \norm{\tilde{u_i} - \tilde{u_j}}  \sqrt{ \norm{u_i}^2 + \norm{u_j}^2 }  \leq 2 \norm{u_i - u_j}  \mper  \]
\end{lemma}
\begin{proof}
Note that $2\norm{u_i} \norm{u_j}  \leq \norm{u_i}^2 + \norm{u_j}^2 $.
Hence,
\begin{eqnarray*}
\norm{\tilde{u_i} - \tilde{u_j}}^2 ( \norm{u_i}^2 + \norm{u_j}^2 )  & =  &
(2 - 2 \inprod{\tilde{u_i},\tilde{u_j}}) ( \norm{u_i}^2 + \norm{u_j}^2 ) \\
& \leq & 2 ( \norm{u_i}^2 + \norm{u_j}^2  -  (\norm{u_i}^2 + \norm{u_j}^2)  \inprod{\tilde{u_i},\tilde{u_j}}  ) \\
\end{eqnarray*}

If $\inprod{\tilde{u_i},\tilde{u_j}} \geq 0$, then
\[ \norm{\tilde{u_i} - \tilde{u_j}}^2 ( \norm{u_i}^2 + \norm{u_j}^2 )  \leq
2 ( \norm{u_i}^2 + \norm{u_j}^2  - 2 \norm{u_i} \norm{u_j}  \inprod{\tilde{u_i},\tilde{u_j}}  ) \leq 2 \norm{u_i - u_j}^2 \]

Else if $ \inprod{\tilde{u_i}, \tilde{u_j} } < 0$, then
\[ \norm{\tilde{u_i} - \tilde{u_j}}^2 ( \norm{u_i}^2 + \norm{u_j}^2 )  \leq
  4 ( \norm{u_i}^2 + \norm{u_j}^2  - 2 \norm{u_i} \norm{u_j}  \inprod{\tilde{u_i},\tilde{u_j}}  )  \leq 4 \norm{u_i - u_j}^2
\]
\end{proof}

\begin{theorem}[Folklore]
\label{thm:graph-mixing-lower}
Given a graph $G = (V,E,w)$, let $\lambda_2$ be the second smallest eigenvalue of the 
normalized Laplacian of $G$. Then there exists a vertex $i \in V$, such that 
\[ \tmix{e_i} \geq \frac{\log 1/\delta}{\lambda_2} \mper \]
\end{theorem}

\begin{proof}

Let $P$ be the random walk matrix of $G$, and let $\alpha_2$ be its second largest eigenvalue. 
Then, $\lambda_2 = 1 - \alpha_2$ (Folklore).
Let $X$ be the second eigenvector of $P$. Then for any $j \in V$
\[ \Abs{\alpha_2^t X(j)} =  {P^t X (j) } = \Abs{ \sum_{l}P^t(j,l)X(l) - \mustat(l)X(l) } 
\leq \normo{ P^t(j,\cdot) - \mustat  } \normi{X} \mper \]
Therefore, taking $i$ to be a vertex such that $\Abs{X(i)} = \normi{X}$, we get 
\[ \normo{ P^t(j,\cdot) - \mustat  } \geq \alpha_2^t = (1 - \lambda_2)^t \mper  \] 
This proves the theorem.

\end{proof}

\end{document}